\documentclass[11pt]{article}

\usepackage[utf8]{inputenc}
\usepackage[english]{babel}
\usepackage{graphicx}         
\usepackage{braket}
\usepackage{amsmath}
\usepackage{amssymb}
\usepackage{amsthm}
\usepackage{mathtools}
\usepackage{bbm}		
\usepackage{mathrsfs}
\usepackage{array}		
\usepackage{color}
\usepackage{enumitem}
\usepackage{subeqnarray}	
\usepackage{wrapfig}
\usepackage{url}

\allowdisplaybreaks

\numberwithin{equation}{section}

\theoremstyle{plain}\newtheorem{definition}{Definition}[section]
\newtheorem{lem}[definition]{Lemma}
\newtheorem{proposition}[definition]{Proposition}
\newtheorem{cor}[definition]{Corollary}
\theoremstyle{remark}\newtheorem{remark}[definition]{Remark}

\theoremstyle{plain}\newtheorem{claim}{Claim}

\theoremstyle{plain}\newtheorem{assumption}{Assumption}

\theoremstyle{plain}\newtheorem{theorem}{Theorem}

\newcommand{\lemit}[1]{\begin{enumerate}[label={(\alph*)}, ref={\thelem\alph*}]{#1}\end{enumerate}}
\newcommand{\remit}[1]{\begin{enumerate}[label={(\alph*)}, ref={\theremark\alph*}]{#1}\end{enumerate}}

\newcommand{\D}{\mathcal{D}}
\newcommand{\R}{\mathbb{R}}
\newcommand{\C}{\mathbb{C}}
\newcommand{\N}{\mathbb{N}}
\newcommand{\fH}{\mathfrak{H}}
\newcommand{\Fock}{\mathcal{F}}
\newcommand{\Number}{\mathcal{N}}
\newcommand{\vac}{|\Omega\rangle}
\newcommand{\id}{\mathbbm{1}}
\newcommand{\cJ}{\mathcal{J}}
\newcommand{\cS}{\mathcal{S}}
\newcommand{\cL}{\mathcal{L}}
\newcommand{\fV}{\mathfrak{V}}
\newcommand{\cU}{\mathcal{U}}
\newcommand{\cQ}{\mathcal{Q}}

\newcommand{\fC}{\mathfrak{C}}
\newcommand{\fE}{\mathfrak{E}}
\newcommand{\HS}{\mathrm{HS}}
\newcommand{\op}{\mathrm{op}}
\let\textl\l
\renewcommand{\l}{\ell}
\renewcommand{\i}{\mathrm{i}}
\newcommand{\e}{\mathrm{e}}
\newcommand{\hc}{\mathrm{h.c.}}

\newcommand{\sym}{\mathrm{sym}}
\newcommand{\wf}{\mathrm{wf}}

\newcommand{\Tr}{\mathrm{Tr}}

\newcommand{\bPhi}{{\boldsymbol{\phi}}}

\newcommand{\bj}{\boldsymbol{j}}
\newcommand{\bk}{\boldsymbol{k}}

\newcommand{\bm}{\boldsymbol{m}}

\renewcommand{\hat}[1]{\widehat{#1}}
\renewcommand{\tilde}[1]{\widetilde{#1}}

\newcommand{\lr}[1]{\left\langle #1 \right\rangle}
\newcommand{\norm}[1]{\lVert#1\rVert}
\newcommand{\onorm}[1]{\lVert#1\rVert_\mathrm{op}}

\renewcommand{\d}{\mathop{}\!\mathrm{d}}
\newcommand{\dx}{\d x}
\newcommand{\dy}{\d y}
\newcommand{\dz}{\d z}

\newcommand{\xk}{x^{(k)}}

\newcommand{\Ubar}{\overline{U}}
\newcommand{\Vbar}{\overline{V}}

\newcommand{\ad}{a^\dagger}
\newcommand{\Ad}{A^\dagger}

\newcommand\mydots{,\makebox[1em][c]{.\hfil.\hfil.},}
\newcommand\mycdots{\makebox[1em][c]{$\cdot$\hfil$\cdot$\hfil$\cdot$}}

\makeatletter
\DeclareFontFamily{OMX}{MnSymbolE}{}
\DeclareSymbolFont{MnLargeSymbols}{OMX}{MnSymbolE}{m}{n}
\SetSymbolFont{MnLargeSymbols}{bold}{OMX}{MnSymbolE}{b}{n}
\DeclareFontShape{OMX}{MnSymbolE}{m}{n}{
    <-6>  MnSymbolE5
   <6-7>  MnSymbolE6
   <7-8>  MnSymbolE7
   <8-9>  MnSymbolE8
   <9-10> MnSymbolE9
  <10-12> MnSymbolE10
  <12->   MnSymbolE12
}{}
\DeclareFontShape{OMX}{MnSymbolE}{b}{n}{
    <-6>  MnSymbolE-Bold5
   <6-7>  MnSymbolE-Bold6
   <7-8>  MnSymbolE-Bold7
   <8-9>  MnSymbolE-Bold8
   <9-10> MnSymbolE-Bold9
  <10-12> MnSymbolE-Bold10
  <12->   MnSymbolE-Bold12
}{}
\let\llangle\@undefined
\let\rrangle\@undefined
\DeclareMathDelimiter{\llangle}{\mathopen}
                     {MnLargeSymbols}{'164}{MnLargeSymbols}{'164}
\DeclareMathDelimiter{\rrangle}{\mathclose}
                     {MnLargeSymbols}{'171}{MnLargeSymbols}{'171}
\makeatother

\newcommand\smallO[1]{
        \mathchoice
            {
                \ensuremath{\mathop{}\mathopen{}{\scriptstyle\mathcal{O}}\mathopen{}\left(#1\right)}
            }
            {
                \ensuremath{\mathop{}\mathopen{}{\scriptstyle\mathcal{O}}\mathopen{}\left(#1\right)}
            }
            {
                \ensuremath{\mathop{}\mathopen{}{\scriptscriptstyle\mathcal{O}}\mathopen{}\left(#1\right)}
            }
            { 
                \ensuremath{\mathop{}\mathopen{}{o}\mathopen{}\left(#1\right)}
            }
    }

\newcommand{\fHN}{{\fH^N}}
\newcommand{\fHp}{\fH_\perp}

\newcommand{\HN}{H_N}
\newcommand{\Vext}{V^\mathrm{ext}}
\newcommand{\lN}{\lambda_N}

\newcommand{\PsiN}{\Psi_N}
\newcommand{\PsiNn}{\PsiN^{(n)}}

\newcommand{\EN}{\mathscr{E}_N}
\newcommand{\ENn}{\mathscr{E}_N^{(n)}}
\newcommand{\ENz}{\mathscr{E}_N^{(0)}}

\newcommand{\ENnu}{\mathscr{E}_N^{(\nu)}}

\newcommand{\fENn}{\fE_N^{(n)}}
\newcommand{\tfENnu}{\tilde{\fE}_N^{(\nu)}}

\newcommand{\fEn}{\fE^{(n)}}
\newcommand{\fEmo}{\fE^{(-1)}}
\newcommand{\fEzn}{\fEn_0}
\newcommand{\tfEnu}{\tilde{\fE}^{(\nu)}}

\newcommand{\PNn}{P^{(n)}_N}

\newcommand{\dnu}{\delta^{(\nu)}_N}
\newcommand{\dzn}{\delta^{(n)}_0}

\newcommand{\hH}{h}
\newcommand{\eH}{e_\mathrm{H}}
\newcommand{\mH}{\mu_\mathrm{H}}
\newcommand{\cEH}{\mathcal{E}_\mathrm{H}}
\newcommand{\gapH}{g_\mathrm{H}}
\renewcommand{\DH}{\D_\mathrm{H}}

\newcommand{\Om}{O^{(m)}}
\newcommand{\Omj}{\Om_{j_1\mydots j_m}}
\newcommand{\Omom}{\Om_{1\mydots m}}

\newcommand{\UNp}{U_{N,\varphi}}

\newcommand{\FNp}{{\Fock_\perp^{\leq N}}}
\newcommand{\FgNp}{{\Fock_\perp^{>N}}}

\newcommand{\ChiN}{\Chi_{\leq N}}
\newcommand{\ChiNn}{\Chi^{(n)}_{\leq N}}

\newcommand{\FockHN}{\FockH_{\leq N}}

\newcommand{\Fp}{{\Fock_\perp}}

\newcommand{\FockH}{\mathbb{H}}
\newcommand{\FockHz}{\FockH_0}
\newcommand{\FockHo}{\FockH_1}
\newcommand{\FockHt}{\FockH_2}

\newcommand{\FockHj}{\FockH_j}
\newcommand{\FockHnu}{\FockH_\nu}
\newcommand{\FockHjo}{\FockH_{j_1}}
\newcommand{\FockHjt}{\FockH_{j_2}}

\newcommand{\FockHjnu}{\FockH_{j_\nu}}

\newcommand{\tFockH}{\FockH^<}
\newcommand{\FockHminus}{\FockH^<}
\newcommand{\FockHplus}{\FockH^>}

\newcommand{\FockR}{\mathbb{R}}
\newcommand{\FockRz}{\FockR_0}
\newcommand{\FockRza}{\FockRz^{(1)}}
\newcommand{\FockRoa}{\FockRo^{(1)}}
\newcommand{\FockRo}{\FockR_1}
\newcommand{\FockRj}{\FockR_j}
\newcommand{\FockRnu}{\FockR_\nu}
\newcommand{\tFockRt}{\tilde{\FockR}^{(2)}}
\newcommand{\tFockRth}{\tilde{\FockR}^{(3)}}
\newcommand{\tFockRtha}{\tilde{\FockR}^{(3)}_a}
\newcommand{\tFockRta}{\tilde{\FockR}^{(2)}_a}

\newcommand{\Chi}{{\boldsymbol{\chi}}}
\newcommand{\Chin}{\Chi^{(n)}}
\newcommand{\Chinl}{\Chi^{(n,\l)}}
\newcommand{\Chinm}{\Chi^{(n,m)}}
\newcommand{\Chiz}{\Chi_0}
\newcommand{\Chizn}{\Chiz^{(n)}}
\newcommand{\Chizz}{\Chiz^{(0)}}

\newcommand{\Chiznm}{\Chiz^{(n,m)}}
\newcommand{\Chiznmu}{\Chiz^{(n,\mu)}}

\renewcommand{\P}{\mathbb{P}}
\newcommand{\Pn}{\P^{(n)}}
\newcommand{\Q}{\mathbb{Q}}
\newcommand{\Qn}{\Q^{(n)}}

\newcommand{\Pz}{\P_0}
\newcommand{\Pzn}{\Pz^{(n)}}
\newcommand{\Qz}{\Q_0}
\newcommand{\Qzn}{\Qz^{(n)}}

\newcommand{\Ez}{E_0}
\newcommand{\Ezo}{\Ez^{(1)}}
\newcommand{\Ezz}{\Ez^{(0)}}
\newcommand{\Ezn}{\Ez^{(n)}}

\newcommand{\En}{E^{(n)}}
\newcommand{\Ezero}{E^{(0)}}
\newcommand{\Eo}{E^{(1)}}
\newcommand{\Emo}{E^{(-1)}}
\newcommand{\Enu}{E^{(\nu)}}

\newcommand{\In}{\iota^{(n)}}

\newcommand{\gz}{g_0}
\newcommand{\gzn}{\gz^{(n)}}
\newcommand{\gzz}{g_0^{(0)}}

\newcommand{\fgn}{\mathfrak{g}^{(n)}}

\newcommand{\Np}{\Number_\perp}

\newcommand{\Ko}{K_1}
\newcommand{\Kt}{K_2}
\newcommand{\Kth}{K_3}
\newcommand{\Kf}{K_4}

\newcommand{\boldKz}{\mathbb{K}_0}
\newcommand{\boldKo}{\mathbb{K}_1}
\newcommand{\boldKt}{\mathbb{K}_2}
\newcommand{\boldKth}{\mathbb{K}_3}
\newcommand{\boldKf}{\mathbb{K}_4}
\newcommand{\boldKtbar}{\mathbb{K}_2^*}
\newcommand{\boldKthbar}{\mathbb{K}_3^*}

\renewcommand{\dj}{d^{(j)}}

\newcommand{\UP}{\cU_P}

\newcommand{\Pnl}{\Pn_\l}

\newcommand{\ResHz}{\frac{1}{z-\FockH}}
\newcommand{\RestHz}{\frac{1}{z-\tFockH}}
\newcommand{\ResHzz}{\frac{1}{z-\FockHz}}

\newcommand{\FockT}{\mathbb{T}}
\newcommand{\FockS}{\mathbb{S}}
\newcommand{\FockE}{\mathbb{E}}

\newcommand{\FockO}{\mathbb{O}}
\newcommand{\FockOn}{\FockO^{(n)}}
\newcommand{\FockOnk}{\FockOn_k}
\newcommand{\FockOnz}{\FockOn_0}
\newcommand{\FockOno}{\FockOn_1}

\newcommand{\tFockOn}{\tilde{\FockO}^{(n)}}
\newcommand{\tFockOnz}{\tilde{\FockO}^{(n)}_0}
\newcommand{\tFockOno}{\tilde{\FockO}^{(n)}_1}
\newcommand{\tFockOnt}{\tilde{\FockO}^{(n)}_2}

\newcommand{\FockCn}{\mathbb{C}^{(n)}}
\newcommand{\FockB}{\mathbb{B}}
\newcommand{\FockBPn}{\FockB_P^{(n)}}
\newcommand{\FockBQn}{\FockB_Q^{(n)}}
\newcommand{\FockAnjnu}{\FockA^{(n)}_{\bj}}
\newcommand{\FockBnjm}{\FockB^{(n)}_{\bj}}
\newcommand{\tFockAnkjnu}{\tilde{\FockA}^{(n)}_{k,\,\bj}}

\newcommand{\FockI}{\mathbb{I}}
\newcommand{\FockIn}{\FockI^{(n)}}
\newcommand{\FockInk}{\FockIn_k}
\newcommand{\ResInz}{\frac{\FockIn}{z-\FockHz}}
\newcommand{\ResInzbar}{\frac{\FockIn}{\overline{z}-\FockHz}}

\newcommand{\gan}{\gamma^{(n)}}
\newcommand{\goint}{\oint_{\gan}}

\newcommand{\Am}{A^{(m)}}
\newcommand{\Amj}{\Am_{j_1\mydots j_m}}
\newcommand{\Amom}{\Am_{1\mydots m}}
\newcommand{\cAm}{\mathcal{A}^{(m)}_N}

\newcommand{\expAm}{\langle A\rangle^{(m)}}
\newcommand{\expAo}{\langle A\rangle^{(1)}}

\newcommand{\FockA}{\mathbb{A}}
\newcommand{\FockAmred}{\FockA^{(m)}_\mathrm{red}}
\newcommand{\FockAm}{\FockA^{(m)}_N}

\newcommand{\BogU}{\mathbb{U}_\BogV}
\newcommand{\BogV}{\mathcal{V}}

\newcommand{\BogUz}{\mathbb{U}_\BogVz}
\newcommand{\BogVz}{{\mathcal{V}_0}}

\newcommand{\CVz}{C_\BogVz}

\title{Asymptotic expansion of low-energy excitations for weakly interacting bosons}
\author{Lea Boßmann\thanks{Institute of Science and Technology Austria, Am Campus 1, 3400 Klosterneuburg, Austria. \texttt{lea.bossmann@ist.ac.at}},\;
Sören Petrat\thanks{Department of Mathematics and Logistics, Jacobs University Bremen, Campus Ring 1, 28759 Bremen, Germany; and University of Bremen, Department 3 – Mathematics, Bibliothekstr. 5, 28359 Bremen, Germany. \texttt{s.petrat@jacobs-university.de}},\;
and Robert Seiringer\thanks{Institute of Science and Technology Austria, Am Campus 1, 3400 Klosterneuburg, Austria. \texttt{robert.seiringer@ist.ac.at}}}

\date{\today}

\begin{document}
\maketitle
 
\begin{abstract}
We consider a system of $N$ bosons in the  mean-field scaling regime for a class of interactions including the repulsive Coulomb potential. We derive an asymptotic expansion of the low-energy eigenstates  and the corresponding energies, which provides corrections to Bogoliubov theory to any order in $1/N$.
\end{abstract}

\section{Introduction}
We consider a system of $N$ interacting bosons in $\R^d$, $d\geq1$, which are described by the $N$-body Hamiltonian
\begin{equation}\label{HN}
H_N=\sum\limits_{j=1}^N\left(-\Delta_j+\Vext(x_j)\right)+\lN\sum\limits_{1\leq i<j\leq N}v(x_i-x_j)
\end{equation}
with coupling parameter
$$ \lN:=\frac{1}{N-1}\,,$$
corresponding to a mean-field (or Hartree) regime of weak and long-range interactions. The Hamiltonian $\HN$ acts on the Hilbert space of square integrable, permutation symmetric functions on $\R^{dN}$,
\begin{equation*}
\fH^N_\sym:=\bigotimes\limits_\sym^N\fH\,,\qquad \fH:=L^2(\R^d)\,.
\end{equation*}
Our assumptions on the interaction $v$ include the repulsive Coulomb potential ($d=3$), and our conditions on the external trap $\Vext$ are satisfied, e.g., by  harmonic potentials.
We study the spectrum\footnote{We follow the convention to count eigenvalues without multiplicity.} of $\HN$, 
$$\EN^{(0)}<\EN^{(1)}<\dots<\ENn<\dots\,,$$
for excitation energies of order one above the ground state, as well as the corresponding eigenfunctions.
Our main result is an asymptotic expansion of the eigenvalues of $\HN$, which, in the case where the degeneracy does not change in the limit $N\to\infty$, reads
\begin{equation}\label{intro:as:expansion}
\ENn=N\eH + \Ezn + \lN \En_1 + \lN^2\En_2+ \lN^3\En_3+\dots\,,
\end{equation}  
where the $N$-dependence is exclusively in the prefactors $N$ and $\lN$. More precisely, we construct  an asymptotic expansion of the spectral projectors of $\HN$, which implies \eqref{intro:as:expansion}. For eigenvalues whose degeneracy increases in the limit $N\to\infty$, we obtain a comparable result for the sum of those eigenvalues that become degenerate in the limit.
\smallskip

Let us explain the different contributions in \eqref{intro:as:expansion}.
It is well known (see, e.g., \cite{seiringer2011, grech2013, lewin2014, lewin2015_2,LSSY}) that, for  any fixed $n\in\N_0$, the eigenstates $\PsiNn$ of $\HN$ associated with $\ENn$ exhibit Bose--Einstein condensation (BEC) in the minimizer $\varphi$ of the Hartree functional. As \eqref{HN} describes a mean-field regime, the leading order in \eqref{intro:as:expansion} is given by
\begin{equation}\label{intro:leading:order}
\ENn= \lr{\PsiNn,\HN\PsiNn}=\lr{\varphi^{\otimes N},\HN\varphi^{\otimes N}}+\mathcal{O}(1)\,, 
\end{equation}
with
\begin{equation}
\lr{\varphi^{\otimes N},\HN\varphi^{\otimes N}}
=N\lr{\varphi,\left(-\Delta+\Vext+\tfrac12v*\varphi^2\right)\varphi} =:N\eH\,.
\end{equation}
For corresponding results in more singular scaling limits, see  \cite{lieb1999,lieb2002,lieb2006,nam2016_2,boccato2017,boccato2018_2,nam2020,adhkari2020} and \cite{lieb1998,lieb_solovej2001,lieb_solovej2004,solovej2006,erdoes_GS_2008,giuliani2009, yau2009,brietzke2020,brietzke2020_2,fournais2019}.
 
The error in  \eqref{intro:leading:order} is caused by $\mathcal{O}(1)$ particles  which are excited from the condensate.
To compute their energy, one decomposes $\PsiNn$ into contributions from  condensate and  excitations, as was first proposed in \cite{lewin2015_2}. The excitations form a vector in a truncated Fock space over the orthogonal complement of $\varphi$, and the relation between $\PsiNn$ and the corresponding excitation vector is given by a unitary map
\begin{equation}\label{intro:UNp}
\UNp:\fHN\to\FNp:=\bigoplus\limits_{k=0}^N\bigotimes\limits_\sym^k\left\{\varphi\right\}^\perp\,,\qquad \PsiN\mapsto \UNp\PsiNn=:\ChiNn\,,
\end{equation}
with the usual notation $\left\{\varphi\right\}^\perp :=\left\{\phi\in \fH: \lr{\phi,\varphi}_{\fH}=0\right\}$. Hence, 
\begin{equation}
\ENn\;=\;N\eH+\lr{\UNp\PsiNn,\FockHN\UNp\PsiNn}_{\FNp}\,,
\end{equation}
where
\begin{equation}
\FockHN\;:=\;\UNp\left(\HN-N\eH\right)\UNp^*:\FNp\to\FNp
\end{equation}
describes the energy due to excitations from the condensate.

By construction, the excitation Hamiltonian $\FockHN$ is explicitly $N$-dependent. 
To extract the contributions to the energy to each order in $\lN$, we extend $\FockHN$ trivially to an operator $\FockH$ acting on the  full excitation Fock space $\Fp$ and expand it formally as
\begin{equation}\label{intro:taylor}
\FockH \;=\;\FockHz+\sum\limits_{j\geq 1}\lN^\frac{j}{2}\FockHj\,.
\end{equation}
The coefficients $\FockHj$ are $N$-independent operators on $\Fp$, which are explicitly given in terms of $\varphi$ and $v$ (see  Definitions \ref{def:FockHz} and \ref{def:Pna}). In particular, $\FockHj$ contains an even number of creation/annihilation operators for $j$ even, and an odd number for $j$ odd.

The leading order term $\FockHz$ is the well-known Bogoliubov Hamiltonian, which was first proposed by Bogoliubov in 1947 \cite{bogoliubov1947}. 
It is quadratic in the number of creation/annihilation operators and can be diagonalized by Bogoliubov transformations.
The spectrum of $\FockHz$ gives the $\mathcal{O}(1)$ contribution in \eqref{intro:as:expansion}, i.e., for any $\nu\in\N_0$, there exists an eigenvalue $\Ezn$ of $\FockHz$ such that
\begin{equation}\label{intro:Bogoliubov}
\lim\limits_{N\to\infty}\left(\ENnu-N\eH\right)\;=\;\Ezn\,,
\end{equation}
with
$$\Ezz<\Ezo<\dots<\Ezn<\dots$$ 
the eigenvalues of $\FockHz$. For bounded interactions $v$, this was shown in \cite{seiringer2011} for the homogeneous setting, and in \cite{grech2013} for the inhomogeneous case. Lewin, Nam, Serfaty and Solovej \cite{lewin2015_2}  proved \eqref{intro:Bogoliubov} for a larger class of models, including a class of unbounded interaction potentials as well as a variety of one-particle operators.  Moreover, related results on the torus were obtained  in \cite{mitrouskas_PhD,nam2018}.
All error estimates proven in \cite{seiringer2011,grech2013,lewin2015_2,mitrouskas_PhD} are at best of the order $\mathcal{O}(N^{-1/2})$.  We refer to \cite{derezinski2014,nam2014,boccato2017_2,boccato2018} for similar results in more singular scaling limits.\\

In this paper, we derive the remaining terms in the expansion \eqref{intro:as:expansion}.
To keep the notation simple, we restrict---for the remainder of this introduction---to the (non-degenerate) ground state.
Formally, the coefficients in \eqref{intro:as:expansion} can be determined by Rayleigh--Schrödinger perturbation theory in the small parameter $\lN^{1/2}$.
Let us denote by $\Chiz$ the (non-degenerate) normalized ground state of $\FockHz$, and by $\Pz$ and $\Qz$ the corresponding orthogonal projections on $\Fp$, i.e.,
\begin{equation}
\FockHz\Chiz=\Ezz\Chiz\,,\qquad \Pz=|\Chiz\rangle\langle\Chiz|\,,\qquad  \Qz=\id-\Pz\,.
\end{equation}
By \eqref{intro:taylor}, the first order perturbation of $\FockHz$ is
\begin{equation}
\FockH=\FockHz+\lN^\frac12\FockHo+\mathcal{O}(\lN)\,,
\end{equation}
hence first order perturbation theory yields (see, e.g., \cite[Chapter 5]{sakurai})
\begin{equation}\label{intro:first:order:PT}
\ENz-N\eH\; =\; \Ezz+\lN^\frac12\lr{\Chiz,\FockHo\Chiz}_\Fp + \mathcal{O}(\lN) 
\;=\; \Ezz+\mathcal{O}(\lN)\,.
\end{equation}
Here, the $\mathcal{O}(\lN^{1/2})$ contribution vanishes by Wick's rule because $\FockHo$ contains an odd number of creation/annihilation operators and $\Chiz$ is quasi-free. 
For the next order, second order perturbation theory for the Hamiltonian
\begin{equation}
\FockH=\FockHz+\lN^\frac12\FockHo+\lN\FockHt+\mathcal{O}\big(\lN^{3/2}\big)
\end{equation}
yields
\begin{equation}\label{intro:second:order:PT}
\ENz-N\eH=\Ezz+\lN\bigg\langle\Chiz,\bigg(\FockHt+\FockHo\frac{\Qz}{\Ezz-\FockHz}\FockHo\bigg)\Chiz\bigg\rangle_\Fp + \mathcal{O}(\lN^2)\,,
\end{equation}
and the higher orders are constructed similarly. In particular, all terms in the expansion corresponding to half-integer powers of $\lN$ vanish. 

In our main result, we make this formal argument rigorous by proving an asymptotic expansion for the ground state projector $\P$ of $\FockH$.
Recall that
\begin{equation}\label{intro:functional:calculus}
\P=\frac{1}{2\pi\i}\oint_\gamma\ResHz\dz\,,\qquad \Pz=\frac{1}{2\pi\i}\oint_\gamma\ResHzz\dz\,,
\end{equation}
for some closed contour $\gamma$ which encloses both $\ENz-N\eH$ and $\Ezz$ and leaves the remaining spectra of $\FockH$ and $\FockHz$ outside. The existence of such a contour with length of order one is, for sufficiently large $N$, guaranteed by \eqref{intro:Bogoliubov}.
Using  \eqref{intro:taylor}, we expand the resolvent of $\FockH$ around the resolvent of $\FockHz$, which results in an expansion of $\P$, and the trace against $\FockH$ recovers  \eqref{intro:first:order:PT} and \eqref{intro:second:order:PT} (see Theorem~\ref{thm:energy}). 
Finally, we show that the error is sub-leading with respect to the order of the approximation.

In fact, we prove a more general statement, which can be understood as asymptotic expansion of the ground state of $\HN$:
For any operator $A^{(m)}$ on $\fH^m$ that is relatively bounded with respect to $\sum_{j=1}^m(-\Delta_j+\Vext(x_j))$, it holds that
\begin{equation}
\Tr_\fHN\cAm P_N \;=\; 
\Tr_\Fp\FockAm \Pz
+\sum\limits_{\l=1}^a \lN^\frac{\l}{2}\Tr_\Fp\FockAm \P_\l+\mathcal{O}\big(\lN^{\frac{a+2}{2}}\big) 
\,,
\end{equation}
where $P_N$ denotes the projector onto the ground state of $\HN$, $\cAm$ is the symmetrized version of $\Am$ on $\fHN$, $\FockAm$ denotes the conjugation of $\cAm$ with $\UNp$, and $\P_\l$ is the $\l$'th order in the expansion of the projector $\P$. The full statement, which extends to excited states with energies of order one above the ground state, is given in Theorem \ref{thm:exp:P}.

Our analysis is restricted to the mean-field regime. It is an open question whether a similar statement holds true for interaction potentials that converge to a delta distribution as $N\to\infty$. \medskip

In the physics literature, higher order corrections to the Lee--Huang--Yang formula for the ground state energy of a low-density Bose gas with short-range interactions have been studied already in the 1950's in \cite{brueckner1957_1,brueckner1957_2,beliaev1958_1,beliaev1958_2,wu1959}, and a series expansion for the ground state energy was conjectured in \cite{sawada1959,hugenholtz1959}. We refer to \cite{braaten1999,braaten2001,weiss2004} for more recent contributions.
However, to the best of our knowledge,  the rigorous derivation of higher order corrections to the Bogoliubov energy in the mean-field scaling has not been studied before.
Other approaches to perturbations around Bogoliubov theory are based on the ideas of renormalization group and constructive field theory, which is very different from our rather direct approach. We refer to \cite{cenatiempo2014} for recent results and a review of the literature, which mostly treats more singular scalings than the mean-field regime. 

Another approach was proposed by Pizzo in \cite{pizzo2015,pizzo2015_2,pizzo2015_3}, where he considers a Bose gas on a torus in the mean-field regime. He constructs an expansion for the ground state and a fixed-point equation for the ground state energy, first for a simpler three-modes Bogoliubov Hamiltonian \cite{pizzo2015}, and subsequently, building on these results, for a Bogoliubov Hamiltonian \cite{pizzo2015_2} and the full Hamiltonian \cite{pizzo2015_3}. The main result is norm convergence of the expansion to the ground state to arbitrary precision. This expansion is based on a multi-scale
analysis in the number of excitations around a product state using Feshbach maps. In contrast to our work, this is done in the $N$-particle space, whereas we make use of the $N$-dependent unitary map $\UNp$ to work in the excitation Fock space $\Fp$.

Finally, we remark that our work is inspired by \cite{QF}, where an analogous expansion of the dynamics generated by $\HN$ was constructed. Related results for the mean-field dynamics in Fock space have been obtained in \cite{ginibre1980,ginibre1980_2}, and different approaches characterizing the dynamics to any order in $1/N$ were discussed in \cite{paul2019,corr}. Let us also note that there are many recent results on the derivation of the Bogoliubov dynamics in the mean-field  regime \cite{grillakis2010,grillakis2011,lewin2015,mitrouskas2016} as well as in more singular scaling limits \cite{grillakis2013,boccato2015,nam2015,grillakis2017,kuz2017,nam2017,chong2016,
brennecke2017_2,petrat2017}.

\subsection*{Notation}
\begin{itemize}
\item 
We denote by $\fC$ an expression which may depend on constants fixed by the model, i.e., constants whose values depend on $\hH$ and $\FockHz$, such as norms of the Hartree minimizer $\varphi$, the gap $\gapH$ above the ground state of $\hH$, and norms $\onorm{U_0}$ (the operator norm), $\norm{V_0}_\HS$ (the Hilbert--Schmidt norm) of the Bogoliubov transformation diagonalizing $\FockHz$.
The notation $\fC(n)$ indicates that the constant may also depend on the number $n$ of the corresponding eigenvalue of $\FockHz$, such as $|\Ezn|$, its degeneracy $\dzn$, and the spectral gap above it.
Finally, $\fC(n,a)$ implies the dependence on an additional parameter $a$.
Constants may vary from line to line.
\item Eigenvalues are always counted without multiplicity, i.e., the (discrete) spectrum of an operator $T$ is denoted as $t^{(0)}< t^{(1)}< t^{(2)}<\dots$, where each eigenvalue $t^{(j)}$ has some finite multiplicity $\delta^{(j)}\geq1$.
\item We denote by $\bj:=(j_1\mydots j_{n})$ a multi-index  and  define $|\bj|:=j_1+\dots+j_n$. Moreover, we abbreviate
\begin{equation}
\xk:=(x_1\mydots x_k)\,,\qquad \d\xk:=\dx_1\mycdots\dx_k
\end{equation} 
for $k\geq 1$ and $x_j\in{\R^d}$.
\end{itemize}

\section{Preliminaries}\label{sec:prelim}
\subsection{Assumptions}\label{sec:assumptions}

We make the following assumptions on the external potential $\Vext$ and the interaction $v$:

\begin{assumption}\label{ass:V}
Let $\Vext:\R^d\to\R$ be measurable, locally bounded and  non-negative and let $\Vext(x)$ tend to infinity as $|x|\to  \infty$, i.e.,
\begin{equation}
\inf\limits_{|x|>R}\Vext(x)\to\infty \text{ as } R\to \infty\,.
\end{equation}
\end{assumption}

Assumption \ref{ass:V} implies that $\Vext$ must be a confining potential. It is, for example, satisfied by $\Vext(x)=\omega x^2$ for  $\omega>0$. 
Let us introduce the abbreviation
\begin{equation}\label{def:T}
 T:\fH\supset\D(T)\to\fH\,,\qquad T:=-\Delta+\Vext\,.
\end{equation}
We denote by 
$$T_j:=\id\otimes\mycdots\otimes\id\otimes T\otimes\id\otimes\mycdots\otimes\id$$
the operator acting as $T$ on the $j$th coordinate.

\begin{assumption}\label{ass:v}
Let $v:\R^d\to\R$ be measurable with $v(-x)=v(x)$ and $v\not\equiv 0$, and assume that there exists a constant $C>0$ such that, in the sense of operators on $\cQ(-\Delta)=H^1(\R^d)$,
\begin{equation}
|v|^2\leq  C\left(1-\Delta\right)\,.  \label{eqn:ass:v:2:Delta:bound}
\end{equation}
Besides, assume that $v$ is of positive type, i.e., that it has a non-negative Fourier transform.
\end{assumption}

Assumption \ref{ass:v} is clearly satisfied by any bounded potential with positive Fourier transform. Moreover, by Hardy's inequality, it is fulfilled by the repulsive Coulomb potential in $d= 3$ dimensions.

\begin{remark}
\remit{
\item 
Note that \eqref{eqn:ass:v:2:Delta:bound} implies that
\begin{equation}
2|v(x_1-x_2)|\leq 1+|v(x_1-x_2)|^2\leq \fC\left(-\Delta_1-\Delta_2+1\right)\;\leq\;\fC(T_1+T_2+1)  \label{eqn:ass:v:2:T:bound}
\end{equation}
in the sense of operators on $\cQ(T_1+T_2)\subset\fH^2$ because $\Vext\geq 0$.
In particular, 
\begin{eqnarray}\label{eqn:ass:v:v*phi^2}
\norm{v*\phi^2}_\infty &\leq& \fC\left(\norm{\nabla\phi}^2+1\right)\,,\\[5pt]
\lr{\phi\otimes\phi,|v(x-y)|^2\phi\otimes\phi}_{\fH^2}
&\leq& \fC\lr{\phi,(T+1)\phi}  
\label{eqn:HS:norm:K}
\end{eqnarray}
for any normalized $\phi\in\cQ(T)$. Moreover, $v$ being of positive type implies that
\begin{equation}\label{eqn:ass:v:positivity}
\int_{\R^{2d}}\dx\dy \,\overline{\phi(x)}v(x-y)\phi(y) \geq 0\,.
\end{equation}
\item
Assumptions \ref{ass:V} and \ref{ass:v} imply that $|v|^2\leq \varepsilon T^2 + C^2\varepsilon^{-1}+C$ for any $\varepsilon>0$, hence $\HN$ is (for each $N$) self-adjoint on its domain $\D\big(\sum_{j=1}^N T_j\big)$ by the Kato--Rellich theorem.

\item Since $\Vext$ is measurable and locally bounded and tends to infinity, it is bounded below, and we take its lower bound to be zero only for convenience.
}
\end{remark}

Next, we recall the Hartree energy functional, which is defined as
\begin{equation}\label{def:Hartree:functional}
\cEH[\phi]:=\int\limits_{\R^d}\left(|\nabla \phi(x)|^2+\Vext(x)|\phi(x)|^2\right)\dx
+\tfrac12\int\limits_{\R^{2d}}v(x-y)|\phi(x)|^2|\phi(y)|^2\dx\dy
\end{equation}
for $\phi\in\DH$ with
\begin{equation}
\DH:=\left\{\phi\in\cQ(T)\,:\, \norm{\phi}_\fH=1\right\}\subset \fH\,.
\end{equation}
Its infimum is denoted by
\begin{equation}
\eH:=\inf\limits_{\phi\in\DH} \cEH[\phi]\,.
\end{equation}
Under Assumptions \ref{ass:V} and \ref{ass:v}, $\cEH$ admits a unique, strictly positive minimizer $\varphi$, which solves the stationary Hartree equation:

\begin{lem}\label{lem:hH}
Let Assumptions \ref{ass:V} and \ref{ass:v} hold. 
\lemit{
\item\label{lem:hH:minimizer}
There exists a unique (up to a  phase) $\varphi\in\DH$ such that 
$$\cEH[\varphi]=\eH\,,$$
and we choose $\varphi$ strictly positive. The minimizer $\varphi$ solves the stationary Hartree equation,
\begin{equation}
\hH\varphi=0\,,
\end{equation}
in the sense of distributions,
where
\begin{equation}\label{h:H}
\hH:\fH\supset\D(T)\to\fH\,,\qquad \hH:T+v*\varphi^2-\mH\,
\end{equation}
with Lagrange multiplier $\mH\in\R$ given by
\begin{equation}
\mH:=\lr{\varphi,\left(T+v*\varphi^2\right)\varphi}\,.
\end{equation} 
\item \label{lem:hH:spectrum}
The operator $\hH$ is self-adjoint on its domain $\D(T)$ and its spectrum is purely discrete.
The minimizer $\varphi$ of $\cEH$ is the unique ground state $\varphi$ of $\hH$, and there exists a complete set of normalized eigenfunctions $\{\varphi_j\}_{j\geq 0}$ for $\hH$. 
Spectrum and eigenstates of $\hH$ are denoted as
\begin{equation}\label{eqn:spectrum:hH}
\hH\varphi_j=\varepsilon^{(j)}\varphi_j\,,
\qquad 0=\varepsilon^{(0)}<\varepsilon^{(1)}<\dots\,,\qquad
\varphi_0:=\varphi\,.
\end{equation} 
In particular, the spectral gap $\gapH$ above the ground state of $\hH$ is positive,
\begin{equation}\label{eqn:gap:hH}
\gapH:=\varepsilon^{(1)}-\varepsilon^{(0)}=\varepsilon^{(1)}>0\,.
\end{equation}
\item \label{lem:hH:K}
Define $K:\fH\to\fH$ as the  operator with kernel
\begin{equation}\label{def:K:kernel}
 K(x;y):=v(x-y)\varphi(x)\varphi(y)\,.
\end{equation}
Then $K$ is positive and  Hilbert--Schmidt. Moreover,
\begin{equation}
\mathcal{A}:=\begin{pmatrix} \hH+qKq & qKq \\[2pt] qKq &\hH+qKq \end{pmatrix} \geq \gapH>0\;\text{ on }\;\fHp\oplus\fHp
\end{equation}
for $\fHp:=\left\{\varphi\right\}^\perp$ and where $q$ denotes the orthogonal projection onto $\fHp$, i.e., 
\begin{equation}\label{p_and_q}
p:=|\varphi\rangle\langle\varphi|\,,\qquad q:=\id_{\fH}-p\,.
\end{equation}
}
\end{lem}

\begin{proof}
For part (a), note first that $\cEH\geq 0$ on $\DH$, hence there exists a sequence $\{\phi_n\}_n\subset\DH$ such that $\cEH[\phi_n]\to\eH$.
Moreover,  $\lr{\phi_n,T\phi_n}\leq C$  because $D(|\phi_n|^2,|\phi_n|^2)\geq 0$ by \eqref{eqn:ass:v:positivity}, where  $D(f,g):=\frac12\int_{\R^{2d}}\dx\dy \overline{f(x)}v(x-y)g(y)$.
Since $T$ has a compact resolvent by Assumption \ref{ass:V},  $\D_C:=\{\psi\in\cQ(T):\norm{\psi}\leq1\,,\lr{\psi,T\psi}\leq C\}$ is compact  \cite[Theorems XIII.16 and XIII.64]{rs4} and there exists a subsequence such that $\phi_n\to\phi\in\D_C$ strongly in $\fH$.
For $\varrho:=|\phi|^2$ and $\varrho_n:=|\phi_n|^2$, $\norm{\varrho*v}_\infty\leq C$ by \eqref{eqn:ass:v:v*phi^2} and $\int\rho_n\to\int\rho$, hence
\begin{equation}
\lim\limits_{n\to\infty} D(\varrho_n,\varrho_n)\geq 2\lim\limits_{n\to\infty}D(\varrho_n-\varrho,\varrho)+D(\varrho,\varrho)=D(\varrho,\varrho)\,.
\end{equation}
Since $\D_C$ is weakly compact in both $H^1(\R^d)$ and the $L^2$-space with norm $\norm{\psi}_V^2:=\int\Vext|\psi|^2$, we find, passing again to a subsequence, that $\liminf_{n\to\infty}\lr{\phi_n,T\phi_n}\geq \lr{\phi,T\phi}$ by weak lower semi-continuity of both norms. With this, part (a) can be shown as in \cite[Lemmas A.1--4]{lieb1999}.
We denote the unique strictly positive minimizer  by $\varphi$.

Part (b) is a consequence of \eqref{eqn:ass:v:v*phi^2} and Assumption \ref{ass:V}, by Kato--Rellich and \cite[Theorems XIII.16 and XIII.64]{rs4}. Finally, the first part of (c) is implied by \eqref{eqn:HS:norm:K}, and the second part follows since $K\geq0$ by \eqref{eqn:ass:v:positivity} and $\hH\geq\gapH$ on $\fHp$ by part~(b).
\end{proof}

In summary, Assumptions \ref{ass:V} and \ref{ass:v} provide all necessary properties of the effective one-body operator $\hH$, in particular the existence of a finite spectral gap above the ground state.
In addition, we require the Hartree functional to be a valid description for the $N$-body energy as $N\to\infty$. Put differently, we assume that $N$-body states with an energy of order one above the ground state exhibit complete BEC in the Hartree minimizer $\varphi$. This is implied by the following statement:

\begin{assumption}\label{ass:cond}
Assume that there exist constants $C_1\geq0$ and $0<C_2\leq 1$, as well as a function $\varepsilon:\N\to\R_0^+$ with 
$$\lim\limits_{N\to\infty} N^{-\frac13}\varepsilon(N) \leq C_1\,,$$
such that 
\begin{equation}\label{eqn:ass:cond}
\HN-N\eH\geq C_2 \sum\limits_{j=1}^N\hH_j-\varepsilon(N)
\end{equation}
in the sense of operators on $\D(\HN)$.
\end{assumption}

We do not know how to prove \eqref{eqn:ass:cond} under our generic Assumptions \ref{ass:V} and \ref{ass:v}. However,  \eqref{eqn:ass:cond} is known to be true for the examples we have in mind:
Any bounded and positive definite interaction potential $v$ satisfies Assumption \ref{ass:cond} with optimal rate $\varepsilon(N)=\mathcal{O}(1)$ \cite[Lemma 1 and Remark 2]{grech2013}. Moreover, the repulsive three-dimensional Coulomb potential fulfils Assumption \ref{ass:cond} with $\varepsilon(N)=\mathcal{O}(N^{1/3})$ \cite[Lemma 3.1]{lewin2015_2}.

\subsection{Excitation Fock space and excitation Hamiltonian}\label{subsec:pre:Fock:space}
In this section, we review the excitation map $\UNp$ from \eqref{intro:UNp}, which was introduced in \cite{lewin2015_2} and maps an $N$-body wave function to the corresponding excitation vector.
Recall that any $\Psi\in\fH^N_\sym$ can be decomposed into condensate and excitations as
\begin{equation}
\Psi=\sum\limits_{k=0}^N{\varphi}^{\otimes (N-k)}\otimes_s\chi^{(k)}\,,
\qquad \chi^{(k)}\in \bigotimes\limits_\sym^k \fHp\,, 
\end{equation}
with $\otimes_s$ the symmetric tensor product, which is for $\psi_a\in \fH^a$ and $\psi_b\in \fH^b$ defined as 
\begin{equation}\begin{split}\label{symm:tensor:prod}
(\psi_a\otimes_s\psi_b)&(x_1\mydots x_{a+b}):=\\
& \frac{1}{\sqrt{a!\,b!\,(a+b)!}}\sum\limits_{\sigma\in \mathfrak{S}_{a+b}}\psi_a(x_{\sigma(1)}\mydots x_{\sigma(a)})\,\psi_b(x_{\sigma(a+1)}\mydots x_{\sigma(a+b)})\,,
\end{split}\end{equation}
with $\mathfrak{S}_{a+b}$ the set of all permutations of $a+b$ elements.
The sequence
\begin{equation}
\ChiN:=\big(\chi^{(k)}\big)_{k=0}^N
\end{equation}
of $k$-particle excitations forms a vector in the truncated excitation Fock space over $\fHp$,
\begin{equation}\label{Fock:space}
\FNp=\bigoplus_{k=0}^N\bigotimes_\sym^k \fHp
\;\subset \;
\Fp=\bigoplus_{k=0}^\infty\bigotimes_\sym^k \fHp\,,
\end{equation}
and vectors in $\Fp$ are denoted as
\begin{equation}
\bPhi=\big(\phi^{(0)},\phi^{(1)},\dots,\phi^{(k)}\dots\big)\,,\qquad 
\bPhi_{\leq N}=\big(\phi^{(0)},\phi^{(1)},\dots,\phi^{(N)}\big)\,.
\end{equation}
We consider the decomposition of $\Fp$ into the subspaces
\begin{equation}
\Fp=\FNp\oplus \FgNp\,,
\end{equation}
and in the following all direct sums are understood with respect to this decomposition.
The creation and annihilation operators on $\Fp$ are
\begin{equation}
(\ad(f)\bPhi)^{(k)}(x_1\mydots x_k)=\frac{1}{\sqrt{k}}\sum\limits_{j=1}^kf(x_j)\phi^{(k-1)}(x_1\mydots x_{j-1},x_{j+1}\mydots x_k)
\end{equation}
for $k\geq 1$, and
\begin{equation}
(a(f)\bPhi)^{(k)}(x_1\mydots x_k)=\sqrt{k+1}\int\d x\overline{f(x)}\phi^{(k+1)}(x_1\mydots x_k,x)
\end{equation} 
for $k\geq 0$, where $f\in\fHp$ and $\bPhi\in\Fp$. They can be expressed in terms of the operator-valued distributions $\ad_x$ and $a_x$,
\begin{equation}
\ad(f)=\int\d x f(x)\,\ad_x\,,\qquad a(f)=\int\d x\overline{f(x)}\,a_x\,,
\end{equation}
which satisfy the canonical commutation relations 
\begin{equation}
[a_x,\ad_y]=\delta(x-y)\,,\qquad [a_x,a_y]=[\ad_x,\ad_y]=0\,.
\end{equation}
We denote the second quantization in $\Fp$  of an $m$-body operator $T^{(m)}$  by $\d\Gamma_\perp$, i.e.,
\begin{eqnarray}
&&\hspace{-0.7cm}\d\Gamma_\perp(T^{(m)})\nonumber\\
&=&0\oplus\dots\oplus 0  \oplus\bigoplus\limits_{k\geq 0}\;\sum\limits_{1\leq j_1<\dots< j_m\leq m+k}T^{(m)}_{j_1\mydots j_m}\\
&=&\frac{1}{m!}\sum\limits_{\substack{i_1\mydots i_m\geq1\\j_1\mydots j_m\geq1}}\lr{\psi_{i_1}\otimes\mycdots\otimes\psi_{i_m}, T^{(m)}\psi_{j_1}\otimes\mycdots\otimes\psi_{j_m}}\nonumber\\
&&\hspace{4cm}\times\,\ad(\psi_{i_1})\mycdots\ad(\psi_{i_m})a(\psi_{j_1})\mycdots a(\psi_{j_m})\nonumber
\end{eqnarray}
for any orthonormal basis $\{\psi_j\}_{j\geq1}$ of $\fHp$.
Equivalently, 
\begin{equation}
\d\Gamma_\perp(T^{(m)})
=\d\Gamma_\perp(q^{\otimes m}\, T^{(m)} q^{\otimes m})
=\d\Gamma(q^{\otimes m}\, T^{(m)} q^{\otimes m})\,,
\end{equation}
where $\d\Gamma$ denotes the usual second quantization in the Fock space over the full space $\fH$.
Finally,  the number operator on $\Fp$ is given by
\begin{equation}
\Np:=\d\Gamma_\perp(\id)=\d\Gamma_\perp(q)\,,\qquad (\Np\bPhi)^{(k)}=k\phi^{(k)}\;\text{ for }
\bPhi\in\Fp\,.
\end{equation}
An $N$-body state $\Psi$ is mapped onto its corresponding excitation vector $\ChiN$  by
\begin{eqnarray}\label{map:U}
\UNp:\fHN \to  \FNp\;, \quad
 \Psi  \mapsto  \UNp \Psi:= \ChiN\,,
\end{eqnarray}
which is unitary and acts as
\begin{equation}\label{eqn:map:U:explicit}
\UNp \Psi=\bigoplus\limits_{j=0}^N q ^{\otimes j}\left(\frac{a({\varphi})^{N-j}}{\sqrt{(N-j)!}}\, \Psi\right)\;\text{ for }\; \Psi\in\fHN
\end{equation}
by \cite[Proposition 4.2]{lewin2015_2}. Note that the product state ${\varphi}^{\otimes N}$ is mapped to the vacuum of $\FNp$,
\begin{equation}
\UNp \,{\varphi}^{\otimes N}=(1,0,0,\dots0)=:\vac\,.
\end{equation}
For $f,g\in\fHp$, \eqref{eqn:map:U:explicit} yields the substitution rules
\begin{subequations}\label{eqn:substitution:rules}
\begin{eqnarray}
\UNp\, \ad({\varphi})a({\varphi})\UNp^*&=&N-\Np\,,\\
\UNp\, \ad(f)a({\varphi}) \UNp^*&=&\ad(f)\sqrt{N-\Np}\,,\\
\UNp\, \ad({\varphi})a(g)\UNp^*&=&\sqrt{N-\Np}a(g)\,,\\
\UNp\, \ad(f)a(g)\UNp^*&=&\ad(f)a(g)
\end{eqnarray}
\end{subequations}
as identities on $\FNp$.
As explained in the introduction, conjugating $\HN$ with $\UNp$ extracts the contribution to the energy which is due to excitations from the condensate. 
\begin{definition}\label{def:FockHN}
Define 
\begin{equation}
\FockHN:=\UNp(\HN-N\eH)\UNp^*
\end{equation}
as operator on $\FNp$.
The eigenvalues $\En$ of $\FockHN$ relate to the eigenvalues $\ENn$ of $\HN$ as
\begin{equation}\label{En}
\En=\ENn-N\eH\,,\qquad n\in\N_0\,.
\end{equation}
\end{definition}
As a consequence of the substitution rules \eqref{eqn:substitution:rules}, $\FockHN$ can be expressed  as 
\begin{eqnarray}
\FockHN
&=&  \boldKz + \left(\frac{N-\Np}{N-1}\right) \boldKo\nonumber\\
&&+\left( \boldKt\frac{\sqrt{(N-\Np)(N-\Np-1)}}{N-1}+ \frac{\sqrt{(N-\Np)(N-\Np-1)}}{N-1} \boldKtbar\right)\nonumber\\
&&+\left( \boldKth\frac{\sqrt{N-\Np}}{N-1}+ \frac{\sqrt{N-\Np}}{N-1} \boldKthbar\right) 
+\frac{1}{N-1} \boldKf\,,\label{eqn:FockHN}
\end{eqnarray} 
where we used the shorthand notation
\begin{subequations}\label{eqn:K:notation}
\begin{eqnarray}
 \boldKz&:=&\int\dx\,\ad_x  \hH_x a_x\,,\label{eqn:K:notation:0}\\
 \boldKo&:=&\int\dx_1\dx_2\, \Ko(x_1;x_2)\ad_{x_1} a_{x_2}\,,\label{eqn:K:notation:1}\\
 \boldKt&:=&\tfrac12\int\dx_1\dx_2\, \Kt(x_1,x_2)\ad_{x_1}\ad_{x_2}\,,\label{eqn:K:notation:2}\\
 \boldKth&:=&\int\dx^{(3)}\, \Kth(x_1,x_2;x_3)\ad_{x_1}\ad_{x_2}a_{x_3}\,\label{eqn:K:notation:3}\\
 \boldKf&:=&\tfrac12\int\dx^{(4)}\, \Kf(x_1,x_2;x_3,x_4)\ad_{x_1}\ad_{x_2}a_{x_3}a_{x_4}\,,\label{eqn:K:notation:4}
\end{eqnarray}
\end{subequations}
with
\begin{subequations}\label{K}
\begin{align}
\label{K:1}
& \Ko:\fHp\to \fHp, \qquad
 \Ko :=   q   K  q \,, 
 \\[7pt]
\label{K:2}
& \Kt \in \fHp\otimes \fHp,\qquad
 \Kt(x_1,x_2) := 
   (q_1q_2 K)(x_1,x_2)\,,
\\[7pt]
\label{K:3}
& \Kth: \fHp \to \fHp\otimes \fHp, \nonumber\\
&\hspace{8mm}\psi\mapsto (\Kth\psi)(x_1,x_2):=q_1q_2 W(x_1,x_2)\varphi(x_1)(q_2\psi)(x_2)
\,,  \\[3pt]
& \Kth^*: \fHp\otimes \fHp \to \fHp, \nonumber\\
&\hspace{8mm}\psi\mapsto (\Kth^*\psi)(x_1)=q_1\int\dx_2\varphi(x_2) W(x_1,x_2)(q_1q_2\psi)(x_1,x_2)
\,,  \\[5pt]
\label{K:4}
& \Kf:\fHp\otimes \fHp\to\fHp\otimes \fHp,\nonumber\\
&\hspace{8mm}\psi\mapsto (\Kf\psi)(x_1,x_2):=q_1q_2W(x_1,x_2)(q_1q_2\psi)(x_1,x_2)\,.
\end{align}
\end{subequations}
Here, $K(x_1,x_2)$ is defined as in \eqref{def:K:kernel}, $K$  is the operator with kernel $K(x_1,x_2)$, and $W$ is the multiplication operator on $\fHp\otimes\fHp$ defined by
\begin{equation}\label{eqn:W(x,y)}
W(x_1,x_2) := v(x_1-x_2) - \left(v*\varphi^2\right)(x_1) - \left(v*\varphi^2\right)(x_2) + \lr{\varphi,v*\varphi^2\varphi}.
\end{equation}
The notation is understood such that the projections $q_1,q_2$ act on the respective functions on their right. For example, the function $\Kth\psi\in\fHp\otimes\fHp$  is obtained from $\psi\in\fHp$ by taking the tensor product of $q\psi$ and $\varphi$, acting on it with the multiplication operator $W$, and finally projecting the resulting function onto the subspace $\fHp\otimes\fHp$. Note that $q\psi=\psi$ for $\psi\in\fHp$, hence the projection $q$ in front of $\psi$ is not necessary here but allows to extend $\Kth$ to a map on the full space $\fH$. An analogous observation applies to $\Ko$, $\Kth^*$ and $\Kf$.
An explicit formula for $\FockHN$ was first derived in \cite[Section 4]{lewin2015_2}, and we rewrote it in a way that is more convenient for our analysis (see Appendix \ref{appendix:Hamiltonian}).\medskip

Finally, we recall the Bogoliubov Hamiltonian $\FockHz$ and introduce some notation:
\begin{definition}\label{def:FockHz}
The Bogoliubov Hamiltonian $\FockHz$ for the model \eqref{HN} is defined as
\begin{equation}
\FockHz:=\boldKz+\boldKo+\boldKt+\boldKtbar\,,
\end{equation}
with $\mathbb{K}_j$ as defined in \eqref{eqn:K:notation}.
The eigenvalues of $\FockHz$ are denoted as
\begin{equation}
\Ezz<\Ezo<\dots<\Ezn<\dots
\end{equation}
with associated eigenspaces
\begin{equation}
\fEzn:=\left\{\bPhi\in\Fp\,:\,\FockHz\bPhi=\Ezn\bPhi\right\}\,,\qquad
\dzn:=\dim\fEzn\,.
\end{equation}
The spectral gap of $\FockHz$ above $\Ezn$ is defined as
\begin{equation}\label{eqn:gap:FockHz}
\gzn:=\Ez^{(n+1)}-\Ezn\,, \qquad n\in\N_0\,,
\end{equation}
and the projections onto $\fEzn$ and its orthogonal complement are given by
\begin{equation}
\Pzn:=\id_{\fEzn}\,,\qquad \Qzn:=\id_{\Fp}-\Pzn\,.
\end{equation}
We denote normalized elements of $\fEzn$ as $\Chizn$.
\end{definition}

\section{Results}
\subsection{Main results}

Our goal is a perturbative expansion of the spectral projectors of $\FockHN=\UNp(\HN-N\eH)\UNp^*$ around the spectral projectors of $\FockHz$. 
For our analysis, it is crucial that the low-energy eigenvalues of $\FockHN$ converge to the corresponding eigenvalues of $\FockHz$, and the same holds true (in a suitable sense) for the respective eigenstates. This was proven in \cite{seiringer2011,grech2013,lewin2015_2}, and we collect the rigorous results in Lemma \ref{lem:known:results:FockHN}.
If different eigenvalues of $\HN-N\eH$ converge to the same limiting eigenvalue of $\FockHz$ as $N\to\infty$, we consider the sum of all corresponding spectral projections of $\HN$:

\begin{definition}\label{def:PNn}
Define 
\begin{equation}\label{In}
\In:=\left\{\nu\in\N_0:\, \lim\limits_{N\to\infty}\big(\ENnu-N\eH\big)=\Ezn\right\}
\end{equation}
and
\begin{equation}
\fENn:=\bigoplus\limits_{\,\nu\in\In}\tfENnu
\end{equation} 
with
\begin{equation}\label{dnu}
\tfENnu:=\left\{\Psi\in\fH^N_\sym\,:\, \HN\Psi=\ENnu\Psi\right\}\,,\qquad
\dnu:=\dim\tfENnu\,.
\end{equation} 
The corresponding orthogonal projections are denoted as
\begin{equation}\label{PNn}
\PNn:=\id_{\fENn}\,.
\end{equation}
\end{definition}
By \cite{lewin2015_2}, the set $\In$, which collects all eigenvalues of $\HN-N\eH$ that converge to the eigenvalue $\Ezn$ of $\FockHz$, is of the form $\{\l\mydots\l+j\}$ for some $\l,j\geq 0$. Moreover,  $1\leq|\In|\leq \dzn$, where the second inequality is strict if at least one of the eigenvalues $\ENnu$ is degenerate. 
The space $\fENn$ is the direct sum of all eigenspaces of $\HN$ associated with eigenvalues with label $\nu\in\In$, hence $\sum_{\nu\in\In}\dnu=\dzn$.\medskip

We consider expectation values with respect to $\PNn$ for a natural class of $m$-body operators, namely for all operators that are relatively bounded with respect to $\sum_{j=1}^m T_j$. We use the following notation:

\begin{definition}\label{def:A}
For $m\in\N$, let $\Am$ be some  operator acting on $\fH^m$.
We denote the corresponding symmetrized operator on $\fHN$ by
\begin{equation}
\cAm:=\binom{N}{m}^{-1}\sum\limits_{1\leq j_1<\dots<j_m\leq N} \Amj\,,
\end{equation}
where $\Amj$ is  the operator acting as $\Am$ on the variables $x_{j_1}\mydots x_{j_m}$ and as identity on all other variables.
Further, we define the corresponding operator $\FockAm$ on $\Fp$ as
\begin{equation}
\FockAm:=\UNp\, \cAm \UNp^*\oplus0\,.
\end{equation}
\end{definition}

We construct an asymptotic expansion of $\PNn$, in the sense that  
$$\Tr_\fHN\Amom\PNn = \Tr_\Fp\FockAm\Pzn + \lN^\frac12 \Tr_\Fp\FockAm\Pn_1 + \lN\Tr_\Fp\FockAm\Pn_2 + \dots \,.$$
The coefficients $\Pnl$ in the expansion of the projector are defined as follows:

\begin{definition}\label{def:Pna}
Define
\begin{equation}\label{eqn:Pna}
\Pnl:=\begin{cases}
\quad \Pzn & \text{ if }\; \l=0\,,\\[5pt]
\displaystyle -\sum\limits_{\nu=1}^\l\,
\sum\limits_{\substack{\bj\in\N^\nu\\[2pt]|\bj|=\l}} \,
\sum\limits_{\substack{\bk\in\N_0^{\nu+1}\\|\bk|=\nu}} 
\FockOn_{k_1}\FockHjo\FockOn_{k_2}\FockHjt
\mycdots
\FockOn_{k_\nu} \FockHjnu\FockOn_{k_{\nu+1}} & \text{ if }\;\l\geq 1\,,
\end{cases}
\end{equation}
with $\Pzn$ as in Definition \ref{def:FockHz}. Here, we abbreviated
\begin{equation}\label{FockO}
\FockOnk:=\begin{cases}
\displaystyle\quad-\Pzn &\quad k=0\,,\\[7pt]
\displaystyle\frac{\Qzn}{\big(\Ezn-\FockHz\big)^k} &\quad k>0\,,
\end{cases}
\end{equation}
and
\begin{subequations}\label{FockHj}
\begin{eqnarray}
\FockHo&:=& \boldKth + \boldKthbar\,,\label{FockHo}\\[5pt]
\FockHt
&:=& -(\Np-1)\boldKo -\Big(\boldKt(\Np-\tfrac12)+\hc\Big)+\boldKf\,,\label{FockHt}\\[5pt]
\FockH_{2j-1}
&:=&c_{j-1}\Big(\boldKth (\Np-1)^{j-1} + \hc\Big)\,,\label{FockHt:2n-1}\\
\FockH_{2j}&:=&
\sum\limits_{\nu=0}^j d_{j,\nu}\Big(\boldKt(\Np-1)^\nu +\hc\Big)\,\label{FockHt:2n}
\end{eqnarray}
\end{subequations}
for $j\geq 2$, with $\mathbb{K}_j$ as in \eqref{eqn:K:notation}. The coefficients $c_j$ and $d_{j,\nu}$ are given as
\begin{subequations}\label{taylor:coeff}
\begin{eqnarray}\label{eqn:taylor:coeff}
c^{(\l)}_0&:=&1\,,\\
c^{(\l)}_j &:=& \frac{(\l-\frac12)(\l+\frac12)(\l+\frac32)\mycdots(\l+j-\frac32)}{j!}\,,\quad c_j\;:=\;c_j^{(0)} \quad (j\geq1),\qquad\\
\label{eqn:taylor:coeff:2}
d_{j,\nu}&:=&\sum\limits_{\l=0}^\nu c_\l^{(0)} c_{\nu-\l}^{(0)} c_{j-\nu}^{(\l)} \qquad (j\geq \nu \geq 0)\,.
\end{eqnarray}
\end{subequations}
\end{definition}

Our main result is the following:

\begin{theorem}\label{thm:exp:P}
Let Assumptions \ref{ass:V}, \ref{ass:v} and \ref{ass:cond} be satisfied
and let $a\in\N_0$.
Let $m\in\N$ and let $\Am$ be a self-adjoint operator on $\fH^m$ such that
\begin{equation}\label{eqn:thm:A:unbounded}
\norm{\Am\psi}_{\fH^m}\leq \fC \Big\|\sum_{j=1}^m (T_j+1)\psi\Big\|_{\fH^m} \qquad \text{ for }\;\psi\in \D\Big(\sum_{j=1}^mT_j\Big) \,.
\end{equation}
Then, for sufficiently large $N$, there exists a constant $\fC(n,m,a)$ such that 
\begin{equation}\label{eqn:thm:exp:P}
\left|\Tr_\fHN \cAm\PNn
-\sum\limits_{\l=0}^a \lN^\frac{\l}{2}\Tr_\Fp\FockAm \Pnl
 \right|\leq \fC(n,m,a) \lN^\frac{a+2}{2} 
\,.
\end{equation}
\end{theorem}

In particular, Theorem \ref{thm:exp:P} proves the validity of Bogoliubov theory up to an error of order $\mathcal{O}(N^{-1})$, i.e., 
\begin{equation}
\Tr_{\fHN}\cAm\PNn = \Tr_\Fp\FockAm\Pzn + \mathcal{O}(\lN)\,,
\end{equation} 
which improves previously known error estimates of order $\mathcal{O}(\lN^{1/2})$. 

The coefficients $\Tr_\Fp\FockAm\Pnl$ in \eqref{eqn:thm:exp:P} are not necessarily $N$-independent because $\FockAm$ arises from conjugating an operator $\cAm$ on the $N$-body Hilbert space with the $N$-dependent unitary map $\UNp$. Unless $\Am$ is an operator acting only on $\fHp^m$ (such as, for example, $A^{(1)}=q$), 
this conjugation yields factors $\sqrt{N-\Np}$ comparable to \eqref{eqn:FockHN}. Hence, to extract the $N$-independent contributions in each order, one needs to expand $\FockAm$ in $\lN^{1/2}$ up to the  order of the approximation. Equivalently, one derives in this way an expansion of the reduced $m$-particle density matrices of $\PNn$.
For example, the one-particle reduced density matrix 
$$\gamma^{(n)}_{1;N}:=\Tr_{\fH^{N-1}}\PNn$$
admits the asymptotic expansion 
\begin{equation}\label{eqn:1:p:RDM:expansion}
\Tr_\fH\Big|\gamma^{(n)}_{1;N}-\sum\limits_{\l=0}^a \lN^\l\,\tilde{\gamma}_{1;\l}^{(n)}\Big|\leq \fC(n,a)\lN^{a+1}\,,
\end{equation}
where the coefficients $\tilde{\gamma}^{(n)}_{1;\l}\in\cL(\fH)$ are independent of $N$ and can be retrieved as described above. For example, the first correction to the leading order $\tilde{\gamma}_{1;0}^{(n)}=\dzn|\varphi\rangle\langle\varphi|$ is given by
\begin{equation}\label{eqn:1:p:RDM}\begin{split}
\tilde\gamma^{(n)}_{1;1} (x;y)\;=\;& \varphi(x)\Tr_\Fp \ad_y \Pn_1 + \varphi(y)\Tr_\Fp a_x\Pn_1 \\
&+ \Tr_\Fp  \ad_y a_x\Pzn -\varphi(x)\varphi(y)\Tr_\Fp\Pzn\Np 
\end{split}\end{equation}
(see also \cite[Theorem 2]{QF} for the dynamical counterpart of this statement).
For the ground state of a homogeneous Bose gas on the torus, a corresponding result was recently shown in \cite{nam2020_2}, using different methods. Note that in this case, the first line in \eqref{eqn:1:p:RDM} vanishes by translation invariance\footnote{
In this case, one computes $\tilde{\gamma}^{(0)}_{1;1} = -\sum_{k\neq 0}\gamma_k^2|\varphi\rangle\langle\varphi| + \sum_{k\neq 0}\gamma_k^2|\varphi_k\rangle\langle\varphi_k|$, where $\varphi_k =\e^{\i k\cdot x}$, $\varphi=\varphi_0$, $\gamma_k= \alpha_k(1-\alpha_k^2)^{-1/2}$, and $\alpha_k=\hat{v}(k)(k^2+\hat{v}(k)+\sqrt{k^4+2k^2\hat{v}(k)})^{-1}$, where $\hat{v}$ denotes the Fourier transform of $v$.
}.

Theorem \ref{thm:exp:P} yields an asymptotic expansion of the projector $\Pn$ onto the subspace $\fEn$ of the excitation Fock space, which is defined as
$$
\fEn=\bigoplus\limits_{\nu\in\In}\tfEnu\,,\qquad \tfEnu=\left\{\bPhi\oplus 0\,:\, \bPhi\in\FNp\,,\,\FockHN\bPhi= \Enu\bPhi\right\}
$$
(see Definition \ref{def:Pn}). The following statement is proven in Section \ref{subsec:proofs:cor}:

\begin{cor}\label{cor:trace:norm}
Let $a\in\N_0$. Under Assumptions \ref{ass:V}, \ref{ass:v} and \ref{ass:cond}, there exists a constant $\fC(n,a)$ such that
\begin{equation}
\Tr_{\Fp}\Big|\Pn-\sum\limits_{\l=0}^a\lN^\frac{\l}{2}\Pnl\Big| \leq \fC(n,a)\lN^\frac{a+1}{2}
\end{equation}
for sufficiently large $N$.
\end{cor}

By means of Bogoliubov transformations, the operators $\Pnl$  can be brought into a more explicit form. For example, the first order correction for the ground state ($n=0$) is given by
\begin{equation}\label{eqn:Pl:explicit:1}
\P^{(0)}_1 = \BogUz^*\left(\BogUz\FockO^{(0)}_1\BogUz^*\right)\left(\BogUz\FockH_1\BogUz^*\right)\vac\langle\Chiz^{(0)}|+\hc\,,
\end{equation}
where $\BogUz$ is the Bogoliubov transformation diagonalizing $\FockHz$ such that $\Chiz^{(0)}=\BogUz^*\vac$. As the action of $\BogUz$ on creation/annihilation operators is  known (see \eqref{eqn:trafo:ax}), it follows that $\BogUz\FockH_1\BogUz^*\vac$ is a superposition of one- and three-particle states. Moreover, 
\begin{equation}\label{eqn:Pl:explicit:2}
\BogUz\FockO_1^{(0)}\BogUz^* = \sum\limits_{\l>0}\sum\limits_{m=1}^{\delta^{(\l)}_0} \frac{1}{\Ezz-\Ez^{(\l)}} \BogUz|\Chiz^{(\l,m)}\rangle\langle\Chiz^{(\l,m)}|\BogUz^*\,,
\end{equation}
where $\{\Chiz^{(\l,m)}\}_{1\leq m\leq \delta^{(\l)}_0}$ denotes a basis of the eigenspace $\mathfrak{E}_0^{(\l)}$ of $\FockHz$ and can be written as
\begin{equation}\label{eqn:Pl:explicit:3}
\Chiz^{(\l,m)}=\BogUz^*
\frac{\big(\ad(\xi_0)\big)^{\nu_0}}{\sqrt{\nu_0!}}
\frac{\big(\ad(\xi_1)\big)^{\nu_1}}{\sqrt{\nu_1!}}
\,\mycdots\,
\frac{\big(\ad(\xi_k)\big)^{\nu_{k}}}{\sqrt{\nu_{k}!}}
\vac
\end{equation}
for suitable $\xi_j\in\fHp$, $k\in\N_0$, and $(\nu_0\mydots \nu_k)\in\N_0^{k+1}$ depending on $\l$ and $m$
(see Lemma~\ref{lem:known:results:FockHz:2}).
Since $\BogUz\FockO^{(0)}_1\BogUz^*$ is particle-number preserving, only the basis elements $\Chiz^{(\l,m)}$ with one and three particles contribute to \eqref{eqn:Pl:explicit:1}, and  applying $\BogUz^*$ to the result yields an explicit formula for $\P^{(0)}_1$. The general case ($n\geq 0$, $\l\geq 1$) can be treated analogously.
\medskip

In our second main result, we derive from Theorem \ref{thm:exp:P} an expansion of the low-energy spectrum of $\HN$ with $N$-independent coefficients.

\begin{theorem}\label{thm:energy}
Let $n\in\N_0$.
Under Assumptions \ref{ass:V}, \ref{ass:v} and \ref{ass:cond}, it holds for any $a\in\N_0$ and sufficiently large $N$ that
\begin{equation}\label{eqn:cor:expansion}
\left|\frac{1}{\dzn}{\sum\limits_{\;\nu\in\In}\dnu\ENnu} -  N\eH -  \sum\limits_{\l=0}^a\lN^\l\En_\l \right|\; \leq\; \fC(n,a) \lN^{a+1}
\end{equation}
for some constant $\fC(n,a)$ and for $\In$, $\dnu$, $\ENnu$ and $\dzn$ as in Definitions \ref{def:FockHz} and \ref{def:PNn}.
The coefficients are given as
\begin{equation}\label{eqn:cor:general}
\En_\l\;:=\;\frac{1}{\dzn}
\sum\limits_{\nu=1}^{2\l}\sum\limits_{\substack{\bj\in\N^\nu\\|\bj|=2\l}}
\sum\limits_{\substack{\bm\in\N_0^{\nu-1}\\|\bm|=\nu-1}}\frac{1}{\kappa(\bm)}
\Tr_\Fp\Pzn\FockHjo\FockOn_{m_1}\mycdots\FockH_{j_{\nu-1}}\FockOn_{m_{\nu-1}}\FockHjnu
\end{equation}
for $\FockOn_m$ as in Definition \ref{def:Pna} and where 
\begin{equation}
\kappa(\bm):=1+\left|\left\{\mu\,:\, m_\mu=0\right\}\right| \in\{1\mydots \nu-1\}
\end{equation} 
is the number of operators $\Pzn$ within the trace.
\end{theorem}

All half-integer powers of $\lN$ vanish by parity. Equivalently, this can be understood as a consequence of Wick's rule (Lemma \ref{lem:wick}) and of the fact that the eigenstates of $\FockHz$ are given explicitly  as Bogoliubov transformations of states with fixed particle number (Lemma~\ref{lem:known:results:FockHz:2}).
Moreover, note that the contribution to \eqref{eqn:cor:general} from each $\nu$ decomposes into products of $\kappa(\bm)$ inner products.

Theorem \ref{thm:energy} recovers the  expressions from perturbation theory as discussed in the introduction. In particular, for any $n\in\N_0$ such that $\dzn=1$ (which applies, e.g., to the ground state), $\ENn$ is a non-degenerate eigenvalue of $\HN$, and \eqref{eqn:cor:expansion} reduces to
\begin{equation}
\ENn=N\eH+\sum\limits_{\l=0}^a \lN^\l E_\l^{(n)} + \mathcal{O}(\lN^{a+1})\,.
\end{equation}
In this case, the first two coefficients in \eqref{eqn:cor:expansion} simplify  to
\begin{subequations}\label{eqn:cor:explicit}
\begin{eqnarray}
E_1^{(n)}&=& \lr{\Chizn,\FockHt\Chizn} +\lr{\Chizn,\FockHo\frac{\Qzn}{\Ezn-\FockHz}\FockHo\Chizn}\,,\\
E_2^{(n)}&=& \sum\limits_{\nu=1}^4\sum\limits_{\substack{\bj\in\N^\nu\\|\bj|=4}}\lr{\Chizn,\FockHjo\frac{\Qzn}{\Ezn-\FockHz}\FockHjt\mycdots\frac{\Qzn}{\Ezn-\FockHz}\FockH_{j_\nu}\Chizn}\nonumber\\
&&-E_1^{(n)}\lr{\Chizn,\FockHo\frac{\Qzn}{(\Ezn-\FockHz)^2}\FockHo\Chizn}\,. 
\end{eqnarray}
\end{subequations}

\begin{remark}\label{rem:growth:of:const}
Theorem \ref{thm:exp:P} holds for any fixed $n\in\N_0$, $a\in\N_0$ and $m\in\N$ for sufficiently large $N$, with an error $\fC(n,m,a)$ that is neither uniform in $n$ nor in $m$ or $a$. In particular, $\fC(n,m,a)$ depends on $|\Ezn|$, hence the statement is non-trivial only for eigenvalues of $\HN$ of order one above the ground state energy.

Moreover, $\fC(n,m,a)$ grows rapidly in the order $a$ of the approximation. In the special case where $v\in L^\infty(\R^d)$, our estimates imply that 
\begin{equation*}
\fC(n,m,a)\leq  \big(\fC(n,m)(a+1)\big)^{(a+6)^2} \,,
\end{equation*} 
 and the bound is certainly worse in the general case (see Remark \ref{rem:v:bd} below). We do not expect this estimate to be optimal, especially as Borel summability was proven for a comparable perturbative expansion of the mean-field dynamics on Fock space for bounded interactions \cite{ginibre1980}. Also in that setting, the available estimates for unbounded potentials are worse and, in particular, insufficient to conclude Borel summability \cite{ginibre1980_2}.
\end{remark}
\medskip

\begin{remark}
As explained in Section \ref{sec:assumptions}, Assumptions \ref{ass:V}, \ref{ass:v} and \ref{ass:cond} are satisfied, e.g., by bounded positive definite potentials and by the repulsive Coulomb potential in $d=3$.
These assumptions ensure that Bogoliubov theory is valid for our model, i.e., that all assumptions in \cite{lewin2015_2} are satisfied.
In that work, it is shown that $\FockHz$ approximates $\FockH$ to leading order for any self-adjoint $T$
that is bounded from below, and for interaction potentials 
$$-c_1(T_1+T_2+c_2) \leq v(x_1-x_2) \leq c_3(T_1+T_2+1)\,,\qquad 0<c_1<1\,,\quad c_2,c_3>0$$
\cite[(A1)]{lewin2015_2} such that there exists a unique non-degenerate minimizer for the Hartree functional, and such that the operators $\Ko$ and $\Kt$ from \eqref{K} ($\Kt$ as operator $\fH^*\to\fH$) are Hilbert--Schmidt \cite[(A2)]{lewin2015_2}. Moreover, it is required that
$$\HN-N\eH\geq c\sum\limits_{j=1}^Nh_j+\smallO{N}$$
for some $0<c<1$ \cite[(A3s)]{lewin2015_2}.
Our analysis, which can be understood as a perturbative expansion of $\FockH$ around the leading order $\FockHz$, relies on the result proven in \cite{lewin2015_2}: We need $\Enu\approx\Ezn$ (for sufficiently large $N$) to find a suitable contour $\gan$ enclosing $\Ezn$ as well as all $\Enu$ with $\nu\in\In$, and we require  that $\Chin\to\Chizn$ strongly in the norm induced by the quadratic form of $\FockHz$  to conclude that $\lr{\Chin,\Np\Chin}$ is bounded uniformly in $N$ (see Lemma~\ref{lem:known:results:FockHN}).

In contrast to the generic setting from \cite{lewin2015_2}, we choose  $T=-\Delta+\Vext$ and consider a positive definite interaction $v$ satisfying the stronger bound \eqref{eqn:ass:v:2:Delta:bound}, which implies \cite[(A1--A2)]{lewin2014} (see Lemma \ref{lem:hH}). 
In particular, \eqref{eqn:ass:v:2:Delta:bound} is crucial to bound $\boldKth$ by powers of $\Np$, and $\boldKf$ in terms of $\d\Gamma_\perp(\hH)^{1/2}$ and powers of $\Np$. 
Moreover, Assumption \ref{ass:cond} is stronger than \cite[(A3s)]{lewin2015_2} since we require an error of at most $\mathcal{O}(N^{1/3})$ to control arbitrary moments of $\Np$ with respect to $\Chin$, as explained below.

Our analysis generalises to certain interactions $v$ which are not of positive type, and to a class of confining potentials $\Vext$ that do not diverge at infinity. More precisely, we can cover all potentials $v$ and $\Vext$ which are such that all assumptions in \cite{lewin2015_2} and Assumption~\ref{ass:cond} are satisfied. 
For example, it is shown in \cite[Section 3.2]{lewin2015_2} that  a trapped two-dimensional gas with repulsive Coulomb  interactions and $\Vext$ diverging sufficiently fast at infinity, 
$$\HN=\sum\limits_{j=1}^N\left(-\Delta_j+\Vext(x_j)\right) - \lN\sum\limits_{i<j}\ln|x_i-x_j|\,,\qquad d=2\,,$$
satisfies \cite[(A1--A3s)]{lewin2015_2} as well as Assumption \ref{ass:cond} \cite[Lemma 3.7]{lewin2015_2} although $v(x)=-\ln|x|$ is not of positive type.
Moreover, it is explained in \cite[Section~3.2]{lewin2015_2} that bosonic atoms below a critical binding number $t_c$, which are described by the rescaled Hamiltonian
$$H_{t,N}=\sum_{j=1}^N\left(-\Delta_j-\frac{1}{t|x_j|}\right)+\lN\sum_{i<j}\frac{1}{|x_i-x_j|}\,,\qquad t<t_c\in(1,2)\,,\qquad d=3\,,$$ 
meet all criteria, including our Assumption \ref{ass:cond}. Other viable choices for $T$ are the Laplace operator on a bounded subset of $\R^d$ with Dirichlet, Neumann or periodic boundary conditions, or relativistic kinetic terms.
\end{remark}
\medskip

Finally, we construct an asymptotic expansion of the $N$-body eigenstates $\Psi_N^{(n)}$ of $\HN$ that correspond to non-degenerate eigenvalues of $\FockHz$.

\begin{theorem}\label{thm:wave:fctn}
Let $a\in\N_0$ and let Assumptions \ref{ass:V}, \ref{ass:v} and \ref{ass:cond} be satisfied. Assume that $n\in\N_0$ such that $\dzn=1$ and let $\PsiNn\in\fENn$. Then, for a suitable choice of the phase of $\Chin_0$, there exists a constant $\fC(n,a)$ such that
\begin{equation}\label{norm_conv_wf}
\Big\|\PsiNn-\sum\limits_{\l=0}^a \lN^\frac{\l}{2}\sum\limits_{k=0}^N\varphi^{\otimes (N-k)}\otimes_s\big(\Chin_\l\big)^{(k)}\Big\|_{\fH^N}\leq \fC(n,a)\lN^\frac{a+1}{2}
\end{equation}
for sufficiently large $N$,
where
\begin{subequations}
\begin{eqnarray}
\Chin_\l&:=&\sum\limits_{j=0}^\l \alpha_j\,\tilde{\Chi}^{(n)}_{\l-j} \qquad (\l\geq 1)\,,\label{chi_def_thm}\\
\tilde{\Chi}^{(n)}_\l&:=&\sum\limits_{\nu=1}^\l\sum\limits_{\substack{\bj\in\N^\nu\\|\bj|=\l}} \Pn_{j_1}\,\mycdots \Pn_{j_\nu}\,\Chin_0 \qquad (\l\geq 1)\,\label{chi_tilde_def_thm},\end{eqnarray}
and with
\begin{equation}
\alpha_0:=1\,,\qquad
\alpha_{2n-1}:=0\,,\qquad
\alpha_{2n}:=-\frac12\sum\limits_{\substack{\bj\in\N_0^4\\j_1,j_2<2n\\|\bj|=2n}} \alpha_{j_1}\alpha_{j_2}\lr{\tilde{\Chi}^{(n)}_{j_3},\tilde{\Chi}^{(n)}_{j_4}} 
\end{equation}
\end{subequations}
for $n\geq 1$.
\end{theorem}

Theorem \ref{thm:wave:fctn} is an immediate consequence of a much more general statement: if a rank-one projector admits an asymptotic expansion in a small parameter $\varepsilon$, this implies an asymptotic expansion of the corresponding wave function. Since we could not find any proof of this seemingly obvious assertion, we prove it for a generic perturbative setting in Appendix \ref{appendix:wf}.
By parity, the parameters $\alpha_\l$ vanish for $\l$ odd, which can be seen analogously to the vanishing of the half-integer powers of $\lN$ in Theorem \ref{thm:energy}. Note that \eqref{norm_conv_wf} also holds with $\Chin_\l$ replaced by $\tilde{\Chi}^{(n)}_\l$ times an overall factor $\alpha^{(a)} = \sum_{\l=0}^a \lN^{\l/2} \alpha_\l$.

\medskip

\begin{remark}
Recall that each Bogoliubov eigenstate $\Chizn$ can be expressed as Bogoliubov transformation $\BogUz^*$ of a wave function with fixed particle number $m_n\in\N_0$ (see  Lemma \ref{lem:known:results:FockHz:2}). Consequently, $\Chin_\l$ can be written as $\BogUz^*$ acting on a superposition of wave functions with $\mu\leq m_n+3\l$ particles with $\mu+\l+m_n$ even. 
To see this, note that $\BogUz\FockOnk\BogUz^*$ is particle number preserving, and $\BogUz\FockH_j\BogUz^*$ has even/odd parity for $j$ even/odd and contains at most $j+2$ creation operators. 
Hence, the maximum number of creation operators in \eqref{chi_tilde_def_thm} is contributed by $\nu = 1$, namely by the term containing exclusively operators $\FockHo$ and exactly one operator $\Pz$ (i.e., $\nu=\l$, $\bj = (1,1\mydots 1)$ and $\bk=(1,1\mydots 1,0)$ in \eqref{eqn:Pna}). 
Such initial data are used for a perturbative expansion of the dynamics of the Bose gas in the mean-field limit in \cite{QF}.
\end{remark}\medskip

\begin{remark}
For any given $\l\in\N$, the formula \eqref{chi_tilde_def_thm} can be simplified further since many terms vanish by parity and most of the remaining terms can be grouped into summands which only differ by a prefactor (compare \eqref{wf_and_energies_iteratively} below), e.g.,
\begin{equation}
\tilde{\Chi}^{(n)}_2 = \left(\Pn_2+\Pn_1\Pn_1\right)\Chizn = \left(\FockOn_2\FockH_2 + \FockOn_1\FockH_1\FockOn_1\FockH_1  \right)\Chizn\,.
\end{equation}
The approximating wave functions in Theorem~\ref{thm:wave:fctn} are constructed via  the eigenvalue equation $\sum_{\l=0}^{\infty} \lN^{\l/2} \Pnl\Chin = \Chin$ (see Appendix \ref{appendix:wf}). 
Alternatively, one can (formally) derive simpler formulas for both $\tilde{\Chi}^{(n)}_\l$ and the coefficients $\En_\l$ from Theorem \ref{thm:energy} by an analogous construction for the eigenvalue equation $\FockH\Chin = \En\Chin$.
A formal computation yields
\begin{subequations}\label{wf_and_energies_iteratively}
\begin{eqnarray}
\tilde{\Chi}^{(n)}_{\l}& =& \sum_{\nu=1}^{\l} \sum\limits_{\substack{\bj\in\N^{\nu}\\|\bj|=\l}} \FockOno \FockH^{\,'}_{j_1} \mycdots\, \FockOno \FockH^{\,'}_{j_{\nu-1}} \FockOno \FockH_{j_{\nu}} \Chizn\,,\\
\En_{\l} &=& \sum_{\nu=1}^{2\l} \sum\limits_{\substack{\bj\in\N^{\nu}\\|\bj|=2\l}} \lr{\Chizn, \FockH_{j_1} \FockOno \FockH^{\,'}_{j_2} \mycdots\,\FockOno \FockH^{\,'}_{j_{\nu-1}} \FockOno \FockH_{j_{\nu}} \Chizn} \,,
\end{eqnarray}
\end{subequations}
where $\FockH^{\,'}_j=\FockH_j$ for $j$ odd and $\FockH^{\,'}_j=\FockH_j -E_{j/2}^{(n)}$ for $j$ even.
Here, $\tilde{\Chi}^{(n)}_{\l}$ is given in terms of the coefficients $\En_\l$, which  are determined iteratively. For the first few orders, one easily verifies that \eqref{wf_and_energies_iteratively} coincides with the expressions from Theorems~\ref{thm:wave:fctn} and \ref{thm:energy} for $\dzn=1$.
\end{remark}

\subsection{Strategy of proof}

In the remainder of this section, we explain the proof of Theorems \ref{thm:exp:P} and \ref{thm:energy}.
We begin with extending $\FockHN$ to the full excitation Fock space $\Fp$ in the following way:

\begin{definition}
We extend $\FockHN$ (see Definition \ref{def:FockHN}) from $\FNp$ to the full Fock space $\Fp$ as
\begin{equation}\label{eqn:FockH:extension}
\FockH:=\FockHN\oplus\Emo\,,
\end{equation} 
where
\begin{equation}\label{eqn:Emo}
\Emo:=\Ezero-(\Eo-\Ezero)
\end{equation}
with $\En$ the eigenvalues of $\FockHN$ (see Definition \ref{def:FockHN}).
Consequently, the low-energy spectrum of $\FockH$ consists of the eigenvalues
\begin{equation}
\Emo<\Ezero<\Eo<\dots<\En<\dots\,.
\end{equation}
\end{definition}

Note that we could have extended $\FockHN$ to $\Fp$ in many ways.
To motivate the choice \eqref{eqn:FockH:extension}, recall  that our aim is to expand the spectral projectors of $\FockH$ around the corresponding spectral projectors of $\FockHz$, which we do by expressing them as contour integrals over the resolvent of $\FockH$ and subsequently expanding $(z-\FockH)^{-1}$ around $(z-\FockHz)^{-1}$.
Let us first consider the case where the eigenvalues $\En$ and $\Ezn$ of $\FockH$ and $\FockHz$, respectively, are non-degenerate. In view of \eqref{intro:functional:calculus}, we require an $\mathcal{O}(1)$ contour $\gan$ that encloses both $\En$ and $\Ezn$ and leaves the remaining spectrum of $\FockH$ outside. 
The choice $\FockH=\FockHN\oplus c$, for $c$ a finite distance away from any point in the spectrum of $\FockHN$, ensures that $\FockH$ has precisely one (infinitely degenerate) additional eigenvalue $c$ compared to $\FockHN$.
Since the spectrum of $\FockHN$ is bounded from below by $\Ezero$, we place $c$ at a finite distance below $\Ezero$, for simplicity such that the spectral gaps below and above $\Ezero$ have the same size.

If $\Ezn$ is degenerate, the expansion must be done carefully because we cannot exclude that non-degenerate eigenvalues of $\FockH$ become degenerate in the limit $N\to\infty$.
By \cite{lewin2015_2},  every low-energy eigenvalue of $\FockH$ converges to an eigenvalue of $\FockHz$ (see Lemma \ref{lem:known:results:FockHN:E}), but the situation may occur that  an eigenvalue $\Ezn$ of $\FockHz$ is, for instance, twice degenerate, and there exist (for any finite $N$) two eigenvalues $E^{(n_1)}\neq E^{(n_2)}$ of $\FockH$ such that 
$$\lim\limits_{N\to\infty} E^{(n_1)} =\Ezn =\lim\limits_{N\to\infty} E^{(n_2)}\,.$$
In this case, it makes sense to expand  the sum of the corresponding projectors around $\Pzn$, which becomes apparent when recalling the formula \eqref{intro:functional:calculus}: Since each closed curve of order one around $\Ezn$ must enclose both poles $E^{(n_1)}$ and $E^{(n_2)}$ of $(z-\FockH)^{-1}$, the contour integral gives precisely the sum of the two spectral projections.
This motivates the following definition:

\begin{definition}\label{def:Pn}
For any $n\in\N_0$, we define the path
\begin{equation}\label{eqn:gamma:n}
\gan:=\left\{\Ezn+\fgn\e^{\i t}\,:\, t\in[0,2\pi)\right\}\subset\C\,,
\end{equation}
where 
\begin{equation}\label{eqn:fgn}
\fgn:= \tfrac12\min\left\{\gz^{(n-1)}\,,\,\gz^{(n)}\right\} 
\end{equation}
for $\gzn$ as in \eqref{eqn:gap:FockHz}.
For $n\in\N_0$, define
\begin{equation} \label{eqn:def:Pn}
\Pn:=\frac{1}{2\pi\i}\goint\ResHz\dz\,,\qquad \Qn:=\id_\Fp-\Pn
\end{equation}
and
\begin{equation}
\fEn\;:=\;\Pn\Fp \;=\; \bigoplus\limits_{\,\nu\in\In} \tfEnu \;\subset\FNp\oplus 0\,,
\end{equation}
with $\In$ as in \eqref{In} and where $\tfEnu$ denotes the eigenspace of $\FockH$ at $\Enu$,
\begin{equation}
\tfEnu:=\left\{\bPhi\oplus0 \,:\,\bPhi\in\FNp\,,\,\FockHN\bPhi=\Enu\bPhi\right\}\,,
\end{equation}
with dimension $\dnu$ as in \eqref{dnu}.
We denote normalized elements of $\fEn$ as
\begin{equation}\label{eqn:Chin}
\Chin:=\ChiNn\oplus0\,.
\end{equation}
For $n=-1$, we define $\P^{(-1)}$ as the projector onto the eigenspace of $\FockH$ associated with $\Emo$, i.e.,
\begin{equation}
\fEmo:=\left\{0\oplus\bPhi\,:\,\bPhi\in\FgNp\right\}\,,\qquad \P^{(-1)}:=\id_{\fEmo}\,.
\end{equation}
\end{definition}

Next, we expand  $\FockH$ in powers of $\lN^{1/2}$.
The $N$-dependence in $\FockH$ has two sources: first, $\FockH$ is  defined as the direct sum of $\FockHN$ on $\FNp$ and a conveniently chosen constant on $\Fock_\perp^{>N}$; second, the operators in $\FockHN$  come with $N$-dependent prefactors.
To deal with the first point, we write $\FockH$ on $\Fp$ as
\begin{equation}\label{FockHplusminus}
\FockH=\FockHminus+\FockHplus
\end{equation}
with
\begin{subequations}
\begin{eqnarray}
\FockHminus
&:=&  \boldKz +  \left(1-\frac{\Np-1}{N-1}\right)  \boldKo
+\left( \boldKt\frac{\sqrt{\big[(N-\Np)(N-\Np-1)\big]_+}}{N-1}+\hc\right)\nonumber\\
&&+\left( \boldKth\frac{\sqrt{\big[N-\Np\big]_+}}{N-1}+ \hc\right) \label{FockHminus}
+\frac{1}{N-1} \boldKf\,,\\
\FockHplus &:= &0\oplus\left(\Emo-\boldKz-\left(1-\frac{\Np-1}{N-1}\right) \boldKo -\frac{1}{N-1}\boldKf\right)\,,\label{FockHplus}
\end{eqnarray}
\end{subequations}
where $[\cdot]_+$ denotes the positive part. Note that $\boldKz$, $\boldKo$ and $\boldKf$ conserve the particle number, hence the restriction to $\FgNp$ in \eqref{FockHplus} makes sense.
The first term $\FockHminus$ is defined on the full space $\Fp$. 
To obtain $\FockHminus$, we added to $\FockH$ the missing contributions to $\boldKz$, $\boldKo$ and $\boldKf$ on the sectors $\FgNp$, and subtracted them again in $\FockHplus$.
Finally, we expand the square roots from $\FockHminus$ in a Taylor series (see \cite[Appendix C]{QF} for a proof).

\begin{lem}\label{lem:calculus}
Let $a\in\N_0$ and $c_\l^{(j)}$ and $d_{\l,j}$ as in \eqref{taylor:coeff}.
\lemit{
\item \label{lem:calculus:1}
Define the operator $\tFockRtha$ on $\Fp$ via the identity
\begin{equation}\label{eqn:lem:calculus:R3:1}
\frac{\sqrt{\big[N-\Np\big]_+}}{N-1}
\;=\; \sum\limits_{\l=0}^a c_\l\, \lN^{\l+\frac12}
(\Np-1)^\l + \lN^{a+\frac32}\,\tFockRtha\,.
\end{equation}
Then $[\tFockRtha,\Np]=0$ and
\begin{equation}\label{eqn:lem:calculus:R3:2}
\norm{\tFockRtha\bPhi}_\Fp
\;\leq\; 2^{a+1}\norm{(\Np+1)^{a+1}\bPhi}_\Fp
\end{equation} 
for $\bPhi\in\Fp$.
\item \label{lem:calculus:2}
Define the operator $\tFockRta$ on $\Fp$ through 
\begin{equation}\label{eqn:lem:calculus:R2:1}
\frac{\sqrt{\big[(N-\Np)(N-\Np-1)\big]_+}}{N-1}
\;=\; \sum\limits_{\l=0}^a \lN^\l \sum\limits_{j=0}^\l d_{\l,j} (\Np-1)^j\,
+\lN^{a+1}\tFockRta\,.
\end{equation}
Then $[\tFockRta,\Np]=0$ and
\begin{equation}\label{eqn:lem:calculus:R2:2}
\norm{\tFockRta\bPhi}_\Fp\;\leq\; (a+1)^24^{a+1} \norm{(\Np+1)^{a+1}\bPhi}_\Fp
\end{equation} 
for $\bPhi\in\Fp$.

}
\end{lem}
With this, we can expand $\FockHminus$  in powers of $\lN^{1/2}$:

\begin{proposition}\label{lem:expansion:HN}
Let $a\in\N_0$.
In the sense of operators on $\Fp$, it holds that
\begin{eqnarray}
\FockHminus
&=&\sum\limits_{j=0}^{a}\lN^{\frac{j}{2}}\FockHj +\lN^{\frac{a+1}{2}}\FockR_a
\end{eqnarray}
with  $\FockH_j$ as in Definitions \ref{def:FockHz} and \ref{def:Pna}.
The remainders are given as
\begin{subequations}\label{FockR}
\begin{eqnarray}
\FockRz&:=&\FockRza+\lN^\frac12\boldKf\,,\label{FockRz}\\
\FockRo&:=&\FockRoa+\boldKf
\end{eqnarray}
and
\begin{eqnarray}
\FockRza&:=&\left(\boldKth \sqrt{\frac{\big[N-\Np\big]_+}{N-1}}+\hc\right)  + \lN^{\frac12}\bigg((\boldKt \tFockRt_0+\hc)\nonumber\\
&&\hspace{4.8cm}-(\Np-1)\boldKo\bigg)\,,\qquad\\
\FockRoa &:=&-(\Np-1)\boldKo
+\left(\boldKt \tFockRt_0+\hc\right) +\lN^{\frac12}\left(\boldKth \tFockRth_0+\hc\right),\quad\\
\FockR_{2j} &:=& 
\boldKth\tFockRth_{j-1}
+\lN^{\frac12} \boldKt \tFockRt_{j}
+\hc,\label{FockHt:R:2n-1}\\
\FockR_{2j+1} &:= &
\boldKt \tFockRt_{j}
+\lN^{\frac12}\boldKth \tFockRth_{j}
+\hc\label{FockHt:R:2n}
\end{eqnarray}
\end{subequations} 
for $j\geq1$, with $\tFockRt_j$ and $\tFockRth_j$ from Lemma \ref{lem:calculus}.
\end{proposition}

The next step is to expand $\Pn$ around $\Pzn$, using that
\begin{equation} \label{eqn:functional:calculus}
 \Pzn=\frac{1}{2\pi\i}\goint\ResHzz\dz
\end{equation}
because $\gan$ from \eqref{eqn:gamma:n} encloses $\Ezn$. In view of the definition \eqref{eqn:def:Pn} of $\Pn$, we first expand $(z-\FockH)^{-1}$ around $(z-\FockHz)^{-1}$ and  integrate the resulting expressions along $\gan$.

\begin{lem}\label{lem:expansion:resolvent}
Let $a\in\N_0$ and $z\in\varrho(\FockH)\cap\varrho(\FockHz)$, where $\varrho$ denotes the resolvent set.
Then
\begin{equation} \label{eqn:lem:exp:resolvent}
\ResHz
\;=\;\ResHzz\sum\limits_{\l=0}^a\lN^\frac{\l}{2}\,\FockT_\l (z)
+ \lN^\frac{a+1}{2}\RestHz\,\FockS_a(z)
+\RestHz\,\FockHplus\ResHz\,,
\end{equation}
where
\begin{eqnarray}
\FockT_\l(z)&=&\sum\limits_{\nu=1}^\l \,\sum\limits_{\substack{\bj\in\N^\nu\\[2pt]|\bj|=\l}} \FockH_{j_1}\ResHzz\FockH_{j_2}\ResHzz\mycdots \FockH_{j_\nu}\ResHzz \,,\label{eqn:T:l}\\
\FockS_a(z)&=&\sum\limits_{\nu=0}^{a} \FockRnu\ResHzz\FockT_{a-\nu}(z)\,.
\label{eqn:S:a}
\end{eqnarray}
Here, the notation is understood such that $\FockT_0(z)=\id$.
\end{lem}
The proof of Lemma \ref{lem:expansion:resolvent} is postponed to Section \ref{subsec:proofs:resolvent}.
Essentially, one uses the identities
$$\FockHminus=\FockHz+\lN^\frac12\FockRz\,,\qquad \FockRz=\FockHo+\lN^\frac12\FockRo,
$$ 
which follow from Proposition \ref{lem:expansion:HN}, to conclude that
\begin{equation}\begin{split}
\RestHz\;=\;&\ResHzz+\lN^\frac12\RestHz\FockRz\ResHzz\\ 
\;=\;& \ResHzz+\lN^\frac12\ResHzz\FockHo\ResHzz +\mathcal{O}(\lN)\,,
\end{split}\end{equation}
and iterating this procedure up to order $\mathcal{O}\big(\lN^{(a+1)/2}\big)$ concludes the proof.\medskip

The next step is to integrate \eqref{eqn:lem:exp:resolvent} along the contour $\gan$ as in \eqref{eqn:functional:calculus}.
The first term in \eqref{eqn:lem:exp:resolvent} gives an integral over products of alternately  $(z-\FockHz)^{-1}$ and $\FockHj$. After decomposing $1=\Pzn+\Qzn$ in each resolvent, we note that the term with exclusively $\Qzn$ vanishes because the integrand is, by construction, holomorphic in the interior of $\gan$. The remaining terms, all of which contain at least one projection $\Pzn$, can be simplified using the residue theorem. Note that $\Pzn/(z-\FockHz)=\Pzn/(z-\Ezn)$, hence the number of operators $\Pzn$ determines the order of the pole at $z=\Ezn$.  

The second term in \eqref{eqn:lem:exp:resolvent} is of the same structure as the first one but starts with $(z-\FockHminus)^{-1}$ instead of $(z-\FockHz)^{-1}$. For later convenience, we decompose the first identity as $1=\Pn+\Qn$. Moreover, in case of $\Qn$, we resolve all remaining identities as $1=\Pzn+\Qzn$  and note that the contribution with $\Qn$ and exclusively $\Qzn$ vanishes  as the integrand is holomorphic.

Finally, in the last term of \eqref{eqn:lem:exp:resolvent}, we decompose both identities as $1=\Pn+\Qn$ and observe that $\Pn\FockHplus=0$ because $\Pn$ projects onto a subset of $\FNp$, where $\FockHplus$ equals zero. This leaves only the term with twice $\Qn$, which vanishes upon integration.
In summary, we obtain the following formula for $\Pn$:

\begin{proposition}\label{prop:exp:P}
Let $a\in\N_0$, $n\in\N_0$, and $\gan$ as in \eqref{eqn:gamma:n}.
Then
\begin{eqnarray}
\Pn 
&=& \sum\limits_{\l=0}^a \lN^\frac{\l}{2}\Pnl+\lN^\frac{a+1}{2} \left(\FockBPn(a)+\FockBQn(a)\right)
\label{eqn:exp:P}
\end{eqnarray}
for $\Pnl$ as in Definition \ref{def:Pna} and where
\begin{equation}
\FockBPn(a)=\sum\limits_{\nu=0}^{a}\sum\limits_{m=1}^{a-\nu}\sum\limits_{\substack{\bj\in\N^m\\[2pt]|\bj|=a-\nu}}\frac{1}{2\pi\i}\goint\frac{\Pn}{z-\FockHminus}\,\FockRnu\ResHzz\FockHjo\ResHzz\mycdots\FockH_{j_m} \ResHzz\dz\,\label{B:P}
\end{equation}
and
\begin{equation}\begin{split}
\FockBQn(a)
\;=\;&\hspace{-0.1cm}\sum\limits_{\nu=0}^{a} 
\sum\limits_{m=1}^{a-\nu}
\sum\limits_{\substack{\bj\in\N^m\\[2pt]|\bj|=a-\nu}}
\sum\limits_{\l=0}^{m} 
\sum\limits_{\substack{\bk\in\{0,1\}^{m+1}\\|\bk|=\l}}
\frac{1}{2\pi\i}\goint
\frac{\Qn}{z-\tFockH}\FockRnu \\
&\hspace{1.5cm}\times\, 
\frac{\FockIn_{k_1}}{z-\FockHz}\FockHjo\frac{\FockIn_{k_2}}{z-\FockHz}\mycdots\FockH_{j_m}\frac{\FockIn_{k_{m+1}}}{z-\FockHz}\dz \qquad\quad
\label{B:Q}
\end{split}\end{equation}
with
\begin{equation}\label{tFockO}
\FockInk=\begin{cases}
\displaystyle\Pzn &\quad k=0\,,\\[7pt]
\displaystyle\Qzn &\quad k=1\,.
\end{cases}
\end{equation}
\end{proposition}

To derive the coefficients $\En_\l$ of the energy expansion in Theorem \ref{thm:energy}, we observe that
\begin{eqnarray}
\Tr_\Fp\FockH\Pn & =& \frac{1}{2\pi\i}\Tr_\Fp\goint\frac{\FockH}{z-\FockH}\dz \;=\; \frac{1}{2\pi\i}\Tr_\Fp\goint\frac{z}{z-\FockH}\dz\nonumber\\
& =& \dzn \Ezn+\frac{1}{2\pi\i}\Tr_\Fp\goint\frac{z-\Ezn}{z-\FockH}\dz\,,
\end{eqnarray}
expand $(z-\FockH)^{-1}$ as in Lemma \ref{lem:expansion:resolvent}, and use the residue theorem to evaluate the  resulting expressions.

It remains to show that the difference
$$\left|\Tr_\Fp\FockAm\Pn -\sum\limits_{\l=0}^a\lN^\frac{\l}{2}\Tr_\Fp\FockAm\Pnl\right|$$
is of order $\lN^{(a+2)/2}$. We prove this in four steps.\medskip

\textbf{Step 1.}
First, recall that all low-energy eigenstates of $\HN$ exhibit condensation in $\varphi$, hence  the leading order contribution to $\Tr_\fHN\Am\PNn$ is determined by the condensate.
To take this into account, we define the auxiliary operator
\begin{equation}\label{Amred}
\FockAmred:=\FockAm-\expAm \oplus0\,, \qquad
\expAm:=\lr{\varphi^{\otimes m},\Am\varphi^{\otimes m}}_{\fH^m}\,,
\end{equation} 
where we already subtracted the leading order, i.e.,
\begin{eqnarray}
\Tr_\fHN\Am \PNn
&=& \Tr_\Fp\FockAm\Pn
\;=\;\dzn\expAm+ \Tr_\Fp\FockAmred\Pn\,.
\end{eqnarray}
Our goal is to conclude from Proposition \ref{prop:exp:P} that
\begin{equation}\label{eqn:thm:claim}
\Tr_\Fp\FockAmred \Pn= \sum\limits_{\l=0}^a\lN^\frac{\l}{2}\Tr_\Fp\FockAmred\Pnl +\mathcal{O}\big(\lN^\frac{a+2}{2}\big)\,,
\end{equation}
i.e., we must show that the error terms in \eqref{eqn:exp:P} are of the right order.
Given \eqref{eqn:thm:claim}, the statement of the theorem can be inferred as follows: By definition of $\FockAmred$, \eqref{eqn:thm:claim} implies that
\begin{equation}\label{eqn:thm:claim:1}
\Tr_\Fp\FockAmred\Pn
=\sum\limits_{\l=0}^a\lN^\frac{\l}{2}\Tr_\Fp\FockAm\Pnl -\expAm\sum\limits_{\l=0}^a\lN^\frac{\l}{2}\Tr_\Fp\Pnl
 +\mathcal{O}\big(\lN^\frac{a+2}{2}\big)\,.
\end{equation}
Due to Proposition \ref{prop:exp:P} and since $\Tr_\Fp\Pn=\Tr_\Fp\Pzn=\dzn$ by definition,  one can show that
\begin{equation}\label{eqn:thm:claim:2}
\dzn=\Tr_\Fp\Pn = \dzn+\sum\limits_{\l=1}^a\lN^\frac{\l}{2}\Tr_\Fp\Pnl +\mathcal{O}\big(\lN^\frac{a+1}{2}\big)
\end{equation} 
for any $a\in\N$, which implies that
$\Tr_\Fp\Pnl=0$ for any $\l\geq1$. Alternatively, this can be inferred  directly from the definition of  $\Pnl$. Hence, \eqref{eqn:thm:claim:1} yields
\begin{equation}\label{eqn:odd:lambda:vanish}
\Tr_\Fp\FockAm\Pn
=\Tr_\Fp\FockAmred\Pn +\dzn\expAm
=\sum\limits_{\l=0}^a\lN^\frac{\l}{2}\Tr_\Fp\FockAm\Pnl 
+\mathcal{O}\big(\lN^\frac{a+2}{2}\big)\,.
\end{equation}
It remains to prove the two estimates \eqref{eqn:thm:claim} and \eqref{eqn:thm:claim:2}.
To deal with both problems simultaneously, let us consider 
$$\FockA\in\big\{\FockAmred\,,\,\id\big\}\,.$$\smallskip

\noindent\textbf{Step 2.} 
First,  we show that $\FockA$ satisfies an estimate of the form 
\begin{equation}\label{eqn:A}
\norm{\FockA\bPhi}_\Fp 
\leq\fC N^{\alpha}\left( \norm{(\Np+1)\bPhi}_\Fp+\big\|\FockHz\bPhi\big\|_\Fp\right)\,.
\end{equation}
For $\FockA=\id$, this holds trivially with $\alpha=0$; for $\FockA=\FockAmred$, we prove \eqref{eqn:A} with $\alpha=-\frac12$ (Lemma \ref{lem:A}). 
Let us  explain the main idea of the proof for the simplest case  $m=1$. 
First, we use $\UNp$ to reduce the problem to an estimate on $\fHN$ and  insert identities $1=p_1+q_1$ (see \eqref{p_and_q}), i.e.,
\begin{equation}\label{eqn:A:explanation}
\norm{\FockA^{(1)}_\mathrm{red}\bPhi}_\Fp = 
\Big\|\left(p_1A^{(1)}_1p_1-\expAo + (q_1A^{(1)}_1p_1+\hc)+q_1A^{(1)}_1q_1\right)\UNp^*\bPhi\Big\|_{\fHN}
\end{equation} 
for any $\bPhi\in\FNp\oplus 0$.
For the first term, one observes that
\begin{equation}
p_1A^{(1)}_1p_1-\expAo  = -q_1 \expAo \,,
\end{equation}
hence every contribution to \eqref{eqn:A:explanation} contains at least one projection $q$ onto the orthogonal complement of the condensate wave function. This gives a prefactor $N^{-1/2}$ because
\begin{equation}
\norm{q_1\UNp^*\bPhi}_\fHN = N^{-\frac12}\norm{\d\Gamma_\perp(q)^\frac12\UNp^*\bPhi}_\fHN
= N^{-\frac12}\norm{\Np^\frac12\bPhi}_\FNp\,.
\end{equation}
To control the action of $A^{(1)}$ on $\UNp^*\bPhi$, note that $A^{(1)}$ is relatively bounded by $\hH$ by assumption, and, for any $\psi_N\in\fH^N_\sym$,
\begin{eqnarray}
\norm{\hH_1\psi_N}_\fHN^2 &=& N^{-1}\sum\limits_{j=1}^N\lr{\psi_N,\hH_j\hH_j\psi_N}_\fHN\nonumber\\
&\leq& N^{-1}\sum\limits_{1\leq j,\l\leq N}\lr{\psi_N,\hH_j\hH_\l\psi_N}_\fHN
\;=\; N^{-1}\norm{\boldKz\psi_N}^2_\fHN
\end{eqnarray}
by permutation symmetry of $\psi_N$ and as $\hH\geq0$.
The full argument is given in Section \ref{subsec:proofs:A}.\medskip

\noindent\textbf{Step 3.}
Proposition \ref{prop:exp:P} implies that
\begin{equation}\label{eqn:proof:thm:exp:P}
\Tr_\Fp\FockA\Pn-\sum\limits_{\l=0}^a \lN^\frac{\l}{2}\Tr_\Fp \FockA\Pnl 
= \lN^{\frac{a+1}{2}}\left(\Tr_\Fp\FockA\FockBPn(a) + \Tr_\Fp\FockA\FockBQn(a)\right)\,,
\end{equation}
with $\FockBPn$ and $\FockBQn$ as defined in \eqref{B:P} and \eqref{B:Q}.
Let us sketch the estimate of the remainders for the leading order $a=0$ and the simplest case of a non-degenerate eigenvalue of $\FockHz$ (and thus $\FockH$). In this case, 
\begin{subequations}
\begin{eqnarray}
\Tr_\Fp\FockA\FockBQn(0)&=&
\frac{1}{2\pi\i}\Tr_\Fp\goint\frac{1}{z-\Ezn}\frac{\Qn}{z-\tFockH}\FockRz\Pzn\FockA\dz\,,\\
\Tr_\Fp\FockA\FockBPn(0)&=&
\frac{1}{2\pi\i}\Tr_\Fp\goint\frac{\Pn}{z-\En}\FockRz\ResHzz\FockA\dz\,,
\end{eqnarray}
\end{subequations}
both of which contain at least one rank-one projection. By construction, the circumference of  $\gan$  as well as its distance to   $\En$ and $\Ezn$ are of order one. Hence, after interchanging trace and integral, it remains to control
\begin{subequations}
\begin{eqnarray}
\left|\lr{\Chizn,\FockA\frac{\Qn}{z-\FockHminus}\FockRz\Chizn}\right|_\Fp
&\leq& \norm{\FockA\Chizn}_\Fp\Big\|\frac{\Qn}{z-\FockHminus}\Big\|_\op\norm{\FockRz\Chizn}_\Fp\,,\label{eqn:explanation:BQ}\\
\left|\lr{\Chin,\FockRz\ResHzz\FockA\Chin}\right|_\Fp
&\leq&\norm{\Chin}_\Fp\Big\|\FockA\ResHzz\FockRz\Chin\Big\|_\Fp\label{eqn:explanation:BP}
\end{eqnarray}
\end{subequations}
for $z\in\gan$. To estimate these expressions, recall that $\FockRz$ is constructed out of the operators $\mathbb{K}_j$ from \eqref{eqn:K:notation} and the Taylor remainders in Lemma \ref{lem:calculus}. By \eqref{eqn:ass:v:v*phi^2} and \eqref{eqn:HS:norm:K}, $\boldKo$ to $\boldKth$ are bounded by powers of $(\Np+1)$. Concerning $\boldKf$, note that it can be written as
\begin{equation}
\boldKf=\d\Gamma_\perp(v)+\d\Gamma_\perp\left(v*\varphi^2\otimes\id+\id\otimes v*\varphi^2+\id\otimes\id \lr{\varphi,v*\varphi^2\varphi}\right)\,.
\end{equation}
Whereas the second term can be controlled by powers of $(\Np+1)$, this is not true for $\d\Gamma_\perp(v)$ since $v$ may be unbounded. However, due to \eqref{eqn:ass:v:2:Delta:bound}, it can be bounded in terms of $\boldKz^{1/2}$ and $(\Np+1)$ (Lemma \ref{lem:K:commutators}).
In summary, we find (see Lemma \ref{lem:K:a:norms:estimate:remainders}) that
\begin{equation}
\norm{\FockRz\Chizn}_\Fp \leq \fC\left(\norm{(\Np+1)^2\Chizn}+\norm{(\Np+1)^\frac32\FockHz\Chizn}\right)\leq\fC(n)
\end{equation}
since $\norm{(\Np+1)^\frac32\FockHz\Chizn}\leq\fC\norm{(\Np+1)^\frac32\Chizn}$ and because finite moments of $\Np$ with respect to $\Chizn$ are bounded uniformly in $N$ (Lemma \ref{lem:moments:Chiz}). Analogously, \eqref{eqn:A} yields
\begin{equation}
\eqref{eqn:explanation:BQ}\leq \fC(n) N^\alpha\,,
\end{equation}
with $\alpha=-1/2$ for $\FockA=\FockAmred$ and $\alpha=0$ for $\FockA=\id$. Moreover,
\begin{equation}
\eqref{eqn:explanation:BP}\leq \fC N^\alpha\Big\|(\Np+1)\ResHzz\FockRz\Chin\Big\|_\Fp\leq \fC N^\alpha\norm{\FockRz\Chin}_\Fp\,.\label{eqn:explanation:Chin}
\end{equation} 
The last inequality, which is proven in Lemma \ref{lem:aux}, follows essentially from the observation that $\Np\leq \fC\,\BogUz(\FockHz-\Ezz+1)\BogUz^*$, for $\BogUz$ the Bogoliubov transformation diagonalizing $\FockHz$ (Lemma \ref{lem:Np:ls:FockHz}), because one can control the action of $\BogUz$ on the number operator (Lemma~\ref{lem:number:BT}) sufficiently well. 
As opposed to \eqref{eqn:explanation:BQ}, we do not \textit{a priori} know this to be of order $N^\alpha$, since we do not have sufficient control of $(\Np+1)^b\Chin$ for $b>1/2$ and of $\FockR_0\Chin$, which contains a contribution $\boldKf\Chin$. 
\\

\noindent\textbf{Step 4.}
To prove a uniform bound for $\Tr_\Fp(\Np+1)^b\Pn$ for any $b\geq1$, we make use of the \textit{a priori} bound
\begin{equation}\begin{split}\label{eqn:explanation:a:priori}
\Tr_\Fp(\Np+1)\Pn & \leq\fC(n)\,,\\
\Tr_\Fp(\Np+1)^b\Pn & \leq \fC(b,n) N^\frac13 \Tr_\Fp(\Np+1)^{b-1}\Pn
\end{split}\end{equation} 
(Lemma \ref{lem:moments:Chi}) to close a bootstrap argument.  
Let us explain the strategy for  the simplest case $b=2$ and a non-degenerate eigenvalue $\Ezn$.
First, we expand $\Pn$ one step around $\Pzn$, i.e.,  we apply \eqref{eqn:proof:thm:exp:P} to $\FockA=(\Np+1)^2$ for $a=0$. Since  $\Tr_\Fp(\Np+1)^2\Pzn$ is bounded uniformly in $N$, it remains to show that the error terms corresponding to \eqref{eqn:explanation:BQ} and \eqref{eqn:explanation:BP} are bounded. Whereas \eqref{eqn:explanation:BQ} is clearly bounded uniformly in $N$, we make use of the above \textit{a priori} bound to estimate \eqref{eqn:explanation:BP}. The positive powers of $N$ arising from this can be compensated for by the prefactor $\lN^{1/2}$ in \eqref{eqn:proof:thm:exp:P}, which, however, requires some manipulations since we do not yet have a sufficient bound for $\boldKf\Chin$. 
This cancellation is precisely the point where the restriction $\varepsilon(N)\leq CN^\frac13$ in Assumption \ref{ass:cond} enters. The full argument is given in Lemma \ref{lem:bootstrap}. Note that for the $d$-dimensional torus, a uniform bound for $\Tr_\Fp(\Np+1)^b$ was shown in \cite[Corollary~3.2]{mitrouskas_PhD} by a different argument.

Finally, the estimate $\Tr_\Fp\boldKf^2\Pn\leq\fC$ follows from a similar bootstrap argument, using the \textit{a priori} bound 
\begin{equation}
\boldKf\leq\fC\left((\Np+1)^\frac32\d\Gamma_\perp(\hH)(\Np+1)^\frac32 + (\Np+1)^4\right)
\end{equation}
together with Assumption \ref{ass:cond} and the previous estimate of $\Tr_\Fp(\Np+1)^b\Pn$.
\medskip

\begin{remark}\label{rem:v:bd}
For interactions $v\in L^\infty(\R^d)$, Step 4 is not necessary.
In this case, Assumption~\ref{ass:cond} holds with $\varepsilon(N)=\mathcal{O}(1)$ \cite[Lemma 1]{grech2013}, hence the \textit{a priori} bound \eqref{eqn:explanation:a:priori} is already uniform in $N$ (see Lemma \ref{lem:moments:Chi}), and, moreover, $\boldKf$ is bounded by powers of $\Np$. 

The latter also explains why the estimate of the growth of $\fC(n,m,a)$ in $a$ is better than for generic $v$ (Remark \ref{rem:growth:of:const}): since all operators $\FockHj$ and $\FockR_j$ from the expansion of $\FockHminus$ are bounded by powers of $\Np$ (and not by $\FockHz$), each commuting with a resolvent $(z-\FockHz)^{-1}$ cancels one of these powers as in \eqref{eqn:explanation:Chin}. Consequently, the final power of $\Np$ acting on $\Chin$ and $\Chizn$ is less than in the generic case, where this effect is cancelled by $\FockHz$ hitting the resolvent. 
Since conjugating powers of $\Np$ with Bogoliubov transformations is the main source for the growth in $a$ (see Lemma \ref{lem:number:BT}), this leads to a better estimate.
\end{remark}

\section{Bogoliubov theory}\label{sec:Bog_theory}
In this section, we summarize some known results concerning the Bogoliubov Hamiltonian $\FockHz$ and its connection to the $N$-body Hamiltonian $\HN$. 
As a preparation,  recall that
\begin{equation} \label{eqn:aux:1}
\ad_x F(\Np)= F(\Np-1)\ad_x \,,\qquad a_x F(\Np)= F(\Np+1)a_x 
\end{equation} 
for any function $F$.
Moreover,  normal ordered expressions can be bounded in terms of $\Np$:
\begin{lem}\label{lem:aux:new}
Let $n,p\geq 0$, $f:\fHp^p\to\fHp^n$ a bounded operator with (Schwartz) kernel $f(x^{(n)};y^{(p)})$, and $\bPhi\in\Fp$. Then
\begin{equation}
\left\|\int\dx^{(n)}\dy^{(p)}f(x^{(n)};y^{(p)})\ad_{x_1}\,\mycdots\,\ad_{x_n}\,a_{y_{1}}\,\mycdots\,a_{y_{p}}\bPhi\right\|_\Fp 
\leq \norm{f}_{\fHp^{p}\to\fHp^n}\norm{(\Np+n)^\frac{n+p}{2}\bPhi}_\Fp\,.
\end{equation}
\end{lem}

A proof is given in \cite[Lemma 5.1]{QF}.
In the following, we will always assume that Assumptions~\ref{ass:V}, \ref{ass:v} and \ref{ass:cond} are satisfied.

\subsection{Bogoliubov transformations}
\label{subsec:pre:Bogoliubov}
We begin with briefly recalling the concept of Bogoliubov transformations, mainly following \cite{solovej_lec,QF}.
Let us consider
\begin{equation}
F=f\oplus Jg=f\oplus \overline{g}=\begin{pmatrix}f \\ \overline{g} \end{pmatrix} \in\fHp\oplus\fHp\,,
\end{equation}
where  $J:\fHp\to\fHp$,  $(Jf)(x)=\overline{f(x)}$, denotes complex conjugation, and define the generalized creation and annihilation operators $A(F)$ and $\Ad(F)$ as
\begin{equation}\label{eqn:A(F)}
A(F)=a(f)+\ad(g)\,, \quad \Ad(F)=A(\cJ F)=\ad(f)+a(g)
\end{equation}
for
$ \cJ=\begin{pmatrix}0 & J\\J&0\end{pmatrix}$. 
An operator $\BogV$ on $\fHp\oplus\fHp$  such that  $F\mapsto A(\BogV F)$ has the same properties as  $F\mapsto A(F)$, i.e., 
\begin{equation}
\Ad(\BogV F)=A(\BogV\mathcal{J}F)\,,\qquad [A(\BogV F_1),\Ad(\BogV F_2)]=[A(F_1),\Ad(F_2)]\,,
\end{equation}
is called a \emph{(bosonic) Bogoliubov map}.

\begin{definition}
A bounded operator $\BogV:\fHp\oplus \fHp\to \fHp\oplus \fHp$
is a Bogoliubov map if 
\begin{equation}
\BogV^*\mathcal{S}\BogV=\mathcal{S}=\BogV\mathcal{S}\BogV^*\,,\qquad \mathcal{J}\BogV\mathcal{J}=\BogV
\end{equation}
for $\cS=\begin{pmatrix}
1 & 0 \\ 0 & -1
\end{pmatrix}$.
Equivalently, $\BogV$ has the block form
\begin{equation}\label{BogV:block:form}
\BogV:=\begin{pmatrix}U & \Vbar\\V & \Ubar\end{pmatrix}\,,\quad U,V:\fHp\to \fHp\,,
\end{equation}
where $U$ and $V$ satisfy the relations
\begin{equation}\label{eqn:rel:U:V}
U^*U=\id+V^*V\,, \qquad UU^*=\id+\Vbar\,\Vbar^*\,, \quad V^*\Ubar=U^*\Vbar\,,\quad UV^*=\Vbar\,\Ubar^*\,.
\end{equation}
We denote the set of Bogoliubov maps on $\fHp\oplus\fHp$ as
\begin{equation}
\fV(\fHp):=\left\{\BogV\in\mathcal{L}\left(\fHp\oplus\fHp\right)\,|\,\BogV \text{ is a Bogoliubov map }\right\}.
\end{equation}
\end{definition}
The adjoint and inverse of  $\BogV\in\fV(\fHp)$  with block form \eqref{BogV:block:form} are given as
\begin{equation}\label{eqn:BogV:-1}
\BogV^*=\begin{pmatrix}
U^* & V^*\\\Vbar^{\,*} & \Ubar^{\,*} \end{pmatrix}\,,\quad
\BogV^{-1}=\cS\BogV^*\cS=\begin{pmatrix}
U^* & -V^*\\ -\Vbar^{\,*} & \Ubar^{\,*} \end{pmatrix}\,.
\end{equation}
Under certain conditions, Bogoliubov maps can be unitarily implemented on $\Fp$ (see, e.g., \cite[Theorem 9.5]{solovej_lec}):
\begin{lem}
Let $\BogV\in\fV(\fHp)$. Then there exists a unitary transformation $\BogU:\Fp\to\Fp$ such that
\begin{equation}\label{eqn:unit:impl}
\BogU A(F)\BogU^*=A(\BogV F)
\end{equation}
for all  $F\in\fHp\oplus\fHp$
if and only if 
\begin{equation}
 \norm{V}_{\HS(\fHp)}^2 :=\Tr_{\fHp}(V^*V)<\infty
\end{equation}
(Shale--Stinespring condition).
In this case,  $\BogV$ is called (unitarily) im\-ple\-men\-ta\-ble. We refer to the unitary implementation of a Bogoliubov map as Bogoliubov transformation.
\end{lem}
If $V$ is Hilbert--Schmidt, the map
$\BogV\mapsto \BogU$
is a group homomorphism, which, in particular, implies that
\begin{equation}
\mathcal{U}_{\BogV^{-1}}=\left(\BogU\right)^{-1}=\BogU^*\,.
\end{equation}
Writing $U$, $V$  as integral operators with (Schwartz) kernels $U(x;y)$ and $V(x;y)$, i.e.,
\begin{equation}
(Uf)(x)=\int U(x;y)f(y)\dy\,, \qquad (Vf)(x)=\int V(x;y)f(y)\dy 
\end{equation} 
for any $f\in \fHp$, we can express the transformation rule \eqref{eqn:unit:impl} as
\begin{equation}\begin{split}
\BogU \,a_x\,\BogU^*&\;=\;\int \dy\, \overline{U(y;x)}\,a_y
+\int \dy\, \overline{V(y;x)}\,\ad_y\,,\\
\BogU\,\ad_x\,\BogU^*&\;=\;\int \dy\, V(y;x)\,a_y + \int\dy\, U(y;x)\,\ad_y\,.\label{eqn:trafo:ax}
\end{split}\end{equation}
In particular, powers of $\Np$ conjugated with $\BogU$ can be bound as follows (see \cite[Lemma 4.4]{QF} for a proof):
\begin{lem}\label{lem:number:BT}
Let $\BogV\in\mathfrak{V}(\fHp)$ be unitarily implementable and denote by $\BogU$ the corresponding Bogoliubov transformation on $\Fp$. Then it holds for any $b\in\N$ that
\begin{equation*}
\BogU(\Np+1)^b\BogU^* \leq C_\BogV^b\, b^b(\Np+1)^b
\end{equation*}
in the sense of operators on $\Fp$, where
\begin{equation}\label{CV}
C_\BogV:=2\norm{V}_\HS^2+\onorm{U}^2+1
\end{equation}
for $\BogV=\begin{pmatrix}U & \Vbar\\ V&\Ubar\end{pmatrix}$ and with $\onorm{\cdot}:=\norm{\cdot}_{\cL(\fHp)}$ and $\norm{\cdot}_\HS:=\norm{\cdot}_{\HS(\fHp)}$.
\end{lem}

Finally, we recall the notion of quasi-free states.
\begin{definition}\label{def:QF}
A normalized state $\bPhi\in\Fp$ is called a quasi-free (pure) state if  there exists some $\BogV\in\fV(\fHp)$ such that 
\begin{equation}
\bPhi=\BogU|\Omega\rangle\,.
\end{equation}
\end{definition}
Alternatively, quasi-free states can be defined via  Wick's rule (e.g.\ \cite[Theorem 1.6]{nam2011}):

\begin{lem}\label{lem:wick}
Let $\bPhi\in\Fp$ be normalized. Then $\bPhi$ is quasi-free if and only if 
\begin{equation}
\lr{\bPhi,\Number\bPhi}_\Fp<\infty
\end{equation}
and
\begin{subequations}\label{eqn:Wick}
\begin{align}
\lr{\bPhi,a^\sharp(f_1)\mycdots a^\sharp(f_{2n-1})\bPhi}_{\Fp} &= 0\,,\\
\lr{\bPhi,a^\sharp(f_1)\mycdots a^\sharp(f_{2n})\bPhi}_{\Fp} &= \sum\limits_{\sigma\in P_{2n}}\prod\limits_{j=1}^n \lr{\bPhi,a^\sharp(f_{\sigma(2j-1)})a^\sharp(f_{\sigma(2j)})\bPhi}_{\Fp}
\end{align}
\end{subequations}
for $a^\sharp\in\{\ad,a\}$, $n\in\N$ and $f_1\mydots f_{2n}\in \fHp$.
Here, $P_{2n}$ denotes the set of pairings
\begin{equation}
P_{2n}:=\{\sigma\in\mathfrak{S}_{2n}:\sigma(2a-1)<\min\{\sigma(2a),\sigma(2a+1)\} \;\forall a\in \{1,2\mydots 2n\} \}\,,
\end{equation}
where $\mathfrak{S}_{2n}$ denotes the symmetric group on the set $\{1,2\mydots 2n\}$.
\end{lem}

\subsection{Properties of $\FockH$ and $\FockHz$}\label{subsec:pre:FockHz}

Since $\FockHz$ is a quadratic Hamiltonian, it can be diagonalized by Bogoliubov transformations, which makes it possible to compute its spectrum:
\begin{lem} \label{lem:known:results:FockHz}
\lemit{
\item \label{lem:BT:diag}
There exists a unitarily implementable Bogoliubov map 
$$\BogVz
=\begin{pmatrix} U_0 & \Vbar_0 \\ V_0 & \Ubar_0 \end{pmatrix}
\in\mathfrak{V}(\fHp)$$ 
such that the corresponding Bogoliubov transformation $\BogUz:\Fp\to\Fp$ diagonalizes $\FockHz$, i.e., there exists a self-adjoint operator $D>0$ on $\fHp$ such that 
\begin{equation}\label{eqn:BT:diag}
\BogUz\FockHz\BogUz^*=\d\Gamma_\perp(D)+\inf\sigma(\FockHz)\,.
\end{equation} 
The spectrum of $D$  is purely discrete and we denote its eigenvalues as
\begin{equation}
0<d^{(0)}<d^{(1)}<\dots<d^{(j)}<\dots\,.
\end{equation}
In particular, $D$ admits a complete set of normalized eigenfunctions, denoted as $\{\xi_j\}_{j\geq0}$.
\item \label{lem:known:results:FockHz:3}
The spectrum of $\FockHz$ is purely discrete, and the ground state energy of $\FockHz$ is negative.
For any $n\in\N$, there exists some $k\in\N_0$ and some tuple $(\nu_0\mydots\nu_{k})\in\N_0^{k+1}$ such that
\begin{equation}\label{eqn:eigenvalues:FockHz}
\Ezn=\Ezz+\nu_0\,d^{(0)}+\nu_1d^{(1)}+\dots+\nu_{k}d^{(k)}\,.
\end{equation}
Further,  $\fgn>0$, for $\fgn$ as in \eqref{eqn:fgn}.
\item \label{lem:known:results:FockHz:2}
The ground state of $\FockHz$ is unique and given by
\begin{equation}
\Chizz=\BogUz^*\vac\,.
\end{equation}
For each $n\in\N$, there exists a basis $\big\{\Chiznm\big\}_{1\leq m\leq \dzn}$ of $\fEzn$ such that 
\begin{equation}\label{eqn:excited:states:FockHz}
\Chiznm=\BogUz^*
\frac{\big(\ad(\xi_0)\big)^{\nu_0}}{\sqrt{\nu_0!}}
\frac{\big(\ad(\xi_1)\big)^{\nu_1}}{\sqrt{\nu_1!}}
\,\mycdots\,
\frac{\big(\ad(\xi_k)\big)^{\nu_{k}}}{\sqrt{\nu_{k}!}}
\vac
\end{equation}
for some $k\in\N_0$ and some tuple $(\nu_0\mydots \nu_{k})\in\N_0^{k+1}$  depending on $m$.
\item \label{lem:moments:Chiz}
Let $b\in\N_0$ and let  $\Chiznm\in\fEzn$ be given by \eqref{eqn:excited:states:FockHz}. Then
\begin{eqnarray}\label{eqn:moments:Chiz:nm}
\lr{\Chiznm,(\Np+1)^b\Chiznm}_\Fp\leq  \left(\fC\, b\, (1+\nu_0+\dots+\nu_{k})\right)^b \leq (\fC(n) b)^b\,,
\end{eqnarray} 
and
\begin{equation}\label{eqn:moments:Chiz}
\norm{(\Np+1)^b\Pzn}_{\cL(\Fp)}\leq (\fC(n) b)^b\,.
\end{equation}
\item \label{lem:Np:ls:FockHz}
In the sense of operators on $\Fp$, it holds that
\begin{equation}
\Np+1\;\leq\;\fC\,\BogUz^*(\Np+1)\BogUz \;\leq \;\fC\left(\FockHz-\Ezz+1\right)\,.
\end{equation}
}
\end{lem}

All statements of Lemma \ref{lem:known:results:FockHz} are well known and are proven for various models in, e.g., \cite{solovej_lec,nam_PhD,lewin2015_2,nam2016,napiorkowski2017}.  In the following, we summarize a proof for our model.

\begin{proof}
\noindent\textbf{Part (a).}
Let us abbreviate $\tilde{K}:=qKq$ for $K$ as in \eqref{def:K:kernel}. By Lemma~\ref{lem:hH},
$\tilde{K}(\hH+\tilde{K})^{-1}$ is Hilbert--Schmidt on $\fHp$ since
\begin{equation}
\norm{\tilde{K}(\hH+\tilde{K})^{-1}}_{\HS}\leq \norm{K}_{\HS}\onorm{(\hH+\tilde{K})^{-1}}\leq \gapH^{-1}\norm{K}_\HS
\end{equation}
as $K\geq 0$ and $\hH\geq\gapH>0$ on $\fHp$. Moreover,  $G:=(\hH+\tilde{K})^{-\frac12}\tilde{K}(\hH+\tilde{K})^{-\frac12}$ is Hilbert--Schmidt on $\fHp$ since
\begin{eqnarray}
\Tr\,(G^*G)&=&\Tr\left(\left(\tilde{K}(\hH+\tilde{K})^{-1}\right)^2\right)\leq \norm{\tilde{K}(\hH+\tilde{K})^{-1}}^2_\HS\,,
\end{eqnarray}
and  $\onorm{G}=\onorm{\tilde{K}^\frac12(\hH+\tilde{K})^{-1}\tilde{K}^{\frac12}}<1$ because
\begin{equation}
\tilde{K}^\frac12(\hH+\tilde{K})^{-1}\tilde{K}^\frac12 \;\leq\;\frac{\tilde{K}}{\gapH+\tilde{K}} \;\leq\;\frac{\onorm{\tilde{K}}}{\gapH+\onorm{\tilde{K}}}\id\,,
\end{equation}
where we used that the inverse is operator monotone and that $x\mapsto x(\gapH+x)^{-1}$ is increasing.
Hence, by \cite[Theorems 1 and 2]{nam2016}, there exists a unitarily implementable $\BogVz\in\mathfrak{V}(\fHp)$ such that 
\begin{equation}
\BogV_0\mathcal{A}\BogV_0^*
=\BogVz\begin{pmatrix} \hH+\tilde{K} & \tilde{K} \\ \tilde{K}& \hH+\tilde{K} \end{pmatrix}\BogVz^* =\begin{pmatrix}D & 0 \\ 0 & J D J \end{pmatrix}
\end{equation}
for some self-adjoint operator $D>0$ on $\fHp$, and
\begin{equation}
\BogUz\FockHz\BogUz^*=\d\Gamma_\perp(D)+\inf\sigma(\FockHz)\,,
\end{equation}
where $\BogUz$ denotes the unitary implementation of $\BogVz$ on $\Fp$. 
Finally, one can show as in step~(6) in the proof of \cite[Theorem A.1]{lewin2015_2} that $D$ has purely discrete spectrum.\medskip

\noindent\textbf{Parts (b) and (c).}
By \cite[Theorem A.1(iii-iv)]{lewin2015_2},  $\sigma(\FockHz)=\sigma_\mathrm{disc}(\FockHz)$ and $\inf\sigma(\FockHz)<0$.
Since $D>0$, $\vac$ is the unique ground state of $\d\Gamma_\perp(D)$ with eigenvalue zero, hence $\BogUz^*\vac$ is the unique ground state of $\FockHz$ with eigenvalue $\Ezz=\inf\sigma(\FockHz)$ by \eqref{eqn:BT:diag}.
By part (a), there is a complete set of normalized eigenstates $\{\xi_j\}_{j\geq0}$ for $D$, hence
\begin{equation}
\d\Gamma_\perp(D)=\sum\limits_{j\geq0}\dj\ad(\xi_j)a(\xi_j)\,.
\end{equation}
Consequently, all eigenstates of $\d\Gamma_\perp(D)$ can be written as
\begin{equation}
\frac{\big(\ad(\xi_0)\big)^{\nu_0}}{\sqrt{\nu_0!}}\,\mycdots\, \frac{\big(\ad(\xi_k)\big)^{\nu_k}}{\sqrt{\nu_k!}}\vac
\end{equation}
for some $k\in\N_0$, and all eigenvalues of $\d\Gamma_\perp(D)$ are of the form
\begin{equation}
\nu_0 d^{(0)}+\nu_1d^{(1)}+\dots+\nu_{k}d^{(k)}\,
\end{equation}
for some $k\in\N_0$ and $(\nu_0\mydots \nu_{k})\in\N_0^{k+1}$.
Finally, \eqref{eqn:excited:states:FockHz} and \eqref{eqn:eigenvalues:FockHz} follow from \eqref{eqn:BT:diag}.
\medskip

\noindent\textbf{Part (d).}
For $\Chiznm$  as in \eqref{eqn:excited:states:FockHz}, we compute by Lemma \ref{lem:number:BT} that
\begin{eqnarray}
&&\hspace{-1cm}\lr{\Chiznm,(\Np+1)^b\Chiznm}_\Fp \nonumber\\
&=&
\Big\|(\Np+1)^\frac{b}{2}\BogUz^*\frac{\big(\ad(\xi_0)\big)^{\nu_0}}{\sqrt{\nu_0!}}
\,\mycdots\,
\frac{\big(\ad(\xi_k)\big)^{\nu_{k}}}{\sqrt{\nu_{k}!}}
\vac\Big\|^2_\Fp \nonumber\\
&\leq&b^b\,\CVz^b\Big\|(\Np+1)^\frac{b}{2}\frac{\big(\ad(\xi_0)\big)^{\nu_0}}{\sqrt{\nu_0!}}
\,\mycdots\,
\frac{\big(\ad(\xi_k)\big)^{\nu_k}}{\sqrt{\nu_k!}}
\vac\Big\|^2_\Fp \,,
\end{eqnarray}
where $\CVz$ denotes the constant from Lemma \ref{lem:number:BT} for $\BogV=\BogVz$.
This proves \eqref{eqn:moments:Chiz:nm} because
\begin{eqnarray}
&&(\Np+1)^\frac{b}{2}\frac{\big(\ad(\xi_0)\big)^{\nu_0}}{\sqrt{\nu_0!}}
\,\mycdots\,
\frac{\big(\ad(\xi_k)\big)^{\nu_k}}{\sqrt{\nu_k!}}
\vac\nonumber\\
&&\quad=
(\nu_0+\dots+\nu_k+1)^\frac{b}{2}\frac{\big(\ad(\xi_0)\big)^{\nu_0}}{\sqrt{\nu_0!}}
\,\mycdots\,
\frac{\big(\ad(\xi_k)\big)^{\nu_k}}{\sqrt{\nu_k!}}\vac\,,
\end{eqnarray}
and \eqref{eqn:moments:Chiz} follows from the decomposition 
$ \Pzn=\sum_{m=1}^{\dzn}|\Chiznm\rangle\langle\Chiznm|$.
\medskip

\noindent\textbf{Part (e).}
This follows from parts (a) and (c) and by Lemma \ref{lem:number:BT} since
\begin{eqnarray}
\lr{\bPhi,(\FockHz-\Ezz)\bPhi}_\Fp 
&=&\lr{\BogUz\bPhi,\sum\limits_{j\geq0}d^{(j)}\ad(\xi_j)a(\xi_j)\BogUz\bPhi}_\Fp\nonumber \\
& \geq& \gzz\lr{\bPhi,\BogUz^*\Np\BogUz\bPhi}_\Fp \,.
\end{eqnarray}
\end{proof}

Next, we recall that for excitation energies of order one, the eigenvalues of $\FockHN$ converge to  eigenvalues of $\FockHz$ as $N\to\infty$. Statements of this kind were proven in \cite{seiringer2011,grech2013,lewin2015_2,mitrouskas_PhD}.

\begin{lem}
\label{lem:known:results:FockHN}
\lemit{
\item \label{lem:known:results:FockHN:E}
For any $\nu\in\N_0$ and $\Enu$ as in Definition \ref{def:FockHN}, there exists some $n\in\N_0$ such that
\begin{equation}\label{eqn:known:results:FockHN:E}
\lim\limits_{N\to\infty}\Enu=\Ezn\,.
\end{equation}
\item \label{lem:known:results:FockHN:geq:Np}
In the sense of operators on $\FNp$, 
\begin{equation}\label{eqn:Np:bd:by:HN}
\Np+1\;\leq\;  \fC\left(\FockHN+ N^\frac13\right)\,.
\end{equation}
\item \label{lem:moments:Chi}
Let $\Chin\in\fEn$ for $n\in\N_0$. Then 
\begin{equation}\label{eqn:moments:Chi:1}
\lr{\Chin,(\Np+1)\Chin}_\Fp \leq \fC(n)\,,
\end{equation}
and
\begin{equation}\label{eqn:moments:Chi:a:priori}
\lr{\Chin,(\Np+1)^b\Chin}_\Fp
\leq \fC(b,n) N^\frac{\l}{3}\lr{\Chin,(\Np+1)^{b-\l}\Chin}_\Fp
\end{equation}
for $b\in\N_0$ and any $0\leq \l\leq  b$.
If $\varepsilon(N)=\mathcal{O}(1)$ in Assumption~\ref{ass:cond}, one obtains the improved bound
\begin{equation}\label{eqn:moments:Chi:improved}
\lr{\Chin,(\Np+1)^b\Chin}_\Fp
\leq \left(\fC(n)+3^\frac{b}{2}\right)^b\,.
\end{equation}
}
\end{lem}

\begin{proof}
By Lemma \ref{lem:hH} and Assumption \ref{ass:cond}, all assumptions (A1), (A2) and (A3s) in \cite{lewin2015_2} are satisfied, hence part (a) follows from \cite[Theorem 2.2(ii)]{lewin2015_2}.
\medskip

\noindent \textbf{Part (b).}
By Assumption \ref{ass:cond}, there exist constants $C_1\geq0$ and $0<C_2\leq1$ such that, for sufficiently large $N$,
\begin{eqnarray}
\HN-N\eH\geq C_2\d\Gamma_\perp(\hH)-C_1N^\frac13
\end{eqnarray}
in the sense of operators on $\fHN$.
Since $\varphi$ is the unique ground state of $\hH$ with eigenvalue zero, it follows that
\begin{eqnarray}
\d\Gamma_\perp(\hH)=\sum\limits_{j\geq0}\varepsilon^{(j)}\ad(\varphi_j)a(\varphi_j)
=\sum\limits_{j\geq1}\varepsilon^{(j)}\ad(\varphi_j)a(\varphi_j)
&\geq \gapH\Np
\end{eqnarray}
on $\fHN$, where $0<\varepsilon^{(1)}\leq\varepsilon^{(2)}\leq\dots$ .
Consequently, it holds for $\bPhi\in\FNp$ that
\begin{eqnarray}
\lr{\bPhi,\Np\bPhi}_{\FNp}
&=&\lr{\UNp^*\bPhi,\Np\UNp^*\bPhi}_{\fHN}\nonumber\\
&\leq&\frac{1}{C_2\gapH}\lr{\bPhi,\left(\FockHN+C_1N^\frac13\right)\bPhi}_{\FNp}\,.
\end{eqnarray}

\noindent\textbf{Part (c).}
By Lemma \ref{lem:hH} and Assumption \ref{ass:cond}, \cite[Theorem 2.2(iv)]{lewin2015_2} implies that there exists some $\Chizn\in\fEzn$ such that, up to a subsequence,
\begin{equation}
\lim\limits_{N\to\infty}\norm{\Chin-\Chizn}_\Fp =0\,,\qquad
\lim\limits_{N\to\infty}\lr{(\Chin-\Chizn),\FockHz(\Chin-\Chizn)}_\Fp =0\,,
\end{equation}
hence, by Lemma \ref{lem:Np:ls:FockHz},
\begin{eqnarray}
&&\hspace{-1.5cm}\lr{\Chin,(\Np+1)\Chin}_\Fp  \nonumber\\
&\leq&\fC\lr{(\Chin-\Chizn),(\FockHz-\Ezz+1)(\Chin-\Chizn)}_\Fp \nonumber\\
&& + \fC\lr{\Chizn,(\FockHz-\Ezz+1)\Chizn}_\Fp\nonumber\\
&& +2\fC\norm{\Chin-\Chizn}_\Fp\norm{(\FockHz-\Ezz+1)\Chizn}_\Fp \nonumber\\
&\leq&\fC(\Ezn-\Ezz+1)
\end{eqnarray}
for sufficiently large $N$. 
Further, part (b) implies that
\begin{eqnarray}
&&\hspace{-1.5cm}\lr{\Chin,(\Np+1)^{b+1}\Chin} _\Fp\nonumber\\
&=&\lr{(\Np+1)^\frac{b}{2}\Chin,(\Np+1)(\Np+1)^\frac{b}{2}\Chin}_{\FNp}\nonumber\\
&\leq&\fC\lr{(\Np+1)^\frac{b}{2}\Chin,(\FockHN+N^\frac13)(\Np+1)^\frac{b}{2}\Chin}_{\FNp}\nonumber\\
&\leq&\fC\lr{(\Np+1)^b\Chin,(\FockHN+N^\frac13)\ChiNn}_{\FNp}\nonumber\\
&&+\fC\norm{(\Np+1)^\frac{b}{2}\Chin}_\FNp\big\|\big[\FockHN,(\Np+1)^\frac{b}{2}\big]\Chin\big\|_{\FNp}\nonumber\\
&\leq&\fC\left(|\Ezn|+ N^\frac13+ 3^\frac{b}{2}\right)\lr{\Chin,(\Np+1)^b\Chin} _\Fp
\end{eqnarray}
by Lemma \ref{lem:prelim:commutators} and since $\Chin\in\fEn$. Iterating over $b$ concludes the proof.
\end{proof}

\section{Proofs}
In the remainder of the paper, we abbreviate
$$\norm{\cdot}_\Fp\equiv\norm{\cdot}\,,\qquad \lr{\cdot,\cdot}_\Fp\equiv\lr{\cdot,\cdot}\,,
\qquad\norm{\cdot}_{\cL(\Fp)}\equiv\onorm{\cdot}\,,\qquad \Tr_\Fp\equiv\Tr\,.$$
We will always assume that Assumptions \ref{ass:V}, \ref{ass:v} and \ref{ass:cond} are satisfied.

\subsection{Asymptotic expansion of $\Pn$}
\subsubsection{Proof of Lemma \ref{lem:expansion:resolvent}}
\label{subsec:proofs:resolvent}
Recall that $\FockH=\FockHminus+\FockHplus$ by \eqref{FockHplusminus}, hence
\begin{eqnarray}\label{eqn:lem:exp:resolvent:0}
\ResHz=\RestHz(z-\FockH+\FockHplus)\ResHz
=\RestHz+\RestHz\FockHplus\ResHz\,.
\end{eqnarray}
Next, we prove by induction over $a\in\N_0$ that
\begin{equation}\label{eqn:lem:resolvent:induction:hyp}
\RestHz=\ResHzz\sum\limits_{\l=0}^a\lN^\frac{\l}{2}\,\FockT_\l (z)
+ \lN^\frac{a+1}{2}\RestHz\sum\limits_{\nu=0}^{a} \FockRnu\ResHzz\FockT_{a-\nu}(z)
\end{equation}
where
\begin{equation}
\FockT_\l(z)=\sum\limits_{\nu=1}^\l\FockHnu\ResHzz\FockT_{a-\nu}(z)\,,\qquad \FockT_0(z)=\id\,.
\end{equation}
\textbf{Base case.}
Proposition \ref{lem:expansion:HN} implies that
\begin{equation}\label{eqn:lem:exp:resolvent:1}
\tFockH=\FockHz+\lN^\frac12\FockRz\,,
\end{equation}
hence 
\begin{eqnarray}
\RestHz
&=& \RestHz\left(z-\FockHminus+\lN^\frac12\FockRz\right)\ResHzz\nonumber\\
&=&\ResHzz+\lN^\frac12\RestHz\FockRz\ResHzz\,.\label{eqn:lem:exp:resolvent:2}
\end{eqnarray}
\textbf{Induction step.} Assume \eqref{eqn:lem:resolvent:induction:hyp} holds for $a-1\in\N$. 
Since 
\begin{equation}
\tFockH
=\sum\limits_{j=0}^{\nu}\lN^\frac{j}{2}\FockHj+\lN^\frac{\nu+1}{2}\FockR_{\nu}
=\sum\limits_{j=0}^{\nu}\lN^\frac{j}{2}\FockHj
+\lN^{\frac{\nu+1}{2}}\FockH_{\nu+1}+\lN^\frac{\nu+2}{2}\FockR_{\nu+1}\,,
\end{equation}
it follows that
\begin{equation}\label{eqn:lem:exp:resolvent:3}
\FockR_{\nu}=\FockH_{\nu+1}+\lN^\frac12\FockR_{\nu+1}\,,
\end{equation}
hence we conclude with  \eqref{eqn:lem:exp:resolvent:2} and by the induction hypothesis that
\begin{eqnarray}
\RestHz
&=&\ResHzz\sum\limits_{\l=0}^{a-1}\lN^\frac{\l}{2}\FockT_\l(z) \nonumber\\
&&+\lN^\frac{a}{2}\RestHz\sum\limits_{\nu=0}^{a-1}\left(\FockH_{\nu+1}+\lN^\frac12\FockR_{\nu+1}\right)\ResHzz\FockT_{a-\nu-1}(z)\nonumber\\
&=&\ResHzz\sum\limits_{\l=0}^{a-1}\lN^\frac{\l}{2}\FockT_\l(z) 
+\lN^\frac{a}{2}\ResHzz\sum\limits_{\nu=0}^{a-1}\FockH_{\nu+1}\ResHzz\FockT_{a-\nu-1}(z)\nonumber\\
&&+\lN^\frac{a+1}{2}\RestHz\FockR_0\ResHzz\sum\limits_{\nu=0}^{a-1}\FockH_{\nu+1}\ResHzz\FockT_{a-\nu-1}(z)\nonumber\\
&&+\lN^\frac{a+1}{2}\RestHz\sum\limits_{\nu=0}^{a-1}\FockR_{\nu+1}\ResHzz\FockT_{a-\nu-1}(z)\nonumber\\
&=&\ResHzz\sum\limits_{\l=0}^{a}\lN^\frac{\l}{2}\FockT_\l(z)\nonumber\\
&&+\lN^\frac{a+1}{2}\RestHz\Big(\FockRz\ResHzz\FockT_a+\sum\limits_{\nu=0}^{a-1}\FockR_{\nu+1}\ResHzz\FockT_{a-\nu-1}(z)\Big)\nonumber\\
&=&\ResHzz\sum\limits_{\l=0}^{a}\lN^\frac{\l}{2}\FockT_\l(z)
+\lN^\frac{a+1}{2}\RestHz\sum\limits_{\nu=0}^{a}\FockRnu\ResHzz\FockT_{a-\nu}(z)\,,
\end{eqnarray}
which concludes the induction.
Finally, we rewrite  $\FockT_a(z)$ as
\begin{eqnarray}
\FockT_a(z)
&=&\sum\limits_{j_1=1}^a\FockH_{j_1}\ResHzz\FockT_{a-j_1}(z)\nonumber\\
&=&\sum\limits_{j_1=1}^a\sum\limits_{j_2=1}^{a-j_1}\FockH_{j_1}\ResHzz\FockH_{j_2}\ResHzz\FockT_{a-(j_1+j_2)}(z)\nonumber\\
&=&\sum\limits_{\nu=1}^a \,\sum\limits_{\substack{\bj\in\N^\nu\\|\bj|=a}}\FockH_{j_1}\ResHzz\mycdots\FockH_{j_\nu}\ResHzz\FockT_0(z)\,,
\end{eqnarray}
which concludes the proof.
\qed

\subsubsection{Proof of Proposition \ref{prop:exp:P}}
Let $n\in\N_0$. The expansion of the resolvent from Lemma \ref{lem:expansion:resolvent} yields
\begin{eqnarray}
\Pn&=&\Pzn+
\sum\limits_{\l=1}^a\lN^\frac{\l}{2}\sum\limits_{\nu=1}^\l\sum\limits_{\substack{\bj\in\N^\nu\\[2pt]|\bj|=\l}} \FockAnjnu 
\;+\;\lN^\frac{a+1}{2}\sum\limits_{\nu=0}^{a}\sum\limits_{m=1}^{a-\nu}\sum\limits_{\substack{\bj\in\N^m\\[2pt]|\bj|=a-\nu}}\FockBnjm\;+\;\FockCn\,,\quad\label{eqn:proof:prop:exp:P:1a}
\end{eqnarray}
where
\begin{eqnarray}
\FockAnjnu&:=&\frac{1}{2\pi\i}\goint \ResHzz\FockHjo \ResHzz\FockHjt \ResHzz\mycdots \FockHjnu \ResHzz\dz\,,\label{eqn:Tjn}\\
\FockBnjm&:=&\frac{1}{2\pi\i}\goint \RestHz\FockRnu\ResHzz\FockHjo \ResHzz\mycdots \FockH_{j_m} \ResHzz\dz\,,\qquad\label{eqn:Bjn}\\
\FockCn&:=&\frac{1}{2\pi\i}\goint\RestHz\,\FockHplus\ResHz\dz\,.\label{eqn:FockC}
\end{eqnarray}
\medskip

\noindent\textbf{Computation of $\FockAnjnu$}.
We decompose $1=\Pzn+\Qzn$ in each term in \eqref{eqn:Tjn} and sort according to the number of projections $\Qzn$, which takes the values $k=0\mydots \nu+1$. This yields
\begin{eqnarray}
\FockAnjnu\hspace{-2pt}
&=&\hspace{-2pt}\sum\limits_{k=0}^{\nu+1} \,\sum\limits_{\substack{\bm\in\{0,1\}^{\nu+1}\\|\bm|=k}}
\frac{1}{2\pi\i}\goint\frac{1}{(z-\Ezn)^{\nu+1-k}}
\tFockOn_{m_1}(z)\FockH_{j_1}\mycdots
\tFockOn_{m_\nu}(z)\FockHjnu\tFockOn_{m_{\nu+1}}(z)\dz\nonumber\\
&=:&\hspace{-2pt}\sum\limits_{k=0}^{\nu+1}\tFockAnkjnu\,,
\end{eqnarray}
where we  abbreviated
\begin{equation}\label{eqn:tFockOno}
\tFockOnz(z):=\Pzn\,,\qquad
\tFockOno(z):=\frac{\Qzn}{z-\FockHz}\,.
\end{equation}
Observe first that the contributions  with exclusively $\Pzn$ ($k=\nu+1$) or exclusively $\Qzn$ ($k=0$) vanish: in case of only $\Pzn$, 
\begin{equation}
\tilde{\FockA}^{(n)}_{0,\,\bj} =\frac{1}{2\pi\i}\left(\goint \frac{1}{(z-\Ezn)^{\nu+1}}\dz\right)\Pzn\FockHjo \Pzn\mycdots \Pzn\FockHjnu \Pzn 
=0\,,
\end{equation}
and in case of only $\Qzn$, the integrand is holomorphic in the area enclosed by $\gamma^{(n)}$, hence $\tilde{\FockA}^{(n)}_{\nu+1,\,\bj}=0$.

For $1\leq k\leq \nu$, the integrand in $\tFockAnkjnu$ has a pole of order $\nu+1-k$ at $z=\Ezn$, hence
the residue theorem implies that
\begin{eqnarray}
\tFockAnkjnu
&=&\hspace{-0.5cm}\sum\limits_{\substack{\bm\in\{0,1\}^{\nu+1}\\|\bm|=k}}\frac{1}{(\nu-k)!}\lim\limits_{z\to\Ezn}\frac{\d^{\nu-k}}{\dz^{\nu-k}}\left(\tFockOn_{m_1}(z)\FockH_{j_1}\mycdots
\tFockOn_{m_\nu}(z)\FockHjnu\tFockOn_{m_{\nu+1}}(z)\right)\,.\qquad
\end{eqnarray}
Let us consider the case where $m_j=1$ for $j=1\mydots k$ and $m_j=0$ for $j=k+1\mydots \nu+1$. By the Leibniz rule and since
\begin{equation}
\left.\frac{\d^m}{\dz^m}\tFockOno(z) \right|_{z=\Ezn}= (-1)^m m!\,\FockOn_{m+1}
\end{equation}
with $\FockOnk$  and  $\FockOnz$ as defined in \eqref{FockO}, i.e.,
\begin{equation}
\FockOnz=-\Pzn\,,\qquad \FockOnk=\frac{\Qzn}{(\Ezn-\FockHz)^k}\,,
\end{equation}
we obtain for this case
\begin{eqnarray}
&&\hspace{-1.5cm}\frac{1}{(\nu-k)!}\lim\limits_{z\to\Ezn}\frac{\d^{\nu-k}}{\dz^{\nu-k}}\left(\tFockOno(z)\FockHjo\mycdots
\tFockOno(z)\right)\FockH_{j_k}
\Pzn\FockH_{j_{k+1}}\mycdots\Pzn\FockHjnu\Pzn\nonumber\\
&=&\frac{1}{(\nu-k)!}\sum\limits_{\substack{\bm\in\N_0^k\\|\bm|=\nu-k}}\binom{\nu-k}{\bm}
\left(\left.\frac{\d^{m_1}}{\dz^{m_1}}\tFockOno(z)\right|_{z=\Ezn}\right)\FockHjo
\mycdots\nonumber\\
&&\mycdots \left(\left.\frac{\d^{m_{k}}}{\dz^{m_{k}}}\tFockOno(z)\right|_{z=\Ezn}\right)\FockH_{j_k}\Pzn\FockH_{j_{k+1}}\mycdots\Pzn\FockHjnu\Pzn\nonumber\\
&=&-\sum\limits_{\substack{\bm\in\N_0^k\\|\bm|=\nu-k}}\FockOn_{m_1+1}\FockHjo\mycdots \FockOn_{m_{k}+1}\FockH_{j_k}\FockOnz\FockH_{j_{k+1}}\mycdots\FockOnz\FockHjnu\FockOnz\nonumber\\
&=&-\sum\limits_{\substack{\bm\in\N^k\\|\bm|=\nu}}\FockOn_{m_1}\FockHjo\mycdots \FockOn_{m_{k}}\FockH_{j_k}\FockOnz\FockH_{j_{k+1}}\mycdots\FockOnz\FockHjnu\FockOnz\,.\label{eqn:Tjn:ones:zeros}
\end{eqnarray}
The other contributions to $\tFockAnkjnu$ are related to \eqref{eqn:Tjn:ones:zeros} through permutations, hence
\begin{eqnarray}
\tFockAnkjnu = -\sum\limits_{\substack{\bm\in\N^k\times\{0\}^{\nu-k+1}\\|\bm|=\nu}}
\FockOn_{m_1}\FockHjo \mycdots \FockOn_{m_{\nu}}\FockHjnu\FockOn_{m_{\nu+1}}
\end{eqnarray}
and consequently
\begin{eqnarray}
\FockAnjnu
=\sum\limits_{k=1}^\nu\tFockAnkjnu
=-\sum\limits_{\substack{\bm\in\N_0^{\nu+1}\\|\bm|=\nu}}\FockOn_{m_1}\FockHjo \mycdots \FockOn_{m_{\nu }}\FockHjnu\FockOn_{m_{\nu+1}}\,.
\end{eqnarray}

\noindent\textbf{Computation of $\FockBnjm$.}
Decomposing the first identity in \eqref{eqn:Bjn} as $1=\Pn+\Qn$, one notes that the term with $\Pn$  yields $\FockBPn$. For the term with $\Qn$, we decompose in each resolvent of $\FockHz$ the identity as $1=\Pzn+\Qzn $. Note that the term containing exclusively $\Qn$ and $\Qzn$ vanishes since the integrand has no poles in the area enclosed by $\gamma^{(n)}$.\medskip

\noindent\textbf{Computation of $\FockCn$.}
Recall that $\Pn$ projects onto a subspace of $\FNp\oplus0$, hence
\begin{equation}
\Pn\FockHplus=\FockHplus\Pn=0\,.
\end{equation}
Consequently, decomposing both identities in $\FockCn$ yields
\begin{eqnarray}
\FockCn = \frac{1}{2\pi\i}\goint\frac{\Qn}{z-\tFockH}\FockHplus\frac{\Qn}{z-\FockH}\dz =0
\end{eqnarray}
since the integrand is holomorphic in the area enclosed by $\gamma^{(n)}$.

\subsection{Auxiliary estimates}
\subsubsection{Preliminaries}
In this section, we collect some preliminary estimates.
First, we provide  bounds for second-quantized $m$-body operators; subsequently, we estimate $\mathbb{K}_j$, $\FockHj$ and $\FockRj$ as well as commutators of $\Np$ with $\FockHN$ and $\FockH$.

\begin{lem}\label{lem:prelim:On}
Let $m\in\N$ and let $\Om$ be an operator on $\fH^m$. Assume that there exist constants $c_1,c_2\geq0$ such that
\begin{equation}
\norm{\Om\psi}_{\fH^m}^2\leq c_1\Big\|\sum\limits_{j=1}^m T_j\psi\Big\|^2_{\fH^m}+c_2\norm{\psi}^2_{\fH^m}
\end{equation}
for any $\psi\in\D(\sum_{j=1}^m T_j)$ and with $T$ as in \eqref{def:T}.
\lemit{
\item \label{lem:prelim:On:1}
Let $\psi\in\fH^m$. Then
\begin{equation}
\norm{\Om\psi}^2_{\fH^m}\leq 2c_1\Big\|\sum\limits_{j=1}^m\hH_j\psi\Big\|_{\fH^m}^2 +2c_3\norm{\psi}^2_{\fH^m}\,,
\end{equation}
where $c_3=\fC c_1m^2+\frac{c_2}{2}$.
\item \label{lem:prelim:On:new}
Let $k\geq m$ and $\psi\in\fH^k_\sym$. Then
\begin{equation}
\bigg\|\sum\limits_{j=1}^m \hH_j\psi\bigg\|^2_{\fH^k} \leq \frac{m}{k}\bigg\|\sum\limits_{j=1}^k\hH_j\psi\bigg\|^2_{\fH^k}\,.
\end{equation}
\item \label{lem:prelim:On:2}
Let $k\geq m$. Then it follows for $\psi_k\in\fH^k_\sym$ that
\begin{equation}
\bigg\|\sum\limits_{1\leq j_1<\dots<j_m\leq k}\Omj\psi_k\bigg\|^2_{\fH^k}
\leq \binom{k}{m}^2\left(\frac{2c_1m}{k}\Big\|\sum\limits_{j=1}^k\hH_j\psi_k\Big\|^2_{\fH^k}+2c_3\norm{\psi_k}^2_{\fH^k}\right)\,.
\end{equation}
}
\end{lem}

\begin{proof}
Part (a) follows since $\hH_j=T_j+(v*\varphi^2)(x_j)-\mH$ and by \eqref{eqn:ass:v:v*phi^2} because
\begin{eqnarray}
\norm{\Om\psi}^2_{\fH^m}
&\leq &c_1\bigg(\Big\|\sum\limits_{j=1}^m\hH_j\psi\Big\|_{\fH^m}+\Big\|\sum\limits_{j=1}^m\left(v*\varphi^2(x_j)-\mH\right)\psi\Big\|_{\fH^m}\bigg)^2+c_2\norm{\psi}^2_{\fH^m}\nonumber\\
&\leq&2c_1\Big\|\sum\limits_{j=1}^m\hH_j\psi\Big\|^2+(2c_1\fC m^2+c_2)\norm{\psi}^2_{\fH^m}\,.
\end{eqnarray}
For part (b),  the permutation symmetry of $\psi$ leads to the estimate
\begin{eqnarray}
\Big\|\sum\limits_{j=1}^m\hH_j\psi\Big\|^2_{\fH^k}
&=& m\lr{\psi, h_1h_1\psi}_{\fH^k} + m(m-1)\lr{\psi,h_1 h_2\psi}_{\fH^k}
\nonumber\\
&=&\frac{m}{k}\sum\limits_{j=1}^k\lr{\psi,\hH_j\hH_j\psi}_{\fH^k}
+ \frac{m(m-1)}{k(k-1)}\sum\limits_{\substack{1\leq j,\l\leq k\\\l\neq j}}\lr{\psi,\hH_j\hH_\l\psi}_{\fH^k}\nonumber\\
&\leq&\frac{m}{k}\sum\limits_{1\leq j,\l\leq k}\lr{\psi,\hH_j\hH_\l\psi}_{\fH^k}
\end{eqnarray}
since $m\leq k$ and $\hH\geq 0$.
For part (c), we obtain with parts (a) and (b)
\begin{equation}
\norm{\Omom\psi_k}^2_{\fH^k}
\;\leq\;\frac{2c_1 m}{k} \Big\|\sum\limits_{j=1}^k\hH_j\psi_k\Big\|^2_{\fH^k}+2c_3\norm{\psi_k}^2_{\fH^k}\,,
\end{equation}
which proves the claim since
\begin{equation*}
\bigg\|\sum\limits_{1\leq j_1<\dots<j_m\leq k}\Omj\psi_k\bigg\|^2_{\fH^k}
\;\leq\;\left(\sum\limits_{1\leq j_1<\dots<j_m\leq k}\Big\|\Omj\psi_k\Big\|\right)^2\,. \qedhere
\end{equation*}
\end{proof}

In the next lemma, we collect bounds for the operators $\boldKo$ to $\boldKf$ from \eqref{eqn:K:notation}.
\begin{lem}\label{lem:K:commutators}
Let $\bPhi\in\Fp$.
\lemit{
\item \label{lem:K}
For $\mathbb{K}_j^{(*)}\in\{\mathbb{K}_j\,,\,\mathbb{K}_j^*\}$,
\begin{subequations}
\begin{eqnarray}
\norm{\boldKo\bPhi}  
&\leq& \fC\norm{(\Np+1)\bPhi} \,,\\[5pt]
\norm{\boldKt^{(*)}\bPhi} 
&\leq&\fC\norm{(\Np+1)\bPhi} \,,\\[5pt]
\norm{\boldKth^{(*)}\bPhi} 
&\leq& \fC\norm{(\Np+1)^\frac{3}{2}\bPhi} \,, \\[5pt]
\norm{\boldKf\bPhi} 
&\leq&\fC\left( \norm{(\Np+1)^2\bPhi} +\norm{\boldKz^\frac12(\Np+1)^\frac32\bPhi} \right)\label{eqn:K4:bound:K0^1/2}\\
&\leq&\fC\left( \norm{(\Np+1)^2\bPhi} +\norm{\FockHz(\Np+1)^\frac32\bPhi} \right)\,.\label{eqn:K4:bound:H0}
\end{eqnarray}
\end{subequations}
\item \label{lem:prelim:commutators}
Let $\l\geq 0$. Then
\begin{subequations}
\begin{eqnarray}
\Big\|\left[\FockHN,(\Np+1)^\l\right]\bPhi\Big\|_\FNp &\leq& 3^\l\,\fC\,\norm{(\Np+1)^\l\bPhi}_\FNp \,,\\
\Big\|\left[\FockHz,(\Np+1)^\l \right]\bPhi\Big\|  &\leq&3^\l\l\,\fC\,\norm{(\Np+1)^\l\bPhi} \,.
\end{eqnarray}
\end{subequations}
}
\end{lem}

\begin{proof}
Since $\norm{K}_{\fH\to\fH} \leq \norm{K}_\HS \leq \fC$ by \eqref{eqn:HS:norm:K} and as \eqref{eqn:ass:v:2:Delta:bound} and \eqref{eqn:ass:v:v*phi^2} imply that
\begin{equation}
\norm{\Kth\psi}_{\fHp^2} \leq \norm{v(x_1-x_2)\varphi(x_1) \psi(x_2)}_{\fHp^2} + \norm{(v*\varphi^2)(x_1)\varphi(x_1)\psi(x_2)}_{\fHp^2} 
\leq \fC \norm{\psi}
\end{equation}
for any $\psi\in\fHp$, the bounds for $\boldKo$, $\boldKt^{(*)}$ and $\boldKth^{(*)}$ follow from Lemma \ref{lem:aux:new}. Finally, note that
\begin{equation}
\boldKf=\d\Gamma_\perp(v) - \d\Gamma_\perp(\tilde{\Kf})\,,
\end{equation}
where $\tilde{\Kf}$ denotes the multiplication operator on $\fHp\otimes\fHp$ corresponding to
\begin{equation}
\tilde{\Kf}(x_1,x_2):=(v*\varphi^2)(x_1)+(v*\varphi^2)(x_2)-\lr{\varphi,v*\varphi^2\varphi}\,.
\end{equation}
As above, \eqref{eqn:ass:v:v*phi^2} and Lemma \ref{lem:aux:new} imply that
$
\norm{\d\Gamma_\perp(\tilde{\Kf})\bPhi} \leq \fC\norm{(\Np+1)^2\bPhi} \,.
$
Moreover,
\begin{eqnarray}
\lr{\psi,|v(x_1-x_2)|^2\psi}_{\fH^k}
&\leq&\fC\left(\norm{\psi}^2_{\fH^k}+\lr{\psi,\hH_1\psi}_{\fH^k}+\norm{v*\varphi^2-\mH}_\infty\norm{\psi}^2_{\fH^k}\right)\nonumber\\
&\leq&\fC\Big(\norm{\psi}^2_{\fH^k}+\frac{1}{k}\Big\langle\psi,\sum\limits_{j=1}^k\hH_j\psi\Big\rangle_{\fH^k}\Big)
\end{eqnarray}
for $\psi\in\fH^k$ by  \eqref{eqn:ass:v:2:T:bound}  and \eqref{eqn:ass:v:v*phi^2}, hence it follows from Lemmas \ref{lem:aux:new} and \ref{lem:prelim:On:2} that
\begin{eqnarray}
\norm{\d\Gamma_\perp(v)\bPhi}^2
&\leq& \sum\limits_{k\geq 0}\Big\|\sum\limits_{1\leq i<j\leq k}v(x_i-x_j)\phi^{(k)}\Big\|_{\fHp^k}^2 
\nonumber\\
&\leq&\fC\left(\sum\limits_{k\geq0} k(k-1)^2 \lr{\phi^{(k)},\boldKz\phi^{(k)}}_{\fH^k} 
+\norm{(\Np+1)^{2}\bPhi}^2 \right)\nonumber\\
&\leq&\fC\left(\norm{\boldKz^\frac12(\Np+1)^{\frac32}\bPhi} ^2+ \norm{(\Np+1)^{2}\bPhi}^2 \right)\,,
\end{eqnarray}
where we used that $\d\Gamma_\perp(\hH)=\boldKz$.
Moreover, $\boldKz^\frac12\leq 1+\boldKz=1+\FockHz-\boldKo-\boldKt-\boldKtbar$ implies 
\begin{eqnarray}
\norm{\boldKz^\frac12(\Np+1)^{\frac32}\bPhi}^2
&\leq&\norm{(\FockHz+1)(\Np+1)^{\frac32}\bPhi}^2\nonumber\\
&&+\left|\lr{(\Np+1)^\frac32\bPhi,\boldKo(\Np+1)^\frac32\bPhi}\right|\nonumber\\
&&+2\left|\lr{(\Np+1)^\frac32\bPhi,\boldKt(\Np+1)^\frac32\bPhi}\right| \nonumber\\
&\leq&\fC\big(\norm{\FockHz(\Np+1)^{\frac32}\bPhi}  + \norm{(\Np+1)^{2}\bPhi} \big)^2\,,
\end{eqnarray}
where we used that $|\lr{\bPhi,\mathbb{K}_j\bPhi}|\leq \fC\norm{(\Np+1)^\frac12\bPhi}^2$ for $j=1,2$ by  \eqref{eqn:HS:norm:K}.
\medskip

\noindent\textbf{Part (b).}
Since $[\boldKz,\Np]=[\boldKo,\Np]=[\boldKf,\Np]=0$, \eqref{eqn:FockHN} implies that
\begin{eqnarray}
[\FockHN,(\Np+1)^\l]
&=&\left[\boldKt,(\Np+1)^\l\right]g_{\Np}+g_{\Np}\left[\boldKtbar,(\Np+1)^\l\right]\nonumber\\
&&+\left[\boldKth,(\Np+1)^\l\right]\tilde{g}_{\Np}+\tilde{g}_{\Np}\left[\boldKthbar,(\Np+1)^\l\right]\,,\label{eqn:proof:moments:3}
\end{eqnarray}
where $g_{\Np}:=\frac{\sqrt{[(N-\Np)(N-\Np-1)]_+}}{N-1}$ and $\tilde{g}_{\Np}:=\frac{\sqrt{[N-\Np]_+}}{N-1}$. For $N\geq2$,
\begin{eqnarray}
\norm{g_{\Np}\bPhi}_{\FNp}&\leq&2\norm{\bPhi}_{\FNp}\,,\qquad 
\norm{\tilde{g}_{\Np}\bPhi}_{\FNp}\leq 3(N+1)^{-\frac12}\norm{\bPhi}_{\FNp}\,.\label{eqn:proof:moments:6}
\end{eqnarray}
By \eqref{eqn:aux:1}, we find that
\begin{equation}
[\boldKt,(\Np+1)^\l]=-\boldKt\left((\Np+3)^\l-(\Np+1)^\l\right)\,,\label{eqn:proof:moments:4}
\end{equation}
and analogously for $\boldKtbar$, $\boldKth$ and $\boldKthbar$. 
Since it holds for $a,k\geq0$ and $c\geq1$ that
\begin{equation}\label{eqn:proof:moments:5}
(k+a)^c-k^c\leq c\, a(k+a)^{c-1}\leq c\, a^c(k+1)^{c-1}\,,
\end{equation}
we conclude with part (a)  that 
\begin{eqnarray}
\big\|[\boldKt,(\Np+1)^\l]g_{\Np}\bPhi\big\|_{\FNp}
&\leq&\fC\, \norm{\big((\Np+3)^{\l+1}-(\Np+1)^{\l+1}\big)g_{\Np}\bPhi}_{\FNp}\nonumber\\
&\leq& \l\,3^{\l}\fC\,\norm{(\Np+1)^\l\bPhi}_{\FNp}\,,\\
\big\|[\boldKth,(\Np+1)^\l]\tilde{g}_{\Np}\bPhi\big\|_{\FNp}
&\leq&3^\l\l\,\fC\bigg\|\left(\frac{\Np+1}{N+1}\right)^\frac12(\Np+1)^\l\bPhi\bigg\|_{\FNp}\nonumber\\
&\leq& 3^\l\l\,\fC\norm{(\Np+1)^\l\bPhi}_{\FNp}\,, 
\end{eqnarray}
and similarly for $\boldKtbar$ and $\boldKthbar$. The proof for $\FockHz$ works analogously.
\end{proof}

Next, we observe that the operators $\FockHj$ and $\FockRj$ can be bounded in terms of $\Np$ and $\FockHz$, which follows immediately from  Lemma \ref{lem:K}.

\begin{lem}\label{lem:K:a:norms:estimate:remainders}
Let $\bPhi\in\Fp$ and $b\geq0$.
\lemit{
\item \label{lem:H:norms}
For any $j\in\N$, it holds that
\begin{eqnarray}
&&\hspace{-0.5cm}\norm{(\Np+1)^b\FockHj\bPhi} \nonumber \\
&&\leq
\fC(b+j)\left( \big\|(\Np+1)^{b+\frac{j}{2}+1}\bPhi\big\| 
+\norm{(\Np+1)^b\FockHz(\Np+1)^\frac32\bPhi}\right)\,.\qquad\quad
\end{eqnarray}

\item \label{lem:estimate:remainders}
Further,
\begin{subequations}
\begin{eqnarray}
\norm{(\Np+1)^b\FockRza\bPhi}  &\leq& \fC(b)\Big(\norm{(\Np+1)^{b+\frac32}\bPhi}\nonumber\\
&&\qquad +\lN^\frac12\norm{(\Np+1)^{b+2}\bPhi} \Big)\,,\qquad\\[5pt]
\norm{(\Np+1)^b\FockRz\bPhi} &\leq& \fC(b)\Big(\norm{(\Np+1)^{b+\frac32}\bPhi} +\lN^\frac12\norm{(\Np+1)^{b+2}\bPhi} \nonumber\\
&& \qquad +\lN^\frac12 \norm{(\Np+1)^b\FockHz(\Np+1)^\frac32\bPhi} \Big)\,,\\[3pt]
\norm{(\Np+1)^b\FockRo\bPhi}  &\leq&
\fC(b)\Big(\norm{(\Np+1)^{b+2}\bPhi}  + \lN^\frac12\norm{(\Np+1)^{b+\frac52}\bPhi}  
\nonumber\\ 
&&\qquad+ \norm{(\Np+1)^b\FockHz(\Np+1)^\frac32\bPhi} \Big)\,,
\end{eqnarray}
\end{subequations} 
and, for any $j\in\N_0$,
\begin{eqnarray}
\norm{(\Np+1)^b\FockRj\bPhi} 
&\leq&  \fC(b,j) \Big(\norm{(\Np+1)^{b+\frac{j+3}{2}}\bPhi} 
+\lN^\frac12\norm{(\Np+1)^{b+\frac{j+4}{2}}\bPhi} \nonumber\\ 
&&\qquad\qquad+\norm{(\Np+1)^b\FockHz(\Np+1)^\frac32\bPhi} \Big)\,.
\end{eqnarray}
}
\end{lem}

\subsubsection{Bound for $\FockAmred$}
\label{subsec:proofs:A}
In this section, we show that $\FockAmred $ as defined in \eqref{Amred} is bounded in terms of $\FockHz$ and $\Np$.

\begin{lem}\label{lem:A}
For $\Am$ as in \eqref{eqn:thm:A:unbounded} and the corresponding operator $\FockAmred$ as in \eqref{Amred}, it holds that
\begin{equation}
\norm{\FockAmred\bPhi}  
\leq\fC(m) N^{-\frac12}\left( \norm{(\Np+1)\bPhi} +\big\|\FockHz\bPhi\big\| \right)\,.
\end{equation}
\end{lem}

\begin{proof}
In the following, we abbreviate 
$$\psi_N:=\UNp^*\bPhi\,.$$
Decomposing $1=p_{j_1}\mycdots p_{j_m}+(1-p_{j_1}\mycdots p_{j_m})$ and observing that 
\begin{equation}
p_{j_1}\mycdots p_{j_m}\Amj p_{j_1}\mycdots p_{j_m} 
= \expAm\, p_{j_1}\mycdots p_{j_m}
\end{equation} 
yields
\begin{subequations}\label{eqn:A2}
\begin{eqnarray}
\norm{\FockAmred\bPhi} \hspace{-1pt}
&\leq&\hspace{-1pt}\tbinom{N}{m}^{-1}\Big\|\expAm\sum\limits_{1\leq j_1<\dots<j_m\leq N}(1-p_{j_1}\mycdots p_{j_m})\psi_N\Big\|_{\fHN}\label{eqn:A2:1}\\
&&\hspace{-1pt}+\tbinom{N}{m}^{-1}\Big\|\sum\limits_{1\leq j_1<\dots<j_m\leq N}\Amj(1-p_{j_1}\mycdots p_{j_m})\psi_N\Big\|_{\fHN}\label{eqn:A2:2}\\
&&\hspace{-1pt}+\tbinom{N}{m}^{-1}\Big\|\sum\limits_{1\leq j_1<\dots<j_m\leq N}(1-p_{j_1}\mycdots p_{j_m})\Amj p_{j_1}\mycdots p_{j_m}\psi_N\Big\|_{\fHN}\,.\qquad\quad\label{eqn:A2:3}
\end{eqnarray}
\end{subequations}
To estimate the contributions in \eqref{eqn:A2}, observe first that
\begin{equation}
\norm{\Am\varphi^{\otimes m}}_{\fH^m} \leq \fC
\end{equation} 
by Lemma \ref{lem:prelim:On:1} because $\hH \varphi=0$. Further, it was shown in  \cite[Lemma 3.2]{corr} that
\begin{equation}
\norm{q_1\mycdots q_\l\psi_N}_\fHN \leq \left(N^{-\l}+2^\l N^{-\l}\lr{\bPhi,\Np^\l\bPhi}_{\FNp}\right)^\frac12
\leq N^{-\frac12} 2^\frac{\l}{2}\norm{(\Np+1)^\frac12\bPhi}_\FNp
\end{equation}
for $\l\in\{1\mydots N\}$ because $\frac{\Np}{N}\leq 1$ as operator on $\FNp$.
Hence, by permutation symmetry of $\psi_N$, it holds that
\begin{eqnarray}
\eqref{eqn:A2:1} 
&\leq&\fC \norm{(1-p_1\mycdots p_m)\psi_N}_{\fHN}\nonumber\\
&\leq&\fC\sum\limits_{\l=1}^m\binom{m}{\l}\norm{q_1\mycdots q_\l\, p_{\l+1}\mycdots p_m\psi_N}_{\fHN}\nonumber\\
&\leq&\fC(m)\, N^{-\frac12} \norm{(\Np+1)^\frac12\bPhi}_\FNp\,. \label{eqn:A2:4}
\end{eqnarray}
For \eqref{eqn:A2:2}, Lemma~\ref{lem:prelim:On:1} implies
\begin{eqnarray}
\eqref{eqn:A2:2}
&\leq&\norm{\Amom(1-p_1\mycdots p_m)\psi_N}_\fHN\nonumber\\
&\leq& \fC\left(\Big\|\sum\limits_{j=1}^m\hH_j(1-p_1\mycdots p_m)\psi_N\Big\|_\fHN +\norm{(1-p_1\mycdots p_m)\psi_N}_\fHN\right)\nonumber\\
&\leq& \fC(m)\left(\Big\|\sum\limits_{j=1}^m  \hH_j \psi_N\Big\|_\fHN + N^{-\frac12}  \norm{(\Np+1)^\frac12\bPhi}_\FNp\right)\label{eqn:A2:6}
\end{eqnarray}
since $\hH_j=q_j\hH_jq_j$ and $q_j(1-p_1\mycdots p_m) = q_j$.
For the first term, Lemma \ref{lem:prelim:On:new} yields
\begin{equation}
\Big\|\sum\limits_{j=1}^m \hH_j \psi_N\Big\|_\fHN
\leq \sqrt{\frac{m}{N}} \Big\|\sum\limits_{j=1}^N\hH_j\psi_N\Big\| \leq \fC(m)N^{-\frac12}\left(\norm{\FockHz\bPhi}_\FNp + \norm{(\Np+1)\bPhi}_\FNp\right)
\end{equation}
because $\hH\varphi=0$ implies that $\boldKz=\d\Gamma_\perp(h)=\UNp^*\boldKz\UNp$
and since
\begin{eqnarray}
\norm{\FockHz\bPhi}_{\Fp} \geq \norm{\boldKz\bPhi} -\fC\,\norm{(\Np+1)\bPhi}_{\Fp}\,\label{eqn:A1:2}
\end{eqnarray}
by Lemma \ref{lem:K}.
Finally, for $m\ll N$,
\begin{subequations}
\begin{eqnarray}
&&\hspace{-1cm}\left(\tbinom{N}{m}\eqref{eqn:A2:3}\right)^2\nonumber\\
&=&\tbinom{N}{m} \sum\limits_{1\leq j_1<\dots<j_m\leq N}\Big\langle\psi_N,p_1\mycdots p_m\Amom(1-p_1\mycdots p_m)(1-p_{j_1}\mycdots p_{j_m})\nonumber\\
&&\hspace{5cm}\times\,\Am_{j_1\mydots j_m}p_{j_1}\mycdots p_{j_m}\psi_N\Big\rangle_\fHN\nonumber\\
&=&\tbinom{N}{m}\sum\limits_{\l=0}^m\tbinom{N-m}{\l}\tbinom{m}{\l}\big\langle\psi_N,p_1\mycdots p_m\Amom(1-p_1\mycdots p_m)(1-p_{\l+1}\mycdots p_{\l+m})\nonumber\\
&&\hspace{5cm}\times \,A^{(m)}_{\l+1\mydots\l+m}\,p_{\l+1}\mycdots p_{\l+m}\,\psi_N\Big\rangle_\fHN\nonumber\\
&\leq&\tbinom{N}{m}\sum\limits_{\l=0}^{m-1}\tbinom{N-m}{\l}\tbinom{m}{\l} \norm{\Amom\varphi^{\otimes m}}^2_{\fH^m}\norm{\bPhi}^2_\FNp\nonumber\\[5pt]
&&+ \tbinom{N}{m}\tbinom{N-m}{m}\Big\langle\psi_N,p_1\mycdots p_m\Amom(1-p_1\mycdots p_m)(1-p_{m+1}\mycdots p_{2m})\nonumber\\
&&\hspace{5cm}\times \,\Am_{m+1\mydots 2m}p_{m+1}\mycdots p_{2m}\psi_N\Big\rangle_\fHN\nonumber\\
&\leq& \fC(m)\tbinom{N}{m}^2 N^{-1} \left(\norm{\bPhi}^2_{\FNp} + \norm{(\Np+1)^\frac12\bPhi}_\FNp^2  \right)\,,
\end{eqnarray}
\end{subequations}
where we used that
\begin{eqnarray}
\tbinom{N}{m} \sum\limits_{\l=0}^{m-1}\tbinom{N-m}{\l}\tbinom{m}{\l}
\leq m\,2^m\tbinom{N}{m}^2\frac{\binom{N-m}{m-1}}{\binom{N}{m}}
\leq \fC(m)\tbinom{N}{m}^2 N^{-1}
\end{eqnarray}
and that
\begin{eqnarray}
&&\hspace{-1cm}\lr{\psi_N,p_1\mycdots p_m\Amom(1-p_1\mycdots p_m)(1-p_{m+1}\mycdots p_{2m})\Am_{m+1\mydots 2m}\,p_{m+1}\mycdots p_{2m}\psi_N}_\fHN\nonumber\\
&=&\Big\langle(1-p_{m+1}\mycdots p_{2m})\psi_N,p_1\mycdots p_m\Amom\Am_{m+1\mydots 2m}\,p_{m+1}\mycdots p_{2m}\nonumber\\
&&\hspace{6cm}\times(1-p_1\mycdots p_m)\psi_N\Big\rangle_\fHN\nonumber\\
&\leq&\norm{\Amom\varphi^{\otimes m}}_{\fH^m}^2\norm{(1-p_1\mycdots p_m)\psi_N}^2_{\fHN}\nonumber\\
&\leq& \fC(m) \left(  N^{-\frac12}\norm{(\Np+1)^\frac12\bPhi}_\FNp\right)^2
\end{eqnarray}
as in \eqref{eqn:A2:4}.
\end{proof}

\subsubsection{Resolvent estimates}
\begin{lem}\label{lem:aux}
Let
$$ \FockIn\in\left\{ \id\,,\,\Pzn\,,\,\Qzn\right\}$$
and   $z\in\gan$.
\lemit{
\item\label{lem:aux:ResHz:op}
It holds that
\begin{subequations}
\begin{eqnarray}
\Big\|\ResInz\Big\|_{\cL(\Fp)}\leq \fC(n)\,, 
\end{eqnarray}
and, for sufficiently large $N$,
\begin{equation}
\Big\|\RestHz\Big\|_{\cL(\Fp)} \leq \fC(n)\,.
\end{equation}
\end{subequations}
\item \label{lem:aux:moments:O}
Let $b\geq 0$. Then
\begin{subequations}
\begin{eqnarray}
\Big\|(\Np+1)^{b+1}\ResInz\bPhi\Big\| 
&\leq& \fC(n,b) \norm{(\Np+1)^{b}\bPhi} \,,\\
\Big\|(\Np+1)^b\FockHz\ResInz\bPhi\Big\|  
&\leq&\fC(n,b)\norm{(\Np+1)^{b}\bPhi}  \,.
\end{eqnarray}
\end{subequations}
}
\end{lem}

\begin{proof}
By definition \eqref{eqn:fgn}  of $\fgn$, it follows that
\begin{eqnarray}
\inf\limits_{\substack{z\in\gan\\\lambda\in\sigma(\FockHz)}}|z-\lambda|
&=&\min\left\{\big|z-\Ezn\big|\,,\,\big|z-\Ez^{(n-1)}\big|\,,\,\big|z-\Ez^{(n+1)}\big|\right\}=\fgn\,,
\end{eqnarray}
which implies the first  part of (a), and the second part follows with Lemma~\ref{lem:known:results:FockHN:E}.
For part (b), recall that there exists a Bogoliubov transformation $\BogUz$   diagonalizing $\FockHz$ (Lemma~\ref{lem:BT:diag}), i.e.,
\begin{equation}\label{eqn:U:diag:Hz}
\left[\BogUz^*(\Np+1)\BogUz,\FockHz\right]=0\,,\qquad
\left[\BogUz \ResInz \BogUz^*\,,\,\Np\right]=0\,.
\end{equation}
As a consequence, Lemma \ref{lem:Np:ls:FockHz} implies that
\begin{equation}
\BogUz^*(\Np+1)^k\BogUz \leq \fC^k \left(\FockHz-\Ezz+1\right)^k\,,
\end{equation}
hence 
\begin{eqnarray}
\ResInzbar\BogUz^*(\Np+1)^{2}\BogUz\ResInz
&\leq & \fC^{2}\left|\ResInz\right|^{2}\left(|\FockHz-z|+|z-\Ezz+1|\right)^{2}\nonumber\\
&\leq& \fC(n)^{2}
\end{eqnarray}
because  $|z-\Ezz+1|\leq |\Ezn|+\fgn + |\Ezz|+1\leq \fC(n)$.
Consequently, Lemma \ref{lem:number:BT} leads for $b\geq 1$ to the estimate
\begin{eqnarray}
\Big\|(\Np+1)^b\ResInz\bPhi\Big\| 
&=&\Big\|(\Np+1)^b\BogUz^*\BogUz\ResInz\BogUz^*\BogUz\bPhi\Big\| \nonumber\\
&\leq&\fC(b)\Big\|(\Np+1)(\Np+1)^{b-1}\BogUz\ResInz\BogUz^*\BogUz\bPhi\Big\| \nonumber\\
&=&\fC(b)\Big\|(\Np+1)\BogUz\ResInz\BogUz^*(\Np+1)^{b-1}\BogUz\bPhi\Big\| \nonumber\\
&\leq&\fC(n,b)\norm{(\Np+1)^{b-1}\BogUz\bPhi} \nonumber\\
&\leq& \fC(n,b)\norm{(\Np+1)^{b-1}\bPhi} \,.
\end{eqnarray}
The second statement of (b) is a consequence of the triangle inequality since
\begin{equation*}
\FockHz\ResInz=-\FockIn+z\ResInz\,. \qedhere
\end{equation*}
\end{proof}

\subsubsection{Bounds for moments of $\Np$ and $\boldKf$ with respect to $\Pn$}

In this section, we show that moments of $\Np$ with respect to both $\Chin$ and $\boldKf\Chin$ are bounded uniformly in $N$.

\begin{lem}\label{lem:bootstrap}
Let $\Chin\in\fEn$ and $b\geq 0$. Then 
\lemit{
\item
\begin{equation}
\lr{\Chin,(\Np+1)^b\Chin}  \leq \fC(n,b)\,, 
\end{equation}
\item 
\begin{equation}
\norm{(\Np+1)^b\boldKf\Chin}  \leq \fC(n,b)\,.
\end{equation}
}
\end{lem}

\begin{proof}
\noindent\textbf{Part (a).}
Proposition \ref{prop:exp:P} with $a=0$ implies that
\begin{subequations}\label{eqn:lem:boot}
\begin{eqnarray}
&&\hspace{-1.5cm}\Tr\left(\Pn(\Np+1)^{b+1}\right)\nonumber\\
&=&\Tr\left(\Pzn(\Np+1)^{b+1}\right)\label{eqn:lem:boot:1}\\
&&+\lN^\frac12\Tr\left(\frac{1}{2\pi\i}\goint\frac{1}{z-\Ezn}\frac{\Qn}{z-\tFockH}\FockRz\Pzn(\Np+1)^{b+1}\dz\right)\qquad\quad\label{eqn:lem:boot:2}\\
&&+\lN^\frac12 \Tr\left(\frac{1}{2\pi\i}\goint\frac{\Pn}{z-\FockHminus}\FockRz\ResHzz(\Np+1)^{b+1}\dz\right).\;\;\;\qquad\label{eqn:lem:boot:3}
\end{eqnarray}
\end{subequations}
For the first term, note that $\eqref{eqn:lem:boot:1}\leq \fC(n,b)$ by Lemma \ref{lem:moments:Chiz}. 
Denoting by $\{\Chiznm\}_{m=1}^{\dzn}$ some orthonormal basis of $\fEzn$ and interchanging trace and integral by Fubini's theorem, we estimate the second term as
\begin{eqnarray}
&&\hspace{-1cm}\left|\eqref{eqn:lem:boot:2}\right|\nonumber\\
&\leq&\lN^\frac12\fgn \sup\limits_{z\in\gan}\left(\left|\frac{1}{z-\Ezn}\right|  \sum\limits_{m=1}^{\dzn}\big\|(\Np+1)^{b+1}\Chiznm\big\| \Big\|\frac{\Qn}{z-\tFockH}\Big\|_{\op}\norm{\FockRz\Chiznm}\right) \nonumber\\
&\leq& N^{-\frac12} \fC(n,b) 
\end{eqnarray}
by Lemmas \ref{lem:moments:Chiz}, \ref{lem:aux:ResHz:op}, \ref{lem:estimate:remainders} and \ref{lem:prelim:commutators} and since $\FockHz\Chiznm=\Ezn\Chiznm$.
Similarly, we find for the last term
\begin{subequations}
\begin{eqnarray}
&&\hspace{-1cm}\left|\eqref{eqn:lem:boot:3}\right|\nonumber\\
&\leq& \lN^\frac12 \fgn \sup\limits_{\substack{z\in\gan\\\nu\in\In}}\left(\left|\frac{1}{z-\Enu}\right| \sum\limits_{m=1}^{\dzn}\left| \lr{\Chinm,\FockRz\ResHzz(\Np+1)^{b+1}\Chinm}\right|\right)\nonumber\\
&\leq&\fC(n) N^{-\frac12}\sum\limits_{m=1}^{\dzn}\left|\lr{\Chinm,\FockRza\ResHzz(\Np+1)^{b+1}\Chinm}\right|\qquad\label{eqn:lem:boot:4}\\
&&+\fC(n)N^{-1}\sum\limits_{m=1}^{\dzn}\left|\lr{\Chinm,\boldKf\ResHzz(\Np+1)^{b+1}\Chinm}\right|\qquad\label{eqn:lem:boot:5}
\end{eqnarray}
\end{subequations}
for $\left\{\Chinm\right\}_{m=1}^{\dzn}$ some orthonormal basis of $\fEn$ and for $\FockRz=\FockRza+\lN^\frac12\boldKf$ as defined in Proposition~\ref{lem:expansion:HN}.
In \eqref{eqn:lem:boot:4}, we obtain the bound
\begin{eqnarray}
&&\hspace{-1.5cm}\left|\lr{\Chinm,(\Np+1)^{b+1}\ResHzz\FockRza\Chinm}\right|\nonumber\\
&\leq& \norm{(\Np+1)^\frac{b}{2}\Chinm}\norm{(\Np+1)^{\frac{b}{2}+1}\ResHzz\FockRza\Chinm}\nonumber\\
&\leq&\fC(n,b)\norm{(\Np+1)^\frac{b}{2}\Chinm}\Big(\norm{(\Np+1)^{\frac{b+3}{2}}\Chinm}\nonumber\\
&&\hspace{4.5cm}+N^{-\frac12}\norm{(\Np+1)^\frac{b+4}{2}\Chinm}\Big)
\end{eqnarray}
by Lemmas \ref{lem:aux:moments:O} and \ref{lem:estimate:remainders}.
Since
\begin{equation}\label{eqn:lem:boot:9}
\norm{(\Np+1)^\frac{b+\l}{2}\Chinm}\leq \fC(n,b+\l) N^\frac{\l}{6}\norm{(\Np+1)^\frac{b}{2}\Chinm}
\end{equation}
for all $\l\in\N_0$ by Lemma \ref{lem:moments:Chi}, it follows that
\begin{equation}
\eqref{eqn:lem:boot:4}\leq \fC(n,b)\norm{(\Np+1)^\frac{b}{2}\Chinm}^2\,.
\end{equation}
Since $[\boldKf,\Np]=0$, the sum in \eqref{eqn:lem:boot:5} can be estimated as
\begin{subequations}
\begin{eqnarray}
&&\hspace{-0.8cm}\left|\lr{\Chinm,\boldKf\ResHzz(\Np+1)^{b+1}\Chinm}\right|\nonumber\\
&\leq& \norm{(\Np+1)^\frac{b}{2}\Chinm}\norm{\boldKf(\Np+1)^{-\frac{b}{2}}\ResHzz(\Np+1)^{b+1}\Chinm}\nonumber\\
&\leq&\fC \norm{(\Np+1)^\frac{b}{2}\Chinm}\Big(
\norm{\FockHz(\Np+1)^{\frac{3-b}{2}}\ResHzz(\Np+1)^{b+1}\Chinm}\nonumber\\
&&\qquad+\norm{(\Np+1)^{2-\frac{b}{2}}\ResHzz(\Np+1)^{b+1}\Chinm}\Big)\nonumber\\
&\leq&N^\frac56 \fC(n,b)\norm{(\Np+1)^\frac{b}{2}\Chinm}^2\label{eqn:lem:boot:6}\\
&&+\fC(n,b) \norm{(\Np+1)^\frac{b}{2}\Chinm}\norm{(\Np+1)^{2-\frac{b}{2}}\ResHzz(\Np+1)^{b+1}\Chinm}\,,\qquad\quad\label{eqn:lem:boot:7}
\end{eqnarray}
\end{subequations}
where we used Lemmas \ref{lem:K} and \ref{lem:aux:moments:O}, \eqref{eqn:lem:boot:9} and that
\begin{eqnarray}
&&\hspace{-0.8cm}\norm{\FockHz(\Np+1)^\frac{3-b}{2}\ResHzz\bPhi}\nonumber\\
&\leq&\norm{(\Np+1)^\frac{3-b}{2}\bPhi}+\Big\|\left[\FockHz,(\Np+1)^\frac{3-b}{2}\right]\ResHzz\bPhi\Big\| \nonumber\\
&&+ |z| \,\norm{(\Np+1)^\frac{3-b}{2}\ResHzz\bPhi}\nonumber\\
&\leq&\fC(n,b)\left(\norm{(\Np+1)^\frac{3-b}{2}\bPhi}+\Big\|(\Np+1)^\frac{3-b}{2}\ResHzz\bPhi\Big\|\right)
\end{eqnarray}
by Lemma \ref{lem:prelim:commutators}.
To control \eqref{eqn:lem:boot:7}, we prove by induction that
\begin{equation}\begin{split}
&\norm{(\Np+1)^{2-\frac{b}{2}}\ResHzz(\Np+1)^{b+1}\bPhi} \\
&\;\leq \fC(n,b)\norm{(\Np+1)^\frac{b}{2}\bPhi}^{1-(\frac{1}{2})^k}\norm{(\Np+1)^{-\frac{b}{2}+ 2^{k+1}}\ResHzz(\Np+1)^{b+1}\bPhi}^{(\frac{1}{2})^k}\label{eqn:lem:boot:8}
\end{split}\end{equation}
for all $k\in\N_0$. The base case $k=0$ is obvious. Now assume that \eqref{eqn:lem:boot:8} holds for some $k\in\N_0$. Then
\begin{eqnarray}
&&\hspace{-1cm}\norm{(\Np+1)^{2-\frac{b}{2}}\ResHzz(\Np+1)^{b+1}\bPhi} \nonumber\\
&\leq&\fC(n,b)\norm{(\Np+1)^\frac{b}{2}\bPhi}^{1-(\frac{1}{2})^k}\bigg\langle(\Np+1)^\frac{b}{2}\bPhi,(\Np+1)^{\frac{b}{2}+1}\ResHzz\nonumber\\
&&\hspace{3.5cm}\times\,(\Np+1)^{-b+2^{k+2}}\ResHzz(\Np+1)^{b+1}\bPhi\bigg\rangle^{(\frac12)^{k+1}}\nonumber\\
&\leq&\fC(n,b)\norm{(\Np+1)^\frac{b}{2}\bPhi}^{1-(\frac{1}{2})^{k+1}}\nonumber\\
&&\qquad\times\,\norm{(\Np+1)^{-\frac{b}{2}+ 2^{k+2}}\ResHzz(\Np+1)^{b+1}\bPhi}^{(\frac12)^{k+1}}
\end{eqnarray}
by Lemma \ref{lem:aux:moments:O}.
Now choose $k$ in \eqref{eqn:lem:boot:8} such that $2^{k+2}\geq b+2$, hence $-\frac{b}{2}+2^{k+1}\geq 1$ and consequently
\begin{eqnarray}
&&\hspace{-1.5cm}\norm{(\Np+1)^{-\frac{b}{2}+2^{k+1}}\ResHzz(\Np+1)^{b+1}\bPhi}\nonumber\\
&\leq& \fC(n,b)\norm{(\Np+1)^{\frac{b+2^{k+2}}{2}}\bPhi}^{(\frac12)^k}\nonumber\\
&\leq&\fC(n,b) N^\frac23\norm{(\Np+1)^\frac{b}{2}\Chinm}^\frac{1}{2^k}
\end{eqnarray}
by Lemma \ref{lem:aux:moments:O} and \eqref{eqn:lem:boot:9}.
In summary, 
\begin{eqnarray}
\Tr\left(\Pn(\Np+1)^{b+1}\right)\leq \fC(n,b) \norm{(\Np+1)^\frac{b}{2}\Chinm}^2 \,.
\label{eqn:lem:boot:10}
\end{eqnarray}
Finally, we prove the lemma via the following bootstrap argument:
\begin{itemize}
\item[(1)] 
Lemma \ref{lem:moments:Chi} implies that 
\begin{equation}
\norm{(\Np+1)^\frac12\Chinm}\leq\fC(n)\xRightarrow{\eqref{eqn:lem:boot:10}}
\Tr\left(\Pn(\Np+1)^2\right)\leq\fC(n)\,.
\end{equation}
\item[(2)] By  step (1), 
\begin{equation}
\norm{(\Np+1)\Chinm}\leq \fC(n)\;\xRightarrow{\eqref{eqn:lem:boot:10}}\;
\Tr\left(\Pn(\Np+1)^3\right)\leq \fC(n)
\end{equation}
\item[(b)] By step $(b-1)$, 
\begin{equation}
\norm{(\Np+1)^\frac{b}{2}\Chinm}\leq \fC(n,b)\;\xRightarrow{\eqref{eqn:lem:boot:10}}\;
\Tr\left(\Pn(\Np+1)^{b+1}\right)\leq \fC(n,b)\,.
\end{equation}
\end{itemize}

\noindent\textbf{Part (b).}
Define
$$\boldKf^-:=\boldKf\big|_{\FNp}\oplus 0\,.$$
By Lemma \ref{lem:K} and Assumption \ref{ass:cond}, there exists a constant $c$ such that
\begin{eqnarray}
&&\hspace{-0.8cm}\norm{(\Np+1)^b\boldKf^-\bPhi} ^2\nonumber\\
&=&\norm{(\Np+1)^b\boldKf\bPhi}_\FNp^2
\nonumber\\
&\leq& \fC\left(\norm{(\Np+1)^{b+2}\bPhi}^2_\FNp + \lr{(\Np+1)^{b+2}\bPhi,\d\Gamma_\perp(\hH)(\Np+1)^{b+2}\bPhi}_\FNp\right)\nonumber\\
&\leq&\fC\bigg(\norm{(\Np+1)^{b+2}\bPhi}^2 \nonumber\\
&&\quad+ \lr{(\Np+1)^{b+2}\bPhi,(\FockHN+cN^\frac13)(\Np+1)^{b+2}\bPhi}_\FNp\bigg)\nonumber\\
&\leq&\fC\bigg(N^\frac16\norm{(\Np+1)^{b+2}\bPhi}\nonumber\\
&&\quad+\left|\lr{\bPhi,(\Np+1)^{b+2}\FockHN(\Np+1)^{b+2}\bPhi}_{\FNp}\right|^\frac12\bigg)^2\,.
\label{eqn:K4:a:priori:1}
\end{eqnarray}
In particular, this implies  that
\begin{eqnarray}
\norm{(\Np+1)^b\boldKf\Chin} = \norm{(\Np+1)^b\boldKf^-\Chin}\leq \fC(n,b)N^\frac16\label{eqn:K4:a:priori:2}
\end{eqnarray}
by part (a) and Lemma \ref{lem:prelim:commutators}. To improve this \textit{a priori} bound, we apply a similar argument to the bootstrapping in part (a). 
As in \eqref{eqn:lem:boot}, 
\begin{subequations}
\begin{eqnarray}
&&\hspace{-0.8cm}\Tr\left(\Pn(\Np+1)^{2b}\boldKf^2\right)
\;=\; \Tr\left(\Pn(\Np+1)^{2b}(\boldKf^-)^2\right)\nonumber\\
&=&\Tr\left(\Pzn(\Np+1)^{2b}(\boldKf^-)^2\right)\label{eqn:lem:boot:K4:1}\\
&&+\lN^\frac12\Tr\left(\frac{1}{2\pi\i}\goint\frac{1}{z-\Ezn}\frac{\Qn}{z-\tFockH}\FockRz\Pzn(\Np+1)^{2b}(\boldKf^-)^2\dz\right)\quad\label{eqn:lem:boot:K4:2}\\
&&+\lN^\frac12 \Tr\left(\frac{1}{2\pi\i}\goint\frac{\Pn}{z-\FockHminus}\FockRz\ResHzz(\Np+1)^{2b}(\boldKf^-)^2\dz\right)\label{eqn:lem:boot:K4:3}\,.
\end{eqnarray}
\end{subequations}
Since $[\boldKf^-,\Np]=0$, Lemma \ref{lem:K} implies for the first term that
\begin{equation}
\eqref{eqn:lem:boot:K4:1}=\sum\limits_{m=1}^{\dzn} \norm{\boldKf(\Np+1)^b\Chiznm}^2_\FNp \leq \fC(n,b)\,.
\end{equation}
In \eqref{eqn:lem:boot:K4:2}, this leads for $z\in\gan$ and $\Chiznm\in\fEzn$ to
\begin{eqnarray}
&&\hspace{-1cm}N^{-\frac12}\left|\lr{\boldKf^-(\Np+1)^{2b+2}\Chiznm,\boldKf^-(\Np+1)^{-2}\frac{\Qn}{z-\tFockH}\FockRz\Chiznm}\right|\nonumber\\
&\leq&\fC\, N^{-\frac12}\norm{(\Np+1)^{2b+4}\Chiznm} \nonumber\\
&&\times \left(N^\frac16\Big\|\frac{\Qn}{z-\tFockH}\FockRz\Chiznm\Big\| +\left|\lr{\frac{\Qn}{z-\tFockH}\FockRz\Chiznm,\tFockH\frac{\Qn}{z-\tFockH}\FockRz\Chiznm}\right|^\frac12\right)\nonumber\\
&\leq& \fC(n,b) N^{-\frac13}\,,
\end{eqnarray}
where we used Lemmas \ref{lem:K} and \ref{lem:prelim:commutators} for the left-hand side and \eqref{eqn:K4:a:priori:1} for the right-hand side of the inner product in the first line, as well as Lemmas \ref{lem:aux:ResHz:op} and \ref{lem:estimate:remainders}.
Finally, for \eqref{eqn:lem:boot:K4:3}, \eqref{eqn:K4:a:priori:2} and Lemma \ref{lem:K} imply that
\begin{eqnarray}
&&\hspace{-1cm}N^{-\frac12}\left|\lr{\Chinm,\FockRz\ResHzz(\Np+1)^{2b}(\boldKf^-)^2\Chinm}\right|\nonumber\\
&\leq&N^{-\frac12}\big\|(\Np+1)^{2b}\boldKf\Chinm\big\| \Big\|\boldKf\ResHzz\FockRz\Chinm\Big\|\nonumber\\
&\leq&\fC(n,b) N^{-\frac13} \norm{(\Np+1)^\frac32(\FockRza+\lN^\frac12\boldKf^-)\Chinm}
\;\leq\;\fC(n,b) N^{-\frac13}
\end{eqnarray}
by definition \eqref{FockRz} of $\FockRz$ and by part (a).
In summary, we find
\begin{equation*}
\Tr\,\Pn(\Np+1)^{2b}\boldKf^2 = \sum\limits_{m=1}^{\dzn}\norm{(\Np+1)^b\boldKf\Chinm}^2\leq\fC(n,b)\,.\qedhere
\end{equation*}
\end{proof}

\subsection{Proof of the main results}\label{subsec:proofs:main:results}

In the following, we consider 
$$\FockA\in\big\{\FockAmred\,,\,\id\big\}$$
for $j\in\N_0$. By Lemma  \ref{lem:A}, $\FockA$ satisfies 
\begin{equation}
\norm{\FockA\bPhi}\leq\fC N^\alpha\left(\norm{(\Np+1)\bPhi}+\norm{\FockHz\bPhi}\right)
\end{equation}
for 
\begin{equation}
\alpha=
\begin{cases}
-\frac12&\text{ if } \FockA=\FockAmred\,,\\
\;\,0 &\text{ if } \FockA=\id\,.
\end{cases}
\end{equation}

\subsubsection{Proof of Theorem \ref{thm:exp:P}}\label{subsubsec:proofs:thm:exp:P}
Recall that by Proposition \ref{prop:exp:P}, 
\begin{equation*}
\Tr\,\FockA\Pn=\sum\limits_{\l=0}^a\lN^\frac{\l}{2}\Tr\,\FockA\Pnl +\lN^\frac{a+1}{2}\left(\Tr\,\FockA\FockBPn(a)+\Tr\,\FockA\FockBQn(a)\right)\,,
\end{equation*}
where
\begin{eqnarray*}
&&\hspace{-0.8cm}\FockBPn(a)\nonumber\\
&=&\hspace{-0.1cm}\sum\limits_{\nu=0}^{a}\sum\limits_{m=1}^{a-\nu}\sum\limits_{\substack{\bj\in\N^m\\[2pt]|\bj|=a-\nu}}\frac{1}{2\pi\i}\goint\frac{\Pn}{z-\FockHminus}\,\FockRnu\ResHzz\FockHjo\ResHzz\mycdots\FockH_{j_m} \ResHzz\dz
\end{eqnarray*}
and
\begin{eqnarray*}
&&\hspace{-0.7cm}\FockBQn(a)\nonumber\\
&=&\hspace{-0.1cm}\sum\limits_{\nu=0}^{a} 
\sum\limits_{m=1}^{a-\nu}
\sum\limits_{\substack{\bj\in\N^m\\[2pt]|\bj|=a-\nu}}
\sum\limits_{\l=0}^{m} 
\sum\limits_{\substack{\bk\in\{0,1\}^{m+1}\\|\bk|=\l}}
\frac{1}{2\pi\i}\goint
\frac{\Qn}{z-\tFockH}\FockRnu 
\frac{\FockIn_{k_1}}{z-\FockHz}\FockHjo \mycdots \FockH_{j_m}\frac{\FockIn_{k_{m+1}}}{z-\FockHz}\dz \qquad\quad
\end{eqnarray*}
with $\FockIn_0=\Pzn$ and $\FockIn_1=\Qzn$.

\subsubsection*{Estimates for $\FockBPn(a)$}

Let $\big\{\Chinl\big\}_{\l=1}^{\dzn}$ denote an orthonormal basis of $\fEn$ such that $\FockH\Chinl=E^{(n,\l)}\Chinl$.
Consequently, $\Pn=\sum_{\l=1}^{\dzn}|\Chinl\rangle\langle\Chinl|$, and interchanging trace and contour integral by Fubini's theorem yields
\begin{eqnarray}
\left|\Tr\,\FockA\FockBPn(a)\right| 
&\leq&\fC\sum\limits_{\nu=0}^{a}\sum\limits_{m=1}^{a-\nu}\sum\limits_{\substack{\bj\in\N^m\\[2pt]|\bj|=a-\nu}}
\sum\limits_{\l=1}^{\dzn}
\goint\left|\frac{1}{z-E^{(n,\l)}}\right|\nonumber\\
&&\times\left|\lr{\Chinl,\FockRnu\ResHzz\FockHjo\mycdots\FockH_{j_m} \ResHzz\FockA\Chinl}\right|\dz\,.\quad\label{eqn:B:P:1}
\end{eqnarray}
Lemmas \ref{lem:H:norms} and \ref{lem:aux:moments:O} lead to the estimate
\begin{eqnarray}\label{eqn:B:P:v:unbounded}
\Big\|(\Np+1)^b\FockH_j\frac{\FockIn}{z-\FockHz}\bPhi\Big\|\leq \fC(b,j) \norm{(\Np+1)^{b+\frac{j}{2}+1}\bPhi}
\end{eqnarray}
for $\FockIn\in\{\id,\Pzn,\Qzn\}$, hence
\begin{eqnarray}
&&\hspace{-1.3cm}\left|\lr{\Chinl,\FockRnu\ResHzz\FockHjo\ResHzz\mycdots\FockH_{j_m} \ResHzz\FockA\Chinl}\right|\nonumber\\
&\leq&\norm{\Chinl}\Big\|\FockA\ResHzz\FockH_{j_m}\ResHzz\mycdots\FockHjo\ResHzz\FockRnu
 \Chinl\Big\|\nonumber\\
&\leq&\fC N^\alpha\Big\|\FockH_{j_m}\ResHzz\mycdots\FockHjo\ResHzz\FockRnu
 \Chinl\Big\|\nonumber\\
&\leq&\fC(n,a) N^\alpha \norm{(\Np+1)^{\frac{3}{2}(a-\nu)}\FockRnu\Chinm}
\;\leq\;\fC(n,a) N^\alpha
\end{eqnarray}
by Lemmas \ref{lem:A}, \ref{lem:estimate:remainders} and \ref{lem:bootstrap}. Here, we used that $\FockRz=\FockRza+\lN^\frac12\boldKf$ and that $\FockRo=\FockRoa+\boldKf$, and applied Lemmas \ref{lem:K} and \ref{lem:bootstrap}.
In summary, this yields 
\begin{eqnarray}
\left|\Tr\, \FockA\FockBPn(a)\right|
&\leq& N^\alpha \fC(n,a) \,.
\end{eqnarray}

\subsubsection*{Estimates for $\FockBQn(a)$}
By definition of $\FockBQn$, it follows that
\begin{eqnarray}
\left|\Tr\,\FockA\FockBQn(a)\right|\hspace{-2.4pt}
&\leq& \hspace{-2.4pt}\fC(n)\sum\limits_{\nu=0}^{a} 
\sum\limits_{m=1}^{a-\nu}
\sum\limits_{\substack{\bj\in\N^m\\[2pt]|\bj|=a-\nu}}
\sum\limits_{\l=0}^{m} 
\sum\limits_{\substack{\bk\in\{0,1\}^{m+1}\\|\bk|=\l}}\nonumber\\
&&\hspace{-2.4pt}\times\sup\limits_{z\in\gan}
\left|\Tr\, \FockA\frac{\Qn}{z-\tFockH}\FockRnu 
\frac{\FockIn_{k_1}}{z-\FockHz}\FockHjo\frac{\FockIn_{k_2}}{z-\FockHz}\mycdots \FockH_{j_m}\frac{\FockIn_{k_{m+1}}}{z-\FockHz}
\right|.\qquad\quad\label{eqn:B:Q:4}
\end{eqnarray}
Each term contains at least one projector $\Pzn$, i.e., there exists some $\sigma\in\{1\mydots m+1\}$ such that $k_\sigma=0$. 
Decomposing $\Pzn=\sum_{\mu=1}^{\dzn}|\Chiznmu\rangle\langle\Chiznmu|$
for a basis $\{\Chiznmu\}_{\mu=1}^{\dzn}$ of $\fEzn$ as in Lemma \ref{lem:known:results:FockHz:2}, we obtain
\begin{subequations}
\begin{eqnarray}
&&\hspace{-1.5cm} \left|\Tr\, \FockA\frac{\Qn}{z-\tFockH}\FockRnu 
\frac{\FockIn_{k_1}}{z-\FockHz}\FockHjo\frac{\FockIn_{k_2}}{z-\FockHz}\mycdots \FockH_{j_m}\frac{\FockIn_{k_{m+1}}}{z-\FockHz}
\right|\nonumber\\
&\leq&\fC(n)\sum\limits_{\mu=1}^{\dzn}\left\|\frac{\Qn}{z-\tFockH}\FockRnu 
\frac{\FockIn_{k_1}}{z-\FockHz}\FockHjo\mycdots\frac{\FockIn_{k_{\sigma-1}}}{z-\FockHz}\FockH_{j_{\sigma-1}}\Chiznmu\right\|\label{eqn:B:Q:1}\\
&&\times\left\|\FockA\frac{\FockIn_{k_{m+1}}}{z-\FockHz} \FockH_{j_m}\frac{\FockIn_{k_m}}{z-\FockHz}\mycdots\frac{\FockIn_{k_{\sigma+1}}}{z-\FockHz}\FockH_{j_{\sigma}}\label{eqn:B:Q:2}
\Chiznmu\right\|\,.
\end{eqnarray}
\end{subequations}

Using the estimate \eqref{eqn:B:P:v:unbounded} in combination with Lemmas \ref{lem:K:a:norms:estimate:remainders}, \ref{lem:prelim:commutators}, we find for \eqref{eqn:B:Q:1}
\begin{eqnarray}
&&\hspace{-1.5cm}\bigg\|\frac{\Qn}{z-\tFockH}\FockRnu 
\frac{\FockIn_{k_1}}{z-\FockHz}\FockHjo\mycdots\frac{\FockIn_{k_{\sigma-1}}}{z-\FockHz}\FockH_{j_{\sigma-1}}\Chiznmu\bigg\|\nonumber\\
&\leq&\fC(n,a)\norm{(\Np+1)^\frac{\nu+2\sigma+1+j_1+\dots+j_{\sigma-1}}{2}\Chiznmu}\,
\end{eqnarray}
analogously to above, and for \eqref{eqn:B:Q:2} 
\begin{eqnarray}
&&\hspace{-1cm}\bigg\|\FockA\frac{\FockIn_{k_{m+1}}}{z-\FockHz} \FockH_{j_m}\frac{\FockIn_{k_m}}{z-\FockHz}\mycdots\frac{\FockIn_{k_{\sigma+1}}}{z-\FockHz}\FockH_{j_{\sigma}}\Chiznmu\bigg\|\nonumber\\
&\leq&N^\alpha \fC(n,a)\bigg\|(\Np+1)^\frac{j_{\sigma+1}+\dots+j_{m}+2(m-\sigma)}{2}\FockH_{j_{\sigma}}\Chiznmu\bigg\|\nonumber\\
&\leq&N^\alpha\fC(n,a)\norm{(\Np+1)^\frac{j_\sigma+\dots+ j_m+2(m-\sigma+1)}{2}\Chiznmu}
\end{eqnarray}
since $\FockHz\Chiznmu=\Ezn\Chiznmu$.
Combining both estimates yields with Lemma \ref{lem:moments:Chiz}
\begin{equation}
\left|\Tr\,\FockA\FockBQn(a)\right| \leq\fC(n,a) N^\alpha\,.
\end{equation}

\subsubsection{Proof of Corollary \ref{cor:trace:norm}}
\label{subsec:proofs:cor}
For any bounded operator $\FockA\in\cL(\Fp)$, Proposition \ref{prop:exp:P} implies that
\begin{equation}
\Big|\Tr\,\FockA\Pn -\sum\limits_{\l=0}^a\lN^\frac{\l}{2}\Tr\, \FockA\Pnl\Big| \leq \lN^\frac{a+1}{2}\left(\big|\Tr\,\FockA\FockB^{(n)}_P(a)\big| + \big|\Tr\,\FockA\FockB^{(n)}_Q(a)\big|\right)\,,
\end{equation}
and one infers from the previous section that
\begin{equation}
\big|\Tr\,\FockA\FockB^{(n)}_P(a)\big|+ \big|\Tr\,\FockA\FockB^{(n)}_Q(a)\big| \leq \onorm{\FockA} \fC(n,a)\,.
\end{equation}
Consequently,
\begin{equation*}
\Tr\,\Big|\Pn -\sum\limits_{\l=0}^a\lN^\frac{\l}{2} \Pnl\Big| = 
\sup\limits_{\substack{ \FockA \text{ compact}\\\onorm{\FockA}=1}} \Big|\Tr\,\FockA\Pn -\sum\limits_{\l=0}^a\lN^\frac{\l}{2}\Tr\, \FockA\Pnl\Big|
\leq \lN^\frac{a+1}{2} \fC(n,a)\,. \qed
\end{equation*}

\subsubsection{Proof of Theorem \ref{thm:energy}}
Let us abbreviate
$\goint':=\frac{1}{2\pi\i}\goint$. 
Note first that 
\begin{eqnarray}
\sum\limits_{\nu\in\In}\dnu\Enu 
&=&\Tr\, \FockH\Pn 
\;=\; \Tr\goint'\frac{\FockH}{z-\FockH}\dz
\;=\;\Tr\goint'\frac{z}{z-\FockH}\dz\nonumber\\
&=&\Ezn\Tr\, \Pn+ \Tr\goint'\frac{z-\Ezn}{z-\FockH}\dz\,.
\end{eqnarray}
Since $\Tr\,\Pn=\dzn$ and 
\begin{equation}
\goint'\frac{z-\Ezn}{z-\FockHz}\dz=\Pzn\goint'1\dz+\goint'\frac{\Qzn}{z-\FockHz}(z-\Ezn)\dz= 0\,,
\end{equation}
this implies by Lemma \ref{lem:expansion:resolvent}  that
\begin{subequations}
\begin{eqnarray}
\Tr\,\FockH\Pn 
&=&\dzn\Ezn+\sum\limits_{\l=1}^a\lN^\frac{\l}{2}\sum\limits_{\nu=1}^\l\sum\limits_{\substack{\bj\in\N^\nu\\|\bj|=\l}}\Tr\goint'\ResHzz\FockHjo\ResHzz \mycdots\nonumber\\
&&\hspace{3.5cm}\times\,\FockHjnu\frac{z-\Ezn}{z-\FockHz}\dz\label{eqn:cor:1}\\ 
&&+ \lN^\frac{a+1}{2}\sum\limits_{\nu=0}^a\sum\limits_{m=1}^{a-\nu}\sum\limits_{\substack{\bj\in\N^m\\|\bj|=a-\nu}}\Tr\goint'\RestHz\FockRnu\ResHzz\FockHjo\nonumber\\
&&\hspace{3.5cm}\times\,\ResHzz\mycdots\FockH_{j_m}\frac{z-\Ezn}{z-\FockHz}\dz\,.\qquad\qquad\label{eqn:cor:2}
\end{eqnarray}
\end{subequations}
For $z\in\gan$, it holds that $|z-\Ezn|\leq\fC$, hence  the proof of Theorem \ref{thm:exp:P} for $\FockA=\id$ yields   
\begin{equation}
|\eqref{eqn:cor:2}|\leq\lN^\frac{a+1}{2}\fC(n,a)\,.
\end{equation}
Moreover,  all half-integer powers of $\lN$ in \eqref{eqn:cor:1} vanish by parity: define the unitary map 
\begin{equation}
\UP:\Fock\to\Fock\,,\qquad \UP\ad(f)\UP =\ad(-f)=-\ad(f)
\end{equation}
for any $ f\in\fH$. Clearly, $\UP$ preserves $\Fp$ and acts on the operator-valued distributions $\ad_x$ and $a_x$ as $\UP\ad_x\,\UP=-\ad_x$ and $\UP a_x\,\UP=-a_x$. By definition \eqref{FockHj}, $\FockHj$ contains an even number of creation/annihilation operators for $j$ even and an odd number for $j$ odd, hence
\begin{eqnarray}
\UP\FockHj\UP=(-1)^j\FockHj\,,\qquad
\UP\ResHzz\UP=\ResHzz
\end{eqnarray}
because  $\UP\FockHz^\l\UP=\FockHz^\l$ for any $\l\in\R$.
Consequently,
\begin{eqnarray}
\Tr\, \ResHzz\FockHjo\mycdots\FockHjnu\ResHzz 
&=& \Tr\, \UP\ResHzz\FockHjo\mycdots\FockHjnu\ResHzz \UP\nonumber\\
&=&(-1)^\l\Tr\, \ResHzz\FockHjo\mycdots\FockHjnu\ResHzz\qquad
\end{eqnarray}
for any $\bj$ such that $|\bj|=\l$. This yields 
\begin{equation}
\Tr\,\FockH\Pn = \dzn\Ezn + \sum\limits_{\l=1}^a\lN^\l\sum\limits_{\nu=1}^{2\l} \FockE^{(n)}_{\l,\nu} + \mathcal{O}(\lN^{a+1})
\end{equation}
with
\begin{equation}
\FockE^{(n)}_{\l,\nu}:=\sum\limits_{\substack{\bj\in\N^\nu\\|\bj|=2\l}}\goint'\Tr\left(\ResHzz\right)^2\FockHjo\ResHzz\mycdots\FockHjnu(z-\Ezn)\dz\,.\label{eqn:cor:B}
\end{equation}
For $\nu=1$, one computes
\begin{equation}
\FockE^{(n)}_{\l,1}=\goint'\Tr\,\Pzn\FockH_{2\l}\frac{\dz}{z-\Ezn}=\Tr\,\Pzn\FockH_{2\l}\,.
\end{equation}
For $\nu\geq2$, we decompose each identity as $1=\Pzn+\Qzn$ and order the summands according to the number $k$ of projections $\Qzn$, which yields
\begin{equation}
\FockE^{(n)}_{\l,\nu}
=\sum\limits_{k=1}^{\nu-2}\FockE^{(n)}_{\l,\nu,k}+\FockE^{(n)}_{\l,\nu,\nu-1}
\end{equation}
with
\begin{eqnarray}
\FockE^{(n)}_{\l,\nu,\nu-1}
&=&\sum\limits_{\substack{\bj\in\N^\nu\\|\bj|=2\l}}\Tr\,\Pzn\FockHjo\FockOno\mycdots\FockOno\FockHjnu
\end{eqnarray}
for $\FockOn_m$ as in Definition \ref{def:Pna}, and 
\begin{subequations}
\begin{eqnarray}
\FockE^{(n)}_{\l,\nu,k}\hspace{-3pt}&=&\hspace{-7pt}\sum\limits_{\substack{\bj\in\N^\nu\\|\bj|=2\l}}\Bigg(\goint'
\Tr\,\Pzn\left[\FockHjo\tFockOno\mycdots\FockH_{j_k}\tFockOno\FockH_{j_{k+1}}\Pzn\mycdots\Pzn\FockHjnu\right]_\mathrm{p}\nonumber\\
&&\hspace{5cm}\times\,\frac{\dz}{(z-\Ezn)^{\nu-k}}\label{eqn:cor:A:1}\\
&&\hspace{-7pt}+\goint'\Tr\,\tFockOnt\left[\FockHjo\tFockOno\mycdots\FockH_{j_{k-1}}\tFockOno\FockH_{j_k}\Pzn\mycdots\Pzn\FockHjnu\right]_\mathrm{p}\nonumber\\
&&\hspace{5cm}\times\frac{\dz}{(z-\Ezn)^{\nu-k-1}}
\Bigg)\qquad\quad\label{eqn:cor:A:2}
\end{eqnarray}
\end{subequations}
for $k\leq\nu-2$.
Here, we abbreviated
$$\tFockOn_m=\frac{\Qzn}{(z-\FockHz)^m}\,,$$ 
and the notation $[\cdot]_\mathrm{p}$ indicates the sum of all possibilities to distribute the operators $\Pzn$ over the slots between the operators $\FockHj$. By cyclicity of the trace, 
\begin{eqnarray}
&&\hspace{-1cm}\sum\limits_{\substack{\bj\in\N^\nu\\|\bj|=2\l}}\Tr\,\Pzn\left[\FockHjo\tFockOno\mycdots\FockH_{j_{k}}\tFockOno\FockH_{j_{k+1}}\Pzn\mycdots\Pzn\FockHjnu\right]_\mathrm{p}\nonumber\\
&=& \sum\limits_{\substack{\bj\in\N^\nu\\|\bj|=2\l}}\frac{\nu-k}{\nu}
\Tr\left[\tFockOno\FockHjo\mycdots\tFockOno\FockH_{j_{k}}\Pzn\FockH_{j_{k+1}}\mycdots\Pzn\FockHjnu\right]_\mathrm{p}\,,\label{eqn:cor:cyclicity}
\end{eqnarray}
which can be seen by observing that the first line is a sum of $\binom{\nu-1}{k}$ terms while the sum in the second line has $\binom{\nu}{k}=\frac{\nu}{\nu-k}\binom{\nu-1}{k}$ addends. Next, we note that for any $f$ which is holomorphic in the interior of $\gan$,  the residue theorem implies that
\begin{equation}\label{eqn:cor:residue}
\goint' f(z)\frac{\dz}{(z-\Ezn)^{\nu-k}} = \frac{1}{\nu-k-1}\goint' f'(z)\frac{\dz}{(z-\Ezn)^{\nu-k-1}}\,.
\end{equation}
Since 
\begin{equation}\label{eqn:cor:3}
\frac{\d^m}{\dz^m}\tFockOno=(-1)^m m!\,\tFockOn_{m+1}\,,
\end{equation}
it follows that
\begin{eqnarray}
&&\hspace{-1cm}\sum\limits_{\substack{\bj\in\N^\nu\\|\bj|=2\l}}\frac{\d}{\dz}\Tr\left[\tFockOno\FockHjo\mycdots\tFockOno\FockH_{j_{k}}\Pzn\FockH_{j_{k+1}}\mycdots\Pzn\FockHjnu\right]_\mathrm{p}\nonumber\\
&=& -\nu\sum\limits_{\substack{\bj\in\N^\nu\\|\bj|=2\l}}\Tr\,\tFockOnt\left[\FockHjo\tFockOno\mycdots\tFockOno\FockH_{j_{k}}\Pzn\FockH_{j_{k+1}}\mycdots\Pzn\FockHjnu\right]_\mathrm{p}\label{eqn:cor:derivative}
\end{eqnarray}
because, by the product rule, the first line is a sum of $k\binom{\nu}{k}=\nu\binom{\nu-1}{k-1}$ terms. Integrating by parts yields
\begin{eqnarray}
&&\hspace{-0.7cm}\FockE^{(n)}_{\l,\nu,k}\nonumber\\
&=&\sum\limits_{\substack{\bj\in\N^\nu\\|\bj|=2\l}}
\frac{1}{\nu}\goint'\Tr\left[\tFockOno\FockHjo\tFockOno\mycdots\tFockOno\FockH_{j_{k}}\Pzn\mycdots\Pzn\FockHjnu\right]_\mathrm{p}\frac{\dz}{(z-\Ezn)^{\nu-k}}\nonumber\\
&=&\sum\limits_{\substack{\bj\in\N^\nu\\|\bj|=2\l}}
\frac{1}{\nu-k}\goint'\Tr\,\Pzn\left[\FockHjo\tFockOno\mycdots\FockH_{j_{k}}\tFockOno\FockH_{j_{k+1}}\Pzn\mycdots\Pzn\FockHjnu\right]_\mathrm{p}\nonumber\\
&&\hspace{7cm}\times\frac{\dz}{(z-\Ezn)^{\nu-k}}\,.\qquad\quad
\end{eqnarray}
Consequently, the residue theorem and \eqref{eqn:cor:3} lead to
\begin{eqnarray}
&&\hspace{-0.8cm}\FockE^{(n)}_{\l,\nu,k}\nonumber\\
&=&\sum\limits_{\substack{\bj\in\N^\nu\\|\bj|=2\l}}\frac{1}{(\nu-k)!}\frac{\d^{\nu-k-1}}{\dz^{\nu-k-1}}\Tr\,\Pzn\left[\FockHjo\tFockOno\mycdots\tFockOno\FockH_{j_{k+1}}\Pzn\mycdots\Pzn\FockHjnu\right]_\mathrm{p}\Big|_{z=\Ezn}\nonumber\\
&=&\sum\limits_{\substack{\bj\in\N^\nu\\|\bj|=2\l}}\frac{(-1)^{\nu-k-1}}{\nu-k}\sum\limits_{\substack{\bm\in\N_0^k\\|\bm|=\nu-k-1}}\Tr\,\Pzn\left[\FockHjo\FockOn_{m_1+1}\mycdots\FockOn_{m_k+1}\FockH_{j_{k+1}}\Pzn\mycdots\Pzn\FockH_{j_\nu}\right]_\mathrm{p}\nonumber\\
&=&\sum\limits_{\substack{\bj\in\N^\nu\\|\bj|=2\l}}\frac{(-1)^{\nu-k-1}}{\nu-k}\sum\limits_{\substack{\bm\in\N^k\\|\bm|=\nu-1}}\Tr\,\Pzn\left[\FockHjo\FockOn_{m_1}\mycdots\FockOn_{m_k}\FockH_{j_{k+1}}\Pzn\mycdots\Pzn\FockHjnu\right]_\mathrm{p}.\qquad
\end{eqnarray}
Recall that the subscript ``p" indicates the sum over all possibilities to distribute $\Pz$. In particular, this implies that all empty slots are subsequently filled up with the tuple $(\FockOn_{m_1},\dots,\FockOn_{m_k})$ without permuting the positions of the $\FockOn_{m_j}$.
Using the notation $\FockOnz=-\Pzn$, one can equivalently write
\begin{eqnarray}
&&\hspace{-0.9cm}\sum\limits_{k=1}^{\nu-1}\FockE^{(n)}_{\l,\nu,k} \nonumber\\
&=&\sum\limits_{\substack{\bj\in\N^\nu\\|\bj|=2\l}}\sum\limits_{k=1}^{\nu-1}\sum\limits_{\substack{\bm\in\N^k\\|\bm|=\nu-1}}
\frac{1}{\nu-k}\Tr\,\Pzn\left[\FockHjo\FockOn_{m_1}\mycdots\FockH_{j_k}\FockOn_{m_k}\FockH_{j_{k+1}}\FockOnz\mycdots\FockOnz\FockHjnu\right]_\mathrm{p}\nonumber\\
&=&\sum\limits_{\substack{\bj\in\N^\nu\\|\bj|=2\l}}
\sum\limits_{\substack{\bm\in\N_0^{\nu-1}\\|\bm|=\nu-1}}
\frac{1}{\kappa(\bm)}\Tr\,\Pzn\FockHjo\FockOn_{m_1}\mycdots\FockH_{j_{\nu-1}}\FockOn_{m_{\nu-1}}\FockHjnu\,,
\end{eqnarray}
where we denoted by $\kappa(\bm)-1$ the number of operators $\FockOnz$. 
Finally, in case of a non-degenerate eigenvalue $\Ezn$, some terms  vanish by parity, which leads to the simplified expressions \eqref{eqn:cor:explicit}.\qed

\appendix
\section{Excitation Hamiltonian}\label{appendix:Hamiltonian}

For $h$ and $\eH$ as in Lemma~\ref{lem:hH} and $W$ as defined in \eqref{eqn:W(x,y)}, i.e.,
$$
W(x_1,x_2) = v(x_1-x_2) - \left(v*\varphi^2\right)(x_1) - \left(v*\varphi^2\right)(x_2) + \lr{\varphi,v*\varphi^2\varphi}\,,
$$
it follows that
\begin{equation}
H_N = N\eH + \sum_{j=1}^N h_j + \lN \sum\limits_{1\leq i<j\leq N} W(x_i,x_j)\,.
\end{equation}
We denote by $\{\varphi_n\}_{n\geq0}$, $\varphi_0=\varphi$, an eigenbasis for $h$ 
and abbreviate
\begin{eqnarray}
h_{mn}&:=&\lr{\varphi_m, h \varphi_n}\,,\\
W_{mnpq}&:=&\int\dx\dy\,\overline{\varphi_m(x)}\,\overline{\varphi_n(y)}W(x,y)\varphi_p(x)\varphi_q(y)\,,
\end{eqnarray}
and $a^{\sharp}_m := a^{\sharp}(\varphi_m)$.
Since $h_{m0}=h_{0n}=0$ and $h_{mn}=0$ for $m\neq n$, it follows that
\begin{eqnarray}
H_N &=& N \eH + \sum_{m,n \geq 0} h_{mn} \ad_m a_n + \frac{\lN}{2} \sum_{m,n,p,q \geq 0} W_{mnpq} \ad_m \ad_n a_p a_q \nonumber \\
&=&  N \eH + \sum_{m > 0} h_{mm} \ad_m a_m + \frac{\lN}{2}  W_{0000} \ad_0 \ad_0 a_0 a_0 \nonumber\\
&&+ \left( \lN \sum_{m > 0} W_{000m} \ad_0 \ad_0 a_0 a_m + \hc \right) \nonumber\\
&& + \frac{\lN}{2}\left(\sum\limits_{m,n>0} W_{m0n0}\ad_m\ad_0 a_n a_0 +\hc\right)+ \left( \frac{\lN}{2} \sum_{m,n > 0} W_{mn00} \ad_m \ad_n a_0 a_0 + \hc \right) \nonumber\\
&&+ \frac{\lN}{2} \sum_{m,n > 0} \left( W_{0mn0} \ad_0 \ad_m a_n a_0 + W_{m00n}\ad_m\ad_0 a_0 a_n\right)\nonumber\\
&& + \left( \lN \sum_{m,n,p > 0} W_{mnp0} \ad_m \ad_n a_p a_0 + \hc \right)\nonumber\\
&& + \frac{\lN}{2} \sum_{m,n,p,q > 0} W_{mnpq} \ad_m \ad_n a_p a_q\,.\label{eqn:app:A}
\end{eqnarray}
As $W_{0000}= W_{000m} = W_{m0n0}=0 $, 
$W_{0mn0} = \lr{\varphi_m, K_1\varphi_n}_\fH$, $W_{mn00} = \lr{\varphi_m\otimes \varphi_n, K_2}_{\fH^2}$, and  $W_{mnp0} =\lr{\varphi_m\otimes\varphi_n,K_3\varphi_p}_{\fH^2} $,
\eqref{eqn:FockHN} follows  from \eqref{eqn:app:A} by the substitution rules \eqref{eqn:substitution:rules}.

\section{Asymptotic expansion of the wave function}
\label{appendix:wf}

\begin{theorem}\label{thm:wave:fctn:general}
Let $\fH$ be a Hilbert space, let $\chi\in\fH$ with $\norm{\chi}=1$ and define $P:=|\chi\rangle\langle\chi|$.
Assume that $P$ admits an asymptotic expansion in the small parameter $\varepsilon>0$, i.e., there exists a family of $\varepsilon$-independent operators $\left\lbrace P_\l\right\rbrace_{\l\in\N_0}$ such that, for any $a\in\N_0$,
\begin{equation}\label{eqn:assumption:thm:wf}
\Tr_\fH\Big|P-\sum\limits_{\l=0}^a\varepsilon^\l P_\l\Big| \leq C(a)\,\varepsilon^{a+1}
\end{equation}
for some constant $C(a)>0$ and sufficiently small $\varepsilon$. Moreover, assume that there exists some normalized $\chi_0\in\fH$ such that $P_0=|\chi_0\rangle\langle\chi_0|$.
Then, for a suitable choice of the phase of $\chi_0$, there exists for any $a\in\N_0$ a constant $\tilde{C}(a)>0$ such that 
\begin{equation}
\Big\|\chi-\sum\limits_{\l=0}^a\varepsilon^\l\Chi_\l\Big\| \leq \tilde{C}(a)\varepsilon^{a+1}\,,
\end{equation}
where
\begin{subequations}\label{eqn:ansatz:chil}
\begin{eqnarray}
\chi_\l&:=&\sum\limits_{j=0}^\l \alpha_j\tilde{\chi}_{\l-j} \qquad (\l\geq 1)\,,\label{eqn:ansatz:chil:1}\\
\tilde{\chi}_\l&:=&\sum\limits_{\nu=1}^\l\sum\limits_{\substack{\bj\in\N^\nu\\|\bj|=\l}} P_{j_1}\,\mycdots P_{j_\nu}\chi_0 \qquad (\l\geq 1)\,,\label{eqn:ansatz:chil:2}
\end{eqnarray}
and
\begin{equation}
\alpha_0=1\,,\qquad
\alpha_\l:=-\frac12\sum\limits_{\substack{\bj\in\N_0^4\\j_1,j_2<\l\\|\bj|=\l}} \alpha_{j_1}\alpha_{j_2}\lr{\tilde{\chi}_{j_3},\tilde{\chi}_{j_4}} \qquad (\l\geq 1)\,.\label{eqn:ansatz:chil:3}
\end{equation}
\end{subequations}
\end{theorem} 

Before proving Theorem \ref{thm:wave:fctn:general}, let us first formally derive \eqref{eqn:ansatz:chil}. 
Inserting \eqref{eqn:assumption:thm:wf} and the ansatz 
\begin{equation}\label{eqn:chil:formal:ansatz}
\chi=\sum\limits_{\l\geq 0}\varepsilon^\l\chi_\l
\end{equation}
into the equation $P\chi=\chi$
yields  formally
\begin{equation}
\sum\limits_{\l= 0}^\infty\sum\limits_{k=0}^\l\varepsilon^\l P_k\chi_{\l-k} = \sum\limits_{\l= 0}^\infty\varepsilon^\l\chi_\l\,,
\end{equation}
hence
\begin{equation}
\chi_\l -P_0\chi_\l = \sum\limits_{k=1}^\l P_k\chi_{\l-k}
\end{equation}
and consequently
\begin{equation}\label{eqn:chil:formal:derivation}
\chi_\l = \sum\limits_{k=1}^\l P_k\chi_{\l-k} + \alpha_\l\chi_0
\end{equation}
for any $\l\geq0$ and $\alpha_\l\in\C$, $\alpha_0=1$.  By induction over $\l\in\N_0$, one easily verifies that  $\chi_\l$ can equivalently be written as \eqref{eqn:ansatz:chil:1} with $\tilde{\chi}_\l$ given by \eqref{eqn:ansatz:chil:2}, without any further restriction on the parameters $\alpha_\l$. 
It remains to derive the formula \eqref{eqn:ansatz:chil:3} for the (so far free) parameters $\alpha_\l$. To this end, we observe that formally
\begin{equation}
P= |\chi\rangle\langle\chi| = \sum\limits_{\l=0}^\infty\varepsilon^\l\sum\limits_{k=0}^\l|\chi_k\rangle\langle\chi_{\l-k}|\,,
\end{equation}
which motivates the definition
\begin{equation}\label{eqn:P:l:wf}
P_\l^\wf:=\sum\limits_{k=0}^\l|\chi_k\rangle\langle\chi_{\l-k}|\,.
\end{equation}
By \eqref{eqn:ansatz:chil:1}, this can equivalently be expressed as
\begin{equation}
P_\l^\wf=\sum\limits_{k=0}^\l\sum\limits_{i=0}^k\sum\limits_{m=0}^{\l-k} \alpha_i\,\overline{\alpha}_m |\tilde{\chi}_{k-i}\rangle\langle\tilde{\chi}_{\l-k-m}|=\sum\limits_{\substack{\bj\in\N_0^4\\|\bj|=\l}} \alpha_{j_1}\,\overline{\alpha}_{j_2}|\tilde{\chi}_{j_3}\rangle\langle\tilde{\chi}_{j_4}|\,.
\end{equation}
Formally, it is clear that $P_\l^\wf$ are the coefficients in the expansion of $P$, and our goal will be to rigorously establish the equality $P_\l^\wf=P_\l$. 
By \eqref{eqn:assumption:thm:wf} and since $\Tr_\fH P_0=1$, it follows that
\begin{equation}
1=\Tr_\fH P = 1+\sum\limits_{\l=1}^a \varepsilon^\l \Tr_\fH P_\l + \mathcal{O}(\varepsilon^{a+1})\,,
\end{equation}
hence $\Tr_\fH P_\l =0$ for any $\l\geq 1$.
Therefore, we choose the free parameters $\alpha_\l$ such that $\Tr_\fH P_\l^\wf=0$ for any $\l\geq 0$, which implies that 
\begin{equation}
\alpha_\l+\overline{\alpha}_\l =- \sum\limits_{\substack{\bj\in\N_0^4\\j_1,j_2<\l\\|\bj|=\l}} \alpha_{j_1}\,\overline{\alpha}_{j_2}\lr{\tilde{\chi}_{j_3},\tilde{\chi}_{j_4}}=0\,,
\end{equation}
and choosing $\alpha_\l$ real results in \eqref{eqn:ansatz:chil:3}.
Next, we prove an auxiliary lemma:

\begin{lem}\label{lem:thm:wf}
Under the assumptions of Theorem \ref{thm:wave:fctn:general}, it holds for any $\l\in\N_0$ that 
\begin{equation}\label{eqn:lem:thm:wf}
P_\l=\sum\limits_{j=0}^\l P_j P_{\l-j}\,.
\end{equation}
\end{lem}

\begin{proof}
By assumption, it holds for any $a\in\N_0$ that
\begin{equation}
P=\sum\limits_{\l=0}^a\varepsilon^\l P_\l + \varepsilon^{a+1} R_a
\end{equation}
for some $R_a\in\cL(\fH)$ with $\onorm{R_a}\leq C(a)$. Since $P^2=P$, this implies that
\begin{equation}
\sum\limits_{\l=0}^a\varepsilon^\l P_\l + \varepsilon^{a+1} R_a=\sum\limits_{\l=0}^a\varepsilon^\l\left(\sum\limits_{m=0}^\l P_m P_{\l-m}\right) +\varepsilon^{a+1}\tilde{R}_a
\end{equation}
with
\begin{equation}
\tilde{R}_a = \sum\limits_{\l=0}^a\sum\limits_{m=0}^{\l-1}\varepsilon^m P_\l P_{m+a+1-\l} + \sum\limits_{k=0}^a\varepsilon^k\left(R_a P_k+P_k R_a\right) + \varepsilon^{a+1} R_aR_a\,.
\end{equation}
Consequently, it holds for any $a\in\N_0$ that 
\begin{equation}\label{eqn:lem:thm:wf:1}
\bigg\|\sum\limits_{\l=0}^a\varepsilon^\l\left(P_\l-\sum\limits_{m=0}^\l P_m P_{\l-m}\right)\bigg\|_\mathrm{op}\leq \varepsilon^{a+1} \onorm{R_a-\tilde{R}_a} \leq C(a)\varepsilon^{a+1}\,,
\end{equation}
and \eqref{eqn:lem:thm:wf} follows by induction over $a\in\N$. 
\end{proof}

\noindent \textbf{Proof of Theorem \ref{thm:wave:fctn:general}.}
We prove Theorem \ref{thm:wave:fctn:general} in two steps: first, we show that the operators $P_\l^\wf$ from \eqref{eqn:P:l:wf}, which are constructed from the ansatz \eqref{eqn:ansatz:chil} for the functions $\chi_\l$, equal the coefficients $P_\l$ in the expansion \eqref{eqn:assumption:thm:wf} of $P$; second, we estimate the difference between the truncated power series with coefficients $\chi_\l$ and the function $\chi$.

\begin{claim}\label{claim:Pl=Plwf}
Under the assumptions of Theorem \ref{thm:wave:fctn:general},   it holds for any $\l\in\N_0$ that
\begin{equation}\label{eqn:Pl=Plwf}
P_\l^\wf=P_\l\,.
\end{equation}
\end{claim}

\begin{proof}
We prove \eqref{eqn:Pl=Plwf} by induction over $\l\in\N_0$. By Lemma \ref{lem:thm:wf} and since $\Tr_\fH P_1=0$, we conclude that $\Tr_\fH P_0P_1=0$ and consequently $\alpha_1=0$. Hence, $\chi_1=\tilde{\chi}_1 = P_1\chi_0$, and \eqref{eqn:P:l:wf} and Lemma \ref{lem:thm:wf} imply that $P_1^\wf=P_1$.
Now assume \eqref{eqn:Pl=Plwf}   for some $\l\in\N$. Then, by \eqref{eqn:chil:formal:derivation},
\begin{eqnarray}
P_{\l+1}^\wf
&=&\sum\limits_{k=0}^{\l}|\chi_k\rangle\langle\chi_{\l+1-k} |+ |\chi_{\l+1}\rangle\langle\chi_0|\nonumber\\
&=&\sum\limits_{j=1}^{\l+1}\sum\limits_{k=0}^{\l+1-j}|\chi_k\rangle\langle\chi_{\l+1-k-j}|P_j + \sum\limits_{k=1}^\l(\alpha_{\l+1-k}+P_{\l+1-k})|\chi_k\rangle\langle\chi_0|\nonumber\\
&& + 2\alpha_{\l+1}P_0 + P_{\l+1}P_0\nonumber\\
&=&\sum\limits_{j=1}^{\l+1}P_{\l+1-j}^\wf P_j+ \sum\limits_{k=1}^\l(P_{\l+1-k}+\alpha_{\l+1-k})|\chi_k\rangle\langle\chi_0|\nonumber\\
&& + 2\alpha_{\l+1}P_0+P_{\l+1}P_0\,.
\end{eqnarray}
By induction hypothesis and Lemma \ref{lem:thm:wf}, 
\begin{equation}
\sum\limits_{j=1}^{\l+1}P_{\l+1-j}^\wf P_j + P_{\l+1}P_0 
= \sum\limits_{j=0}^{\l+1}P_{\l+1-j} P_j 
= P_{\l+1}\,,
\end{equation}
hence
\begin{equation}\label{eqn:Pl=Plwf:1}
P_{\l+1}^\wf = P_{\l+1}+\sum\limits_{k=1}^\l(P_{\l+1-k}+\alpha_{\l+1-k})|\chi_k\rangle\langle\chi_0| + 2\alpha_{\l+1}P_0\,.
\end{equation}
By construction, $\Tr_{\fH}P_\l^\wf=\Tr_{\fH}P_\l=0$ for any $\l\geq 1$. Consequently, taking the trace of \eqref{eqn:Pl=Plwf:1} yields
\begin{equation}
\alpha_{\l+1}=-\frac12\sum\limits_{k=1}^\l\lr{\chi_0,(P_{\l+1-k}+\alpha_{\l+1-k})\chi_k},
\end{equation}
which implies that
\begin{equation}
P_{\l+1}^\wf = P_{\l+1}+(1-P_0)\sum\limits_{k=1}^\l(P_{\l+1-k}+\alpha_{\l+1-k})|\chi_k\rangle\langle\chi_0|\,.
\end{equation}
Finally,
\begin{equation}\label{eqn:Pl=Plwf:2}
P_0P_{\l+1}^\wf = P_0 P_{\l+1}\,,\qquad P_{\l+1}^\wf (1-P_0)=P_{\l+1}(1-P_0)
\end{equation}
and, since both $P_{\l+1}$ and $P_{\l+1}^\wf$ are self-adjoint, the first equality  implies that $P_{\l+1}^\wf P_0 = P_{\l+1}P_0$. Adding this to the second equality in \eqref{eqn:Pl=Plwf:2} concludes the proof of Claim~\ref{claim:Pl=Plwf}.
\end{proof}

\begin{claim}
Under the assumptions of Theorem \ref{thm:wave:fctn:general}, it holds for any $a\in\N_0$ that
\begin{equation}
\Big\|\chi-\sum\limits_{\l=0}^a\varepsilon^\l \chi_\l\Big\|_\fH \leq \tilde{C}(a) \,\varepsilon^{a+1}\,.
\end{equation}
\end{claim}

\begin{proof}
By \eqref{eqn:assumption:thm:wf},  all operators $P_\l$ are bounded uniformly in $\varepsilon$. Recall that for any normalized $f,g\in\fH$, it holds that $\norm{f-g}_\fH\leq 1/\sqrt{2}\,\Tr_\fH\big||f\rangle\langle f|-|g\rangle\langle g|\big|$ for a suitably chosen relative phase. By construction, the phase of all $\chi_\l$ is determined by the phase of $\chi_0$.
Hence, setting $n_{\varepsilon,a}:=\norm{\sum_{\l=0}^a\varepsilon^\l \chi_\l}^{-1}$, Claim \ref{claim:Pl=Plwf} implies for a suitable choice of the phase of $\chi_0$ that 
\begin{eqnarray}
\Big\|\chi-\sum\limits_{\l=0}^a\varepsilon^\l\chi_\l\Big\|_\fH
&\leq&\frac{1}{\sqrt{2}}\,\Tr_\fH\Big|P-n_{\varepsilon,a}^2\sum\limits_{\l=0}^a \sum\limits_{k=0}^a \varepsilon^{\l+k}|\chi_\l\rangle\langle\chi_k|\Big|+\left|\frac{1-n_{\varepsilon,a}}{n_{\varepsilon,a}}\right|\nonumber\\
&\leq&\frac{1}{\sqrt{2}}\,\Tr_\fH\Big|P-\sum\limits_{\l=0}^a\varepsilon^\l P_\l^\wf\Big|
+\frac{\varepsilon^{a+1}}{\sqrt{2}}\sum\limits_{\l=0}^a\sum\limits_{j=1}^\l\norm{\chi_\l}_\fH\norm{\chi_{a+j-\l}|}_\fH\nonumber\\
&&+\left|\frac{1-n_{\varepsilon,a}^2}{\sqrt{2}\,n_{\varepsilon,a}^2}\right|+\left|\frac{1-n_{\varepsilon,a}}{n_{\varepsilon,a}}\right|\nonumber\\
&\leq&  \tilde{C}(a)\,\varepsilon^{a+1}  
\end{eqnarray}
by \eqref{eqn:assumption:thm:wf} and \eqref{eqn:ansatz:chil}. Besides,  we used that
\begin{equation}
n_{\varepsilon,a}^{-2} 
= \Tr_\fH \Big|\sum\limits_{\l=0}^a\varepsilon^\l\chi_\l\Big\rangle\Big\langle\sum\limits_{k=0}^a\varepsilon^k\chi_k\Big|
=\Tr_\fH\Big(\sum\limits_{\l=0}^a\varepsilon^\l P_\l\Big) + \varepsilon^{a+1}R_{\varepsilon,a}
=1+\varepsilon^{a+1}R_{\varepsilon,a}
\end{equation}
with 
\begin{equation}
R_{\varepsilon,a}=\sum\limits_{\l=0}^a\sum\limits_{j=1}^\l\varepsilon^{j-1}\Tr_\fH|\chi_\l\rangle\langle\chi_{a+j-\l}|\,,\qquad
|R_{\varepsilon,a}|\leq C(a)
\end{equation}
for some constant $C(a)$, which implies that $\left|\frac{1-n_{\varepsilon,a}^2}{n_{\varepsilon,a}^2}\right|\leq C(a)\varepsilon^{a+1}$ as well as $\left|\frac{1-n_{\varepsilon,a}}{n_{\varepsilon,a}}\right|\leq C(a)\varepsilon^{a+1}$.
\end{proof}

\noindent\textbf{Financial support.} 
L.B.\ gratefully acknowledges funding from the Eu\-ro\-pean Union’s Horizon 2020 research and innovation programme under the Marie Sk{\textl}o\-dows\-ka-Curie Grant Agreement No.~754411. 
R.S.\ was supported by the European Research Council (ERC) under the European Union’s Horizon 2020 research and innovation programme (Grant Agreement No.~694227).\\

\renewcommand{\bibname}{References}
\bibliographystyle{abbrv}
    \bibliography{bib_file}
\end{document}